\documentclass{article}
\usepackage[top=1.5in,bottom=1.5in,left=1.5in,right=1.5in]{geometry}
\usepackage{graphicx,mathtools, mathrsfs,amsmath,amsthm,braket,mathtools,amsfonts,amssymb,mathdots, xfrac,
xparse, listings,color, float, cancel, verbatim, hyperref, cleveref,xcolor,mdframed, accents,
enumerate, wrapfig, adjustbox, multirow, bbm, scalerel, caption, subcaption, tablefootnote}
\usepackage{authblk}  %fancy author list

\usepackage{tikz-cd}
\usetikzlibrary{backgrounds}
\usetikzlibrary{calc}
\usetikzlibrary{positioning}
\usetikzlibrary{decorations.markings}
\usetikzlibrary{shapes.geometric}
\usetikzlibrary{cd}
\usetikzlibrary{perspective}

\newcommand{\ang}[1]{\langle #1 \rangle}
\DeclareMathOperator{\MIP}{MIP}

\DeclareMathOperator{\NP}{NP}
\DeclareMathOperator{\cP}{P}
\DeclareMathOperator{\PCP}{PCP}
\DeclareMathOperator{\poly}{poly}

\DeclareMathOperator{\RE}{RE}
\DeclareMathOperator{\CS}{CS}

\DeclareMathOperator{\SynAlg}{SynAlg}
\DeclareMathOperator{\supp}{supp}
\DeclareMathOperator{\df}{def}

\DeclareMathOperator{\NEXP}{NEXP}
\DeclareMathOperator{\EXP}{EXP}

\DeclareMathOperator{\BCS}{BCS}

\newcommand{\AND}{\wedge}

\numberwithin{equation}{section}
\theoremstyle{plain}
\newtheorem{theorem}{Theorem}[section]

\newtheorem{lemma}[theorem]{Lemma}
\newtheorem{proposition}[theorem]{Proposition}
\newtheorem{prop}[theorem]{Proposition}
\newtheorem{example}[theorem]{Example}
\newtheorem{definition}[theorem]{Definition}
\newtheorem{defn}[theorem]{Definition}

\newtheorem{corollary}[theorem]{Corollary}

\newcommand{\mcA}{\mathcal{A}}
\newcommand{\mcB}{\mathcal{B}}
\newcommand{\mcC}{\mathcal{C}}

\newcommand{\mcH}{\mathcal{H}}

\newcommand{\mcK}{\mathcal{K}}
\newcommand{\mcL}{\mathcal{L}}
\newcommand{\mcM}{\mathcal{M}}

\newcommand{\mcU}{\mathcal{U}}

\newcommand{\mbC}{\mathbb{C}}

\newcommand{\mbZ}{\mathbb{Z}}

\newcommand{\wtd}{\widetilde}

\newcommand{\C}{\mathbb{C}}
\newcommand{\Z}{\mathbb{Z}}
\newcommand{\N}{\mathbb{N}}

\newcommand{\eps}{\varepsilon}

\usepackage{bbm}

\usepackage{qstyle}

\title{RE-completeness of entangled constraint satisfaction problems}

\author[1,2]{Eric Culf}
\author[1,3]{Kieran Mastel}

\affil[1]{\small Institute for Quantum Computing, University of Waterloo, Canada}
\affil[2]{\small Department of Applied Mathematics, University of Waterloo, Canada}
\affil[3]{\small Department of Pure Mathematics, University of Waterloo, Canada}

\begin{document}

\maketitle

\begin{abstract}
    Constraint satisfaction problems (CSPs) are a natural class of decision problems where one must decide whether there is an assignment to variables that satisfies a given formula. Schaefer's dichotomy theorem, and its extension to all alphabets due to Bulatov and Zhuk, shows that CSP languages are either efficiently decidable, or $\NP$-complete. It is possible to extend CSP languages to quantum assignments using the formalism of nonlocal games. Due to the equality of complexity classes $\MIP^\ast=\RE$, general succinctly-presented entangled CSPs are $\RE$-complete. In this work, we show that a wide range of $\NP$-complete CSPs become $\RE$-complete in this setting, including all boolean CSPs, such as 3SAT, as well as $3$-colouring. This also implies that these CSP languages remain undecidable even when not succinctly presented.

    To show this, we work in the weighted algebra framework introduced by Mastel and Slofstra, where synchronous strategies for a nonlocal game are represented by tracial states on an algebra. Along the way, we improve the subdivision technique in order to be able to separate constraints in the CSP while preserving constant soundness, construct commutativity gadgets for all boolean CSPs, and show a variety of relations between the different ways of presenting CSPs as games.
\end{abstract}

\newpage

\section{Introduction}

In a multiprover interactive proof system, multiple provers (who are unable to communicate with each other) try to convince a polynomial-time verifier that a
string $x$ belongs to a language $\mcL$.  Babai, Fortnow, and Lund proved that the class $\MIP$ of classical multiprover interactive proof systems is equal to $\NEXP$ \cite{BFL91}. The proof systems used in \cite{BFL91} are very efficient, and require only two provers and one round of communication.

Since the provers in an MIP protocol are not allowed to communicate, it is
natural to ask what happens if they are allowed to share entanglement. This
leads to the complexity class $\MIP^*$, first introduced by Cleve, H\o{}yer,
Toner, and Watrous \cite{cleve2010consequences}. Entanglement allows the
provers to utilize correlations that cannot be sampled by classical provers \cite{EPR35,Bell64},
but the improved ability of the provers also allows the verifier to set harder
tasks. As a result, figuring out the power of $\MIP^*$ has been difficult, and
there have been successive lower bounds in \cite{KKM+11, IKM09, IV12, Vid16,
Vid20eratum, Ji16, NV18b, Ji17, NV18a, FJVY19}. Finally, in a landmark work, Ji, Natarajan, Vidick, Wright, and Yuen showed that
$\MIP^*=\RE$, the class of languages equivalent to the halting problem
\cite{ji2022mipre}.

A one-round $\MIP$ or
$\MIP^*$ proof system is equivalent to a family of nonlocal games, in which the
provers (now also called players) are given questions and return answers to a
verifier (now also called a referee), who decides whether to accept (in which
case the players are said to win) or reject (the players lose). As a result of the proof that $\MIP = \NEXP$ in \cite{BFL91}, $\MIP$ is equivalent to the class of $\CS-\MIP$ proof systems, which are two-prover one-round proof systems in which the nonlocal games are constraint system (CS) games. In a $\CS$ game, the players try to convince the verifier that a given $\CS$ is satisfiable. $\CS$ games encompass many well-known examples of nonlocal games such as the Mermin-Peres magic square \cite{mermin90simple, PERES1990107}, and graph colouring games~\cite{GW02}. Importantly, in a $\CS-\MIP$ proof system, the CS games are succinctly presented. This means that, for a given instance, there may be exponentially many questions to sample or answers to verify (in the instance size), but these operations are implemented efficiently. This can increase the complexity of the games significantly. Recall that, for any family of constraints $\Gamma$ over an alphabet $\Sigma$, the CSP language $\CSP(\Gamma)_{1,1}$ consists of all constraint systems that can be expressed as the conjunction of constraints from $\Gamma$, where the yes instances are those for which all of the constraints can be satisfied, and the no instances are those where at least one constraint must be unsatisfied. Due to the CSP dichotomy theorem, it is well-understood which CSP languages are $\NP$-complete~\cite{Bul17,Zhu17}. The succinct version of this language is the promise problem $\SuccinctCSP(\Gamma)_{1,s}$ consisting of efficiently-sampleable CSs with constraints from~$\Gamma$, where the yes instances are those for which all the constraints can be satisfied, and the no instances are those where there is a probability less than $s$ of sampling a satisfied constraint for any assignment. If $\CSP(\Gamma)_{1,1}$ is $\NP$-complete, then there is a constant $s\in [0,1)$, such that $\SuccinctCSP(\Gamma)_{1,s}$ is $\NEXP$-complete. Thus, the class of $\CS$ games corresponding to $\SuccinctCSP(\Gamma)_{1,s}$ is complete for $\MIP$.

\paragraph{Summary of results} It is natural to ask if an analogous result holds when the provers are allowed entanglement. In this paper, we provide a partial answer to this question. In general, operator assignments to CSPs are non-commuting, but in order to guarantee the validity of some reductions between CSPs, we need a way to guarantee commutativity (or near-commutativity) between variables. A canonical way do that is by means of empty constraints, which contain variables but impose no relations between them, except that they can be measured simultaneously. However, we have no guarantee that a given set of constraints $\Gamma$ includes an empty constraint. To remedy this, we construct commutativity gadgets: subsystems of constraints that behave like an empty constraint when restricted to two of the variables. 

To characterise the commutativity gadgets we are able to construct, we make use of a property of a constraint we call two-variable falsifiability (TVF) where, for any pair of variables, there is an assignment to that pair that makes the constraint false no matter what value the other variables are assigned. If all constraints in $\Gamma$ are TVF, this precludes the construction of generic commutativity gadgets built from one constraint to replace empty constraints. However, many important CSPs, such as $3$SAT are not TVF, so we can replace empty constraints by such one-constraint gadgets there. Note in particular that any CSP augmented by an empty constraint is non-TVF, and that this does not change its classical complexity.

For TVF constraints, the situation is more complicated, but in the boolean case, we are able to construct a generic but more complicated commutativity gadget. The basic example of an $\NP$-complete TVF constraint satisfaction problem is 1-in-3-SAT, generated by the three-variable boolean constraint where only one of the variables can be $1$, $C=\{(1,0,0),(0,1,0),(0,0,1)\}$. This has a commutativity gadget constructed by connecting three copies of $C$ in a triangular arrangement (see \Cref{fig:flux-capacitor}), first studied by Ji~\cite{Ji13}. For the remaining $\NP$-complete boolean TVF CSPs, we show that there is a structure similar to $C$ hidden within them, and hence that an analogous commutativity gadget may be constructed.

The final case where we are able to construct a commutativity gadget is for graph $3$-colouring. Here, we make use of a triangular prism gadget (see \Cref{fig:a-prism}), also introduced by Ji~\cite{Ji13}. The work of Ji shows that the gadget guarantees commutativity in the case of perfect completeness; for our purposes, we extend this to the case of imperfect completeness, showing soundness of the gadget. 

In full, we show that if $\CSP(\Gamma)_{1,1}$ is $\NP$-complete, and $\Gamma$ is boolean, non-TVF, or $3$-colouring, then there exists a constant $s\in[0,1)$ such that the entangled CSP language $\SuccinctCSP(\Gamma)^\ast_{1,s}$ is $\RE$-complete.

As a direct consequence, we also find that the non-succinct version of the language, $\CSP(\Gamma)^\ast_{1,s}$, is also undecidable, although it may not be $\RE$-complete with polynomial-time reductions. Completeness in $\NEXP$ is with respect to polynomial-time Karp reductions; in fact, with exponential-time reductions, $\NP$-complete problems become complete for $\NEXP$. In the same way, if $\SuccinctCSP(\Gamma)_{1,s}^\ast$ is $\RE$-complete with a polynomial-time reduction, then $\CSP(\Gamma)_{1,s}^\ast$ is $\RE$-complete with an exponential-time reduction. In particular, there is a computable function that reduces the halting problem to $\CSP(\Gamma)_{1,s}^\ast$, so it must be undecidable.

Interactive proof systems allow zero knowledge protocols, in which the prover demonstrates that $x \in \mcL$ without revealing any other information to the verifier. Mastel and Slofstra proved that every language in $\MIP^*$ admits a two-prover one-round perfect zero knowledge proof system with polynomial length questions and answers \cite{MS24}. By improving the techniques used in \cite{MS24}, we remove the need for parallel repetition at the end of the argument. Consequently, we show that every language in $\MIP^*$ admits a two-prover one-round perfect zero knowledge proof system with polynomial length questions and constant length answers.

\paragraph{Proof techniques}

For the proof of our main result, we begin with the output of the $\MIP^*=\RE$ theorem rather than encoding an arbitrary $\MIP^*$ protocol. The proof that $\MIP^*=\RE$ in \cite{ji2022mipre} is very involved, but has as an important consequence that only two-prover one-round proof systems are required to attain the full complexity of $\MIP^\ast$. Dong, Fu, Natarajan, Qin, Xu, and Yao show in \cite{DFNQXY23} that these proof systems can be reduced in size to have polynomial-length questions and constant-length answers. In both
\cite{ji2022mipre} and \cite{DFNQXY23}, the games are synchronous,
meaning that if the players receive the same question then they must reply with
the same answer, and admit oracularizable optimal strategies in the case of perfect completeness, meaning that the two players' measurement operators for any pair of questions asked at the same time commute. It was observed in \cite{MS24} that one-round $\MIP^*$ proof systems in which the games are
synchronous and oracularizable are equivalent to the class of $\BCS$-$\MIP^*$
proof systems, which are one-round two-prover proof systems in which the
nonlocal games are boolean constraint system (BCS) games, that is, $\CS$ games with boolean constraint systems. Since the games in \cite{DFNQXY23} have constant answer size, so do the corresponding BCS games. We prove the $\RE$-hardness of $\NP$-complete CS protocols with entanglement by showing that the classical reduction from the BCS form of the protocol from \cite{DFNQXY23} to a CS protocol is sound against quantum provers.

It is difficult to determine if a classical
transformation of constraint systems remains
sound (meaning that it preserves the soundness of protocols) in the quantum setting. For example, a key part of the $\MIP^*=\RE$ theorem is the construction of a PCP of proximity which is quantum-sound. On the other hand,
there are some transformations that lift fairly easily to the quantum setting. Mastel and Slofstra identify two such classes of transformations between boolean constraint systems: ``classical transformations,''
which are applied constraint by constraint; and ``context subdivision
transformations,'' in which each constraint is split into a number of
subclauses \cite{MS24}. A key result of \cite{MS24} is a systematic analysis of the quantum soundness of these transformations. Reductions between nonlocal games are difficult to reason about since it is necessary to keep track of how
strategies for one game map to strategies for the other game. An advantage of working with constraint systems in the classical setting is that one can work with assignments to the variables and consider what fraction of the constraints they satisfy, rather than work with strategies and winning probabilities. In the quantum setting, we cannot work with assignments in this way, because strategies involve observables that do not necessarily commute. However, 
a similar conceptual simplification is achieved in \cite{MS24} by replacing assignments with approximate representations of the BCS algebra of the constraint system. This algebra is the same as the synchronous algebra of the
BCS game introduced in \cite{helton2017algebras,Kim_2018}. However, the transformations in \cite{MS24} have polynomial soundness dropoff in some cases and use parallel repetition to recover constant soundness. Parallel repetition does not preserve many classes of CS games (for example, graph colouring games), so we require our transformations to only have a constant soundness dropoff. Furthermore, the techniques in \cite{MS24} are specific to BCS games and cannot deal with transformations between strategies for CS games over larger alphabets.

We generalize the techniques of \cite{MS24} to constraint systems on larger alphabets. In our more general model, strategies correspond to representations of a $\CS$ algebra, which generalizes the $\BCS$ algebra. We introduce multiple models of CS algebra that correspond to the constraint-constraint, constraint-variable, and the 2-CS games. In a constraint-constraint game, both players are asked for assignments to constraints from the constraint system, and they win if the assignments agree on all overlapping variables. This style of game was studied in, \emph{e.g.}, \cite{Arkh12,PS23,MS24}. In a constraint-variable game, one player is asked for an assignment to a constraint, as before, while the other is asked for an assignment to a single variable from that constraint. The winning condition is that the two assignments to the variable are the same. This style of game was studied in, \emph{e.g.}, \cite{CM14,Ji13}. Finally, there is a third type of game for $2$-CSs --- where each context has only two variables. Here, each player is asked for an assignment to a single variable, and they win if the assignments satisfy the corresponding constraint. This style of game was studied in, \emph{e.g.}, \cite{GW02,Har24,culf2024}. This game corresponds to a CS algebra which we call the assignment algebra. It is also possible to study this algebra for other CSs, but it is in general unclear whether there is a corresponding game. Nevertheless, the assignment algebra is conceptually useful because each variable is associated to a unique quantum assignment, unlike in the other models where there is a different assignment associated to each constraint. It is sometimes useful to pass between different models, so to prove our result, we prove and exploit relations between these different models. Reductions between CS games can be expressed as homomorphisms between CS algebras, which are easier to describe than mappings between strategies. 

For soundness arguments, we work with near-perfect strategies, which correspond to approximate representations of the CS algebra, generalizing techniques  developed for the BCS setting in \cite{Pad22, PS23}. To keep track of the soundness of reductions between games, we use the weighted algebra formalism introduced in \cite{MS24}. A weighted algebra $\mc{A}$ is an algebra with an associated finitely-supported weight function $\mu$, which indicates elements whose norm should be small under representations of the weighted algebra. This gives a way to approximately capture relations. For example, in the constraint-constraint algebra, the relations within the constraints must be exactly satisfied, while the agreement between variables in different constraints need only be approximately satisfied. Given a tracial state $\tau$ on $\mc{A}$, we can measure how well it satisfies the approximate relations on $\mc{A}$ by computing its defect, denoted $\defect(\tau)$, which is the trace of the sum of squares of the elements of $\mc{A}$ weighted by $\mu$; a defect of $0$ corresponds the relations being perfectly satisfied. For CS games, synchronous strategies are in a one-to-one correspondance with tracial states on the associated CS algebra, and the winning probability is $1-\defect(\tau)$. To study how the winning probability changes as we modify the constraint system, we look at how the defect changes due to mappings between the weighted algebras. Given a trace $\tau$ on a weighted algebra $\mc{B}$, we want to be able to construct a trace $\tau'$ on $\mc{A}$ such that $\defect(\tau')\leq C\defect(\tau)$ for some constant $C$. Often, we can do this in a trace-independent way by means of a $\ast$-homomorphism $\mc{A}\rightarrow\mc{B}$, called a $C$-homomorphism. In some cases, the definition of the mapping will rely on the trace; this is similar to the notion of weak $\ast$-equivalence defined in~\cite{Har24}. In either case, the existence of such a mapping tells us that if the synchronous value of the game corresponding to $\mc{A}$ is $<s$, then the defect of any tracial state is $>1-s$. In particular, this holds for states constructed from tracial states on $\mc{B}$, so the defect of any tracial state on $\mc{B}$ is $>(1-s)/C$. Hence, the winning probability of the game associated to $\mc{B}$ is $<1-(1-s)/C$. As such, the weighted algebra formalism allows us to transfer soundness guarantees from one CS game to another via maps between the corresponding algebras.

An essential component of the  proof of our main result is an improved version of the context subdivision transformation. In particular, our improved subdivision allows us to preserve constant soundness when our answer size before subdivision is constant, without using parallel repetition. Subdivision is a BCS transformation that splits up the constraints into multiple subclauses. The subdivision result in \cite{MS24} assumes the question distribution is maximized on the diagonal, and obtains a soundness drop-off that is polynomial in the number of constraints in the BCS. In our proof, we impose a stronger assumption on the question distribution, namely that it is $C$-diagonally dominant, and prove a version of the subdivision result that removes the dependence on the number of constraints in exchange for a dependence on the maximum constraint size.

A key requirement in order to be able to apply the context subdivision transformation is that each pair of variables must appear together in at least one constraint. This is important in order to be able to recover the commutation of the variables from the original constraint in an approximate setting. As mentioned above, the naive way to guarantee this is to introduce empty constraints; but, this does not preserve the set of constraints of an arbitrary CSP. We can replace empty constraints by commutativity gadgets, which we construct in three different ways, depending on the CSP. For non-TVF constraint systems, we can directly construct simple one-constraint gadgets, where a pair of variables behaves as the variables of an empty constraint. For $3$-colouring, we show soundness of a gadget introduced by Ji~\cite{Ji13}. The most significant innovation is in the context of TVF constraints. Here, no pair of variables behaves as an empty constraint. However, boolean TVF constraints are quite heavily constrained by the TVF condition. We investigate the combinatorial structure of these constraints to build a generic commutativity gadget for all such constraints. An important tool we develop to study TVF constraints is the notion of a TVF graph, which is a graph where the vertices are variables of the TVF constraint and edges are the labelled by the pair of assignments prohibited by the TVF condition. Using this graph structure, we can reduce to constraints that we call incompressible, where the variables are in a sense independent. In the incompressible case, we can uncover an upper-triangular structure in the constraint. From here, we study two cases, depending on which edge labellings appear in the TVF graph, or equivalently whether or not it is possible to construct a pair of variables that must be mutual negations. In both cases, we show that an incompressible TVF constraint that does not satisfy the majority polymorphism (a necessary condition for $\NP$-completeness) must be able to simulate $C=\{(1,0,0),(0,1,0),(0,0,1)\}$, or a similar constraint with some of its variables negated, by filling the constraint in a particular way with variables, constants, and negations of variables (if allowed). Further, we show that this simulation construction is quantum-sound, and therefore an analogous commutativity gadget as found for $C$ in~\cite{Ji13} can be constructed with any $\NP$-complete set of TVF constraints.

\begin{figure}
    \centering
    \begin{subfigure}{0.4\textwidth}
    \centering
    \begin{tikzpicture}[scale=0.8]
		\draw[line width=0.05cm] (0,0) -- (0,3) node[left,pos=0.5]{$\in\EXP$};
		\draw[line width=0.05cm,dotted] (0,3) -- (0,4.25);
		\draw[line width=0.05cm,red] (0,4.25) -- (0,5)node[left,pos=0.5]{$\RE$} node[right]{\color{black}$1$};
		
		\fill (0,3) circle (0.1cm) node[right] {$\approx 0.864$};
		\fill[red] (0,4.25) circle (0.1cm) node[right] {$<1$};
	\end{tikzpicture}
    \caption{Succinct entangled $3$-colouring}
    \end{subfigure}
    \begin{subfigure}{0.4\textwidth}
    \centering
    \begin{tikzpicture}[scale=0.8]
		\draw[line width=0.05cm] (0,0) -- (0,3) node[left,pos=0.5]{$\in\cP$};
		\draw[line width=0.05cm,dotted] (0,3) -- (0,4.25);
		\draw[line width=0.05cm,red] (0,4.25) -- (0,5)node[left,pos=0.5]{Undecidable} node[right]{\color{black}$1$};
		
		\fill (0,3) circle (0.1cm) node[right] {$\approx 0.864$};
		\fill[red] (0,4.25) circle (0.1cm) node[right] {$<1$};
	\end{tikzpicture}
    \caption{Entangled $3$-colouring}
    \end{subfigure}
    \caption{Transitions in complexity based on soundness parameter for entangled $3$-colouring}
    \label{fig:transition}
\end{figure}

\paragraph{Related work and open problems} An interesting consequence of our main result is that there exists a constant soundness parameter $s \in [0,1)$ such that the class of succinct graph $3$-colouring games is $\RE$-complete with entanglement. On the other hand, Culf, Mousavi, and Spirig show that entangled graph $3$-colouring with $s = 0.864$ is in $\cP$, and thus the succinct version with this soundness is in $\EXP$ \cite{culf2024}. This situation is illustrated in~\Cref{fig:transition}. The question of what the complexity is when $s$ is between these two values is an interesting open problem. This connects to a large body of work on inapproximability for classical games. For example, it is known that 3SAT is in $\cP$ for $s<7/8$, but $\NP$-hard for $s=7/8+\varepsilon$ for any $\varepsilon>0$~\cite{Has01}. For many CSPs, optimal inapproximability properties rely on the unique games conjecture~\cite{Kho02,Rag08}. Though the unique games conjecture in its standard form is known to be false for entangled games~\cite{KRT10}, there is recent work that explores how it might be extended to this setting~\cite{MS24b}.

The special case of $3$-colouring was studied independently by Harris~\cite{Har24,Har24b}. This work takes a different approach, connecting the values of synchronous games directly to instances of $3$-colouring, rather than by passing through general constraint-system games and reducing to $3$-colouring via $\NP$-hardness. The zero completeness-soundness gap case was first studied in~\cite{Har24} and later generalised to a nonzero gap in~\cite{Har24b}.

There remain other interesting open problems in this direction. First, not all CSPs are captured by our analysis: we do not know yet the complexity of non-boolean TVF CSPs other than $3$-colouring, even for those CSPs that are $\NP$-complete. Similarly, our work does not say anything about the complexity of entangled proof systems that are classically easy. Prior work by Paddock and Slofstra~\cite{PS23} has shown that all the classically-easy boolean CSPs except for $\tsf{LIN}$, the class of linear systems over $\Z_2$, remain easy with entanglement; but, to the best of our knowledge, the complexity of entangled $\tsf{LIN}$ and of any classically-easy CSPs over larger domains remains unknown. 

Next, we study CSPs that are succinctly presented, corresponding to nonlocal games with exponentially-many possible questions. Our results only apply via exponential-time reductions to CSPs that are not succinctly presented, which correspond to nonlocal games with polynomially-many questions. The complexity of such games under polynomial-time reductions is not well understood and relates to the games version of the quantum PCP conjecture~\cite{NN24}. 

Finally, note that our reductions are designed to work in the case of completeness $1$, but fail in the case where the winning probability in the yes instance may be less than $1$. Classically, some problems, such as Max-Cut and $\tsf{LIN}$ are easy with completeness $1$, but reveal hidden complexity when considered with imperfect completeness \cite{Has01}. Generalising the weighted algebra formalism to the setting of imperfect completeness would allow us to prove quantum hardness for such problems.

\paragraph{Outline} The structure of the paper is as follows. First, in~\Cref{sec:prelims}, we introduce notation and concepts we will use throughout the paper. Next, in~\Cref{sec:cs}, we define the constraint system algebras we use, and study the relationships between them. In~\Cref{sec:csp}, we discuss classical and quantum CSPs, and state the main theorem formally as~\Cref{thm:main-theorem}. We split the proof of the main theorem for the cases of non-TVF CSPs, boolean TVF CSPs, and $3$-colouring between \Cref{sec:non-TVF}, \ref{sec:TVF}, and \ref{sec:2-csps}, respectively. We also split the proof of a corollary of the main theorem for the assignment algebra between \Cref{sec:TVF,sec:2-csps}. Finally, we show a corollary of the main theorem for the constraint-constraint algebra in \Cref{sec:cv-to-cc}. 

\paragraph{Acknowledgements}
We thank William Slofstra, Richard Cleve, and Alex Meiburg for helpful conversations. KM is supported by a CGS D scholarship from NSERC. EC is supported by a CGS D scholarship from NSERC.

\section{Preliminaries and notation}\label{sec:prelims}

\subsection{General notation}

For $n\in\N$, we write the set $[n]=\{1,2,\ldots,n\}$. We assume the logarithm $\log$ is base $2$ unless otherwise specified.

We identify $\Z_k$ with the subset $\{0,\ldots,k-1\}$ of $\N$, and denote the primitive $k$-th root of unity $\omega_k=e^{2\pi i/k}$, dropping the subscript if clear from context. Note that in \cite{MS24}, $\Z_2$ was treated multiplicatively as $\{+1,-1\}$; it is more natural to use the more standard additive form in our setting.

For a graph $G=(V,E)$ and $U\subseteq V$, write $G|_U$ for the subgraph on vertices in $U$ and $G\backslash U$ for the subgraph on vertices in $V\backslash U$.

We deal only with probability distributions on finite sets. Hence, we present any probability distribution $\pi$ on a set $A$ by a function $\pi:A\rightarrow[0,1]$ such that $\sum_{a\in A}\pi(a)=1$. We say a probability distribution $\pi$ on $A\times A$ is \textbf{symmetric} if $\pi(a,b)=\pi(b,a)$ for all $a,b\in A$. Following \cite{marrakchi2023synchronous}, we say that a probability distribution $\pi$ on $A\times A$
is \textbf{$C$-diagonally dominant} if $\pi(a,a) \geq C \sum_{b \in A}
\pi(a,b)$ and $\pi(a,a) \geq C \sum_{b \in A} \pi(b,a)$ for all $a \in A$.

Write $\mbb{u}_n$ for the uniform distribution on $[n]$, \emph{i.e.} $\mbb{u}_n(i)=\frac{1}{n}$ for all $i\in[n]$.

\subsection{Complexity theory}

A two-player \textbf{nonlocal game} $\ttt{G} = (I,\{O_i\}_{i\in I},\pi,V)$ consists of a finite set of questions $I$, a collection of finite answer sets $\{O_i\}_{i\in I}$, a probability distribution $\pi$ on $I\times I$, and a family of functions $V(\cdot,\cdot|i,j):O_i\times O_j\rightarrow \{0,1\}$ for $(i,j)\in I\times I$. In the game, the players (commonly called Alice and Bob) receive questions from $i$ and $j$ from $I$ with probability $\pi(i,j)$, and reply with answers $a\in O_i$ and $b\in O_j$ respectively. They win if $V(a,b|i,j)=1$ and lose otherwise. For the sake of convenience, we have assumed that the players have the same question and answer sets. This assumption can be made without loss of generality. We often think of the question and answer sets as subsets of $\{0,1\}^n$ and $\{0,1\}^{m_i}$ for $i\in I$ respectively. In this case we say that the questions have length $n$ and the answers have length $\max_{i\in I}m_i$. 

A \textbf{correlation} for a set of inputs and outputs $(I,\{O_i\}_{i\in I})$ is a family $p$ of probability distributions $p(\cdot,\cdot|i,j)$ for all $(i,j)\in I\times I$. Correlations describe the players' behaviour in a nonlocal game. The probability $p(a,b|i,j)$ is interpreted as the probability that the players answer $(a,b)$ on questions $(i,j)$. A correlation $p$ is \textbf{classical} if there is a set $\Lambda$ with a probability measure $\mu$, and if for each $\lambda \in \Lambda$ there are functions $f_1^{\lambda},f_2^{\lambda}:I\to \cup_i O_i$ such that $f_1^{\lambda}(i),f_2^{\lambda}(i)\in O_i$ for all $i\in I$, and $p(a,b|i,j) = \text{Pr}_{\lambda \sim \mu}(f_1^{\lambda}(i) = a \AND f_2^{\lambda}(j) = b)$ for all $i,j \in I$, $a \in O_i$, $b \in O_j$. The collection $(\Lambda,\mu,\{f_1^{\lambda}\},\{f_2^{\lambda}\})$ is called a \textbf{classical strategy}. This captures the notion that a strategy for classical unentangled provers consists of some shared randomness that is independent of the verifier's questions, and player strategies that are deterministic for a given state $\lambda$ of the shared randomness. A correlation $p$ is \textbf{quantum} if there are
\begin{itemize}
    \item finite dimensional Hilbert spaces $H_A$ and $H_B$,
    \item a projective measurement $\{M_a^i\}_{a\in O_i}$ on $H_A$ for every $i \in I$,
    \item a projective measurement $\{N_a^i\}_{a \in O_i}$ on $H_B$ for every $i \in I$, and 
    \item a state $\ket{v}\in H_A\otimes H_B$
\end{itemize}
such that $p(a,b|i,j) =\bra{v}M^i_a\otimes N_b^j\ket{v}$ for all $i,j \in I$, $a \in O_i$, and $b \in O_j$. The collection $(H_A, H_B, \{M_a^i\},\{N_a^j\}, \ket{v})$ is called a \textbf{quantum strategy}. A correlation is \textbf{commuting operator} if there exists
\begin{itemize}
    \item a Hilbert space $H$,
    \item projective measurements $\{M_a^i\}_{a\in O_i}$ and $\{N_a^i\}_{a \in O_i}$ on $H$ for every $i \in I$, and
    \item a state $\ket{v}\in H$
\end{itemize}
such that $M_a^iN_b^j = N_b^jM_a^i$ and $p(a,b|i,j) =\bra{v}M^i_a N_b^j\ket{v}$ for all $i,j \in I$, $a \in O_i$, and $b \in O_j$. The collection $(H,\{M_a^i\},\{N_a^j\},\ket{v})$ is called a \textbf{commuting operator strategy}. The set of classical correlations for a set of inputs and outputs $(I,\{O_i\})$ is denoted $C_c(I,\{O_i\})$. Similarly, the set of quantum and commuting operator correlations are denoted $C_q(I,\{O_i\})$ and $C_{qc}(I,\{O_i\})$, respectively. If the inputs and outputs are clear from context we denote the sets by $C_c$, $C_q$, and $C_{qc}$, respectively. It follows from the definitions that $C_c\subseteq C_q\subseteq C_{qc}$.

The \textbf{winning probability} of a correlation $p$ in a nonlocal game $\ttt{G} = (I,\{O_i\},\pi,V)$ is 
\begin{equation*}
    \mfk{w}(\ttt{G};p):= \sum_{i,j\in I}\sum_{a\in O_i, b\in O_j}\pi(i,j)V(a,b|i,j)p(a,b|i,j).
\end{equation*}
Given a strategy $S$ for $\ttt{G}$ with corresponding correlation $p$, the winning probability is also denoted $\mfk{w}(\ttt{G};S) = \mfk{w}(\ttt{G};p)$. The \textbf{classical value} of $\ttt{G}$ is
\begin{equation*}
    \mfk{w}_c(\ttt{G}) := \sup_{p\in C_c}\mfk{w}(\ttt{G};p).
\end{equation*}
The \textbf{quantum value} is 
\begin{equation*}
    \mfk{w}_q(\ttt{G}) := \sup_{p\in C_q}\mfk{w}(\ttt{G};p).
\end{equation*}
The \textbf{commuting operator value} is 
\begin{equation*}
    \mfk{w}_{qc}(\ttt{G}) := \sup_{p\in C_{qc}}\mfk{w}(\ttt{G};p).
\end{equation*}
A correlation $p$ is \textbf{perfect} for $\ttt{G}$ if $\mfk{w}(\ttt{G},p) = 1$, and $\eps$\textbf{-perfect} if $\mfk{w}(\ttt{G},p) \geq 1-\eps$. A strategy is $\eps$-perfect if its corresponding correlation is $\eps$-perfect. Since $C_{c}$ (resp. $C_{qc}$) is closed and compact, $\ttt{G}$ has a perfect classical (resp. commuting operator) strategy if and only if $\mfk{w}_{c}(\ttt{G}) = 1$ (resp. $\mfk{w}_{qc}(\ttt{G}) = 1$). The set of quantum correlations $C_q$ is not necessarily closed. There are games for which $\mfk{w}_{q}(\ttt{G})=1$, but which do not have a perfect quantum correlation. A correlation $p$ is \textbf{quantum approximable} if it is an element of the closure $C_{qa} = \overline{C_q}$, and a game has a perfect quantum approximable strategy if and only if $\mfk{w}_q(\ttt{G}) = 1$. 

A nonlocal game $\ttt{G} = (I,\{O_i\},\pi,V)$ is \textbf{synchronous} if $V(a,b|i,i) = 0$ for all $i \in I$ and $a \neq b \in O_i$. A correlation $p$ is \textbf{synchronous} if $p(a,b|i,i) = 0$ for all $i \in I$ and $a \neq b \in O_i$. The set of synchronous classical, quantum, and commuting operator correlations are denoted $C_c^s$, $C_q^s$, and $C_{qc}^s$, respectively. A correlation is in $C^s_{qc}$ (resp. $C^s_q$) if and only if there is 

\begin{itemize}
    \item a Hilbert space $H$ (resp. finite dimensional Hilbert space $H$),
    \item a projective measurement $\{M_a^i\}_{a\in O_i}$ on $H$ for all $i \in I$, and 
    \item a state $\ket{v}\in H$
\end{itemize}
such that $\ket{v}$ is tracial, in the sense that $\bra{v}\alpha\beta\ket{v} = \bra{v}\beta\alpha\ket{v}$ for all $\alpha$ and $\beta$ in the $\ast$-algebra generated by the operators $M_a^i$, $i\in I$, $a \in O_i$, and $p(a,b|i,j) = \bra{v}M_a^iM_b^j\ket{v}$ for all $i,j\in I$, $a\in O_i$, and $b \in O_j$. The collection $(H,\{M_a^i\},\ket{v})$ is called a \textbf{synchronous commuting operator strategy}. If the Hilbert space $H$ is finite dimensional, then the collection $(H,\{M_a^i\},\ket{v})$ is called a \textbf{synchronous quantum strategy}. The synchronous quantum and commuting operator values $\mfk{w}^s_q(\ttt{G})$ and $\mfk{w}_{qc}^s(\ttt{G})$ of a game $\ttt{G}$ are defined similarly to $\mfk{w}_q(\ttt{G})$ and $\mfk{w}_{qc}(\ttt{G})$ by replacing $C_q$ and $C_{qc}$ with $C^s_q$ and $C^s_{qc}$, respectively. A synchronous strategy is called \textbf{oracularizable} if $M_a^iM_b^j = M_b^jM_a^i$ for all $i,j \in I$, $a \in O_i$, and $b \in O_j$ with $\pi(i,j)>0$. 

Every quantum correlation that is close to being synchronous, in the sense that $p(a,b|i,i) \approx 0$ for all $i \in I$ and $a\neq b \in O_i$, is close to a synchronous quantum correlation \cite{Vidick_2022}. This is also true for commuting operator correlations \cite{lin2024tracialembeddable}. As a result, the synchronous quantum and commuting operator values of a synchronous game are polynomially related to the non-synchronous quantum and commuting operator values. We use a version of this statement from \cite{marrakchi2023synchronous}:
\begin{theorem}[\cite{marrakchi2023synchronous}]\label{thm:synchrounding}
    Suppose $\ttt{G}$ is a synchronous game with a $C$-diagonally dominant
    question distribution. If $\mfk{w}_q(\ttt{G})$ (resp. $\mfk{w}_{qc}(\ttt{G})$) is
    $\geq 1-\eps$, then $\mfk{w}_q^s(\ttt{G})$ (resp. $\mfk{w}_{qc}^s(\ttt{G})$) is
    $\geq 1 - O((\eps/C)^{1/4})$.
\end{theorem}

Note that that any nonlocal game can be transformed into a synchronous game satisfying the above condition by symmetrising the probability distribution and adding consistency checks. This will only perturb the synchronous quantum and commuting operator values slightly, but it may change the general quantum and commuting operator values significantly.

A \textbf{two-prover one-round $\MIP$ protocol} is a family of nonlocal games
$\ttt{G}_x = (I_x,\{O_{xi}\}_{i \in I_x}, \pi_x, V_x)$ for $x \in \{0,1\}^*$,
along with a probabilistic Turing machine $S$ and another Turing machine $V$,
such that 
\begin{itemize}
    \item for all $x \in \{0,1\}^*$ and $i \in I_x$, there are integers 
        $n_x$ and $m_{xi}$ such that $I_x = \{0,1\}^{n_x}$ and $O_{xi}
            = \{0,1\}^{m_{xi}}$, 

    \item on input $x$, the Turing machine $S$ outputs $(i,j) \in I \times I$
        with probability $\pi_x(i,j)$, and  

    \item on input $(x,a,b,i,j)$, the Turing machine $V$ outputs $V_x(a,b|i,j)$. 
\end{itemize}
Let $c, s : \{0,1\}^* \to [0,1]$ be computable functions with $c(x) > s(x)$
for all $x \in \{0,1\}^*$. A language $\mcL \subset \{0,1\}^*$ belongs to $\MIP(2,1,c,s)$ if
there is a MIP protocol $(\{\ttt{G}_x\}, S, V)$ such that $n_x$ and $m_{xi}$ are
polynomial in $|x|$, $S$ and $V$ run in polynomial time in $|x|$, if $x \in
\mcL$ then $\mfk{w}_c(\ttt{G}_x) \geq c$, and if $x \not\in \mcL$ then
$\mfk{w}_c(\ttt{G}_x) \leq s$. The function $c$ is called the \textbf{completeness
probability}, and $s$ is called the \textbf{soundness probability}. The functions
$n_x$ and $m_{xi}$ are called the \textbf{question length} and \textbf{answer length}
respectively. The classes $\MIP^*(2,1,c,s)$ and $\MIP^{co}(2,1,c,s)$ are defined equivalently to
$\MIP(2,1,c,s)$, but with $\mfk{w}_c$ replaced by $\mfk{w}_q$ and $\mfk{w}_{qc}$, respecitvely. The protocols
in these cases are called $\MIP^*$ and $\MIP^{co}$ protocols. The class of classical $\MIP$ protocols were characterized in \cite{BFL91}, which found that $\MIP(2,1,1,1/2) = \NEXP$ with polynomial sized questions and constant sized answers. The $\MIP^* = \RE$ theorem of Ji, Natarajan, Vidick, Wright, and Yuen states that $\MIP^*(2,1,1,1/2) = \RE$
\cite{ji2022mipre}. We use the following stronger version of the $\MIP^* = \RE$ theorem due to \cite{DFNQXY23}.
\begin{theorem}\label{thm:constanswer}(\cite{DFNQXY23})
    $\RE$ is contained in $\MIP^\ast(2,1,1,1/2)$ with polynomial length questions and constant length answers.
\end{theorem}
The constant length answers will be key to preserving constant soundness in our reduction without requiring parallel repetition.

\subsection{Weighted algebra formalism}

We recall some key concepts from the theory of $*$-algebras. See \cite{ozawa2013connes, Schmdgen2020} for a more complete background. A complex $*$-algebra $\mcA$ is a unital algebra over $\mbC$ with an antilinear involution $a \mapsto a^*$, such that and $(ab)^*=b^*a^*$. In a $\ast$-algebra, we denote the \textbf{hermitian square} as $\hsq{a}:=a^\ast a$. Let $\mbC^*\langle X\rangle$ denote the free complex $*$-algebra generated by the set $X$. If $R \subseteq \mbC^*\langle X\rangle$, let $\mbC^*\langle X: R\rangle$  denote the quotient of $\mbC^*\langle X\rangle$ by the two-sided ideal generated by $R$. If $X$ and $R$ are finite then $\mbC^*\langle X:R\rangle$ is called a \textbf{finitely presented} $*$-algebra.

A $*$\textbf{-homomorphism} $\phi:\mcA \to \mcB$ between $*$-algebras is an algebra homomorphism such that $\phi(x^*) = \phi(x)^*$ for all $x \in A$. A $*$\textbf{-representation} of $\mcA$ is a $*$-homomorphism $\rho: \mcA \to \mcB(\mcH)$ from $\mcA$ to the $*$-algebra of bounded operators on a Hilbert space $\mcH$. If $\mcA$ and $\mcB$ are $*$-algebras, and $\C^*\langle X:R\rangle$ is a presentation of $\mcA$, then $*$-homomorphisms $\mcA\to \mcB$ correspond to homomorphisms $\phi: \mcC\langle X\rangle \to \mcB$ such that $\phi(r) = 0$ for all $r\in R$. Thus, a $*$\nobreakdash-representation is an assignment of operators to the elements of $X$ that satisfies the defining relations $R$.

If $\mcA$ is a $*$-algebra, we write $a \geq b$ if
$a-b$ is a sum of hermitian squares, \emph{i.e.} there is $k \geq 0$ and
$c_1,\ldots,c_k \in \mcA$ such that $a-b = \sum_{i=1}^k c_i^* c_i$. A finitely presented $*$-algebra $\mcA$ is called \textbf{archimedean} if for all $a\in \mcA$ there exists a $\lambda>0$ such that $a^*a\leq \lambda1$. The algebras we consider in this work are all archimedean. If $f:\mcA\to \mbC$ is a linear functional then $f$ is \textbf{positive} if $f(a)\geq0$ whenever $a\geq0$. A \textbf{state} on $\mcA$ is a positive unital hermitian linear functional $\tau:\mcA \to \mbC$, that is $\tau(a^*a)\geq 0$ for all $a \in \mcA$, $\tau(1) = 1$, and $\tau(a^*) = \overline{\tau(a)}$ for all $a \in \mcA$. A state is \textbf{tracial} if $\tau(ab) = \tau(ba)$ for all $a,b \in \mcA$, and \textbf{faithful} if $\tau(a^*a)>0$ for all $a\neq 0$. A tracial state $\tau$ induces the \textbf{trace norm} $\|a\|_{\tau} := \sqrt{\tau(a^* a)}$, also called the $\tau$-norm. Trace norms are unitarily invariant, meaning that $\|uav\|_{\tau} = \|a\|_{\tau}$ for all $a \in \mcA$, and all unitaries $u$ and $v$. An element $u \in \mcA$ is called \textbf{unitary} if $u^*u = 1 = uu^*$.

If $\rho: \mcA \to \mcB(\mcH)$ is a $*$-algebra representation, then a vector $|v\rangle \in \mcH$ is \textbf{cyclic} for $\rho$ if the closure of $\rho(\mcA)|v\rangle$ with respect to the Hilbert space norm is equal to $\mcH$. A \textbf{cyclic representation} of $\mcA$ is a tuple $(\rho,\mcH,|v\rangle)$, where $\rho$ is a representation of $\mcA$ on $\mcH$ and $|v\rangle$ is a cyclic vector for $\rho$. If $\tau:\mcA \to \mbC$ is a positive linear functional on $\mcA$, then there is a cyclic representation $\rho_{\tau}$ of $\mcA$, called the \textbf{GNS representation} of $\tau$, such that $\tau(a) = \bra{\xi_{\tau}}\rho_{\tau}(a)\ket{\xi_{\tau}}$ for all $a \in \mcA$. Two representations $\rho: \mcA \to \mcB(\mcH)$ and $\pi: \mcA \to \mcB(\mcK)$ of $\mcA$ are \textbf{unitarily equivalent} if there is a unitary operator $U: \mcH\to \mcK$ such that $U\rho(a)U^* = \pi(a)$ for all $a \in \mcA$. If $\tau$ is the state defined by $\tau(a) = \bra{\xi}\rho(a)\ket{\xi}$ for all $a \in \mcA$ and some cyclic representation $(\rho,\mcH,\ket{\xi})$, then $(\rho,\mcH,\ket{\xi})$ is unitarily equivalent to the GNS representation. A state $\tau$ is \textbf{finite-dimensional} if the Hilbert space $\mcH_{\tau}$ in the GNS representation $(\rho_{\tau},\mcH_{\tau},|\xi_{\tau}\rangle)$ is finite-dimensional. A state $\tau$ on $\mcA$ is called \textbf{Connes-embeddable} if there is a trace-preserving embedding of $\mcA$ into the ultrapower of the hyperfinite $\text{II}_1$ factor.

If $\mcA$ is a $*$-algebra then two
elements $a,b \in \mcA$ are said to be \textbf{cyclically equivalent} if there
is $k \geq 0$ and $f_1,\ldots,f_k,g_1,\ldots,g_k \in \mcA$ such that $a - b =
\sum_{i=1}^k [f_i, g_i]$, where $[f,g] = fg - gf$. We say that $a \gtrsim b$ if
$a-b$ is cyclically equivalent to a sum of squares.
If $\tau$ is a tracial state on $\mcA$ then $\tau(c_i^*c_i) \geq 0$ and $\tau([f_j,g_j]) = 0$. 
Thus if $a \gtrsim b$ then $\tau(a)\geq \tau(b)$, and if $a$ and $b$ are cyclically equivalent then $\tau(a-b) = 0$.

The $*$-algebras we use in this work are built out of the group algebras of the finitely presented groups
\begin{equation*}
	\Z_q^{*V} = \langle V: x^q = 1\text{ for all }x\in V\rangle \text{ and } \Z_q^V = \langle V: x^q = 1, xy=yx \text{ for all } x,y \in V\rangle.
\end{equation*}
Note that the group $\Z_q^V$ is in bijection with the functions $V\rightarrow\Z_q$, via $\phi\mapsto\prod_{x\in V}x^{\phi(x)}$; we often make use of this identification implicitly. The group algebra $\C\Z^{*V}_q$ is the $*$-algebra generated by variables $x \in V$ with the defining relations from $\Z_q^{*V}$, along with the relations $x^*x = xx^* = 1$ for all $x \in V$. Similarly $\C \Z_q^{V}$ is the $*$-algebra generated by variables $x \in V$ with the defining relations of $\Z_q^V$, along with the relations $x^* x = x x^* = 1$ for all $x \in V$. Notice that $\C
\Z_q^{V}$ is the quotient of $\C \Z_q^{*V}$ by the relations $xy = yx$ for all
$x,y \in V$. If $\mcA$ and $\mcB$ are complex $*$-algebras, then we let $\mcA \ast \mcB$ denote their free product, and $\mcA \otimes \mcB$ denote their tensor product. Both are again complex $*$-algebras. 

A $C^*$\textbf{-algebra} $\mcA$ is a complex $*$-algebra with a submultiplicative Banach norm that satisfies the $C^*$ identity $\|aa^*\| = \|a\|^2$ for all $a\in \mcA$. Every $C^*$-algebra can be realized as a norm-closed $*$-subalgebra of the algebra of bounded operators $\mcB(\mcH)$ on some Hilbert space $\mcH$. A $C^*$-algebra is a von Neumann algebra if it can be realized as a $*$-subalgebra of $\mcB(\mcH)$ which is closed in the weak operator topology. See \cite{blackadar06} for more background on $C^*$-algebras and von Neumann algebras.

We make use of the following algebraic tools from~\cite{MS24}.

\begin{lemma}[\cite{MS24}]\label{lem:hermitiansquare}
    Let $a_i\in\mc{A}$, where $\mc{A}$ is a $\ast$-algebra. Then, we have that $\hsq[\big]{\sum_{i=1}^ka_i}\leq 2^{\ceil{\log k}}\sum_{i=1}^k\hsq*{a_i}$.
\end{lemma}

\begin{definition}
    A \textbf{(finitely-supported) weight function} on a set $X$ is a function
    $\mu : X \to [0,+\infty)$ such that $\supp(\mu) := \mu^{-1}((0,+\infty))$
    is finite. A \textbf{weighted} $*$\textbf{-algebra} is a pair $(\mcA,\mu)$
    where $\mcA$ is a $*$-algebra and $\mu$ is a weight function on $\mcA$. 

    If $\tau$ is a tracial state on $\mcA$, then the \textbf{defect of $\tau$} is
    \begin{equation*}
        \df(\tau; \mu) := \sum_{a \in \mcA} \mu(a) \|a\|^2_{\tau},
    \end{equation*}
    where $\|a\|_{\tau} := \sqrt{\tau(a^* a)}$ is the $\tau$-norm.
    When the weight function is clear, we just write~$\df(\tau)$. 
\end{definition}
Since $\mu$ is finitely supported, the sum in the definition of the
defect is finite and hence is well-defined. If $\tau$ is a tracial state on the weighted algebra $(\mcA,\mu)$, and $\df(\tau) = 0$, then $\tau$ is the pullback of a tracial state on the algebra $\mcA/\langle \supp(\mu)\rangle$. The defect thus measures how far away traces on $\mcA$ are from being traces on $\mcA/\langle \supp(\mu)\rangle$. With this in mind, the weighted algebra $(\mcA,\mu)$ can be thought of as a model of the algebra $\mcA/\langle\supp(\mu)\rangle$ with the defect allowing us to discuss approximate traces on this algebra. To keep track of how transformations affect approximate traces we need the following notion of homomorphism between weighted algebras.
\begin{definition}
    Let $(\mcA,\mu)$ and $(\mcB,\nu)$ be weighted $*$-algebras, and let $C > 0$. A
    $C$\textbf{-homomorphism} $\alpha:(\mcA,\mu)\to(\mcB,\nu)$ is a
    $*$-homomorphism $\alpha:\mcA\to\mcB$ such that 
    \begin{equation*}
		\alpha(\sum_{a\in\mcA}\mu(a)a^*a) \lesssim C\sum_{b\in\mcB}\nu(b)b^*b.
	\end{equation*}
\end{definition}
The following lemma shows how $C$-homomorphsims affect approximate traces.
\begin{lemma}[\cite{MS24}]\label{lem:Chom}
    Suppose $\alpha : (\mcA,\mu) \to (\mcB,\nu)$ is a $C$-homomorphism. If $\tau$ is a 
    trace on $(\mcB,\nu)$, then $\df(\tau \circ \alpha) \leq C \df(\tau)$. 
\end{lemma}

\subsection{Von Neumann algebras and stability}\label{sec:vna}

A tracial von Neumann algebra is a von Neumann algebra $\mcM$ equipped with a faithful normal tracial state $\tau$, and $\mcU(\mcM)$ is the unitary group of $\mcM$. If $\tau$ is a tracial state on a $*$-algebra $\mcA$, and $(\rho : \mcA \to \mcB(\mcH), \ket{v})$ is the GNS representation, then the closure $\mcM = \overline{\rho(\mcA)}$ of $\rho(\mcA)$ in the weak operator topology is a von Neumann algebra, and $\tau_0(a) = \bra{v}a\ket{v}$ is a faithful normal tracial state on $\mcM$.

\begin{lemma}[\cite{chapman2023efficiently}] \label{lem:z2stab}
    Let $(\mcM,\tau)$ be a tracial von Neumann algebra, and suppose $f : [k]
    \to \mcM$ is a function such that $f(i)^2 = 1$ for all $i \in [k]$ and
    $\|[f(i),f(j)]\|_{\tau}^2 \leq \eps$ for all $i,j \in [k]$, where $k \geq 1$ and
    $\eps \geq 0$. Then there is a homomorphism $\psi : \Z_2^k \to
    \mcU(\mcM)$ such that $\|\psi(x_i) - f(i)\|_{\tau}^2 \leq \poly(k) \eps$
    for all $i \in [k]$, where the $x_i$ generate $\Z_2^k$.
\end{lemma}

A function $f$ satisfying the conditions of \Cref{lem:z2stab} is called an
\textbf{$\eps$-homomorphism from $\Z_2^{k}$ to $\mcU(\mcM)$}.

\section{Constraint system algebras}\label{sec:cs}

\subsection{Definitions}\label{sec:csp-defs}

\begin{definition}
    A \textbf{constraint} over an alphabet $\Sigma$ (finite set) is a pair $(V,C)$ where $V$ is a finite set of variables, called the \textbf{context}, and $C\subseteq\Sigma^V$. 
\end{definition}

When the context is clear, we refer to $C$ as the constraint. We also identify subsets of $\Sigma^n$ with constraints with context $[n]$. For a constraint $C\subseteq\Sigma^V$ and $U\subseteq V$, write the constraint restricted to context $U$ as $C|_U=\set*{\phi|_U}{\phi\in C}$.

\begin{definition}
    A \textbf{constraint system} (CS) $S$ over an alphabet $\Sigma$ is a set of variables $X$ along with constraints $(V_i,C_i)$ for $i=1,\ldots,m$, where $V_i\subseteq X$.
    \begin{itemize}
        \item A CS is \textbf{satisfiable} if there exists an assignment $f:X\rightarrow\Sigma$ such that $f|_{V_i}\in C_i$ for all~$i$.

        \item If $\Sigma=\Z_k$, we say $S$ is a \textbf{$k$-ary} CS (note that we can identify any CS with a $|\Sigma|$-ary CS); and if $\Sigma=\Z_2$, we say $S$ is a \textbf{boolean CS} (BCS).

        \item If every $V_i$ has two elements, we say $S$ is a \textbf{$2$-CS}.
    \end{itemize}
\end{definition}

For any finite set of variables $X$, we can suppose that there is an ordering on the set, and that that ordering induces an ordering on any subset of variables $V\subseteq X$. For a set of variables $V$, recall that $\C\Z_k^{\ast V}$ is the free algebra generated by order-$k$ unitaries labelled by the elements $x\in V$, and $\C\Z_k^{V}$ is its abelianisation. For any order-$k$ unitary $x$ and $a\in\Z_k$, write $\Pi_a^{(k)}(x)=\frac{1}{k}\sum_{n=0}^{k-1}\omega_k^{-na}x^n$ for the projector onto the $\omega_k^a$-eigenspace of $x$; where $k$ is clear, we suppress the superscript $(k)$. Given $\phi\in\Z_k^{V}$, we define the element $\Phi_{V,\phi}^{(k)}\in\C\Z_k^{\ast V}$ as
$$\Phi_{V,\phi}^{(k)}=\prod_{x\in V}\Pi_{\phi(x)}^{(k)}(x),$$
where the product is ordered according to the ordering of $V$, and as before the superscript~$(k)$ is suppressed where clear. We use the same notation for the image of $\Phi_{V,\phi}^{(k)}$ under any homomorphism, when the homomorphism is clear. In particular, in $\C\Z_k^{V}$, $\{\Phi_{V,\phi}\}_{\phi\in\Z_k^V}$ forms a generating PVM for the commutative algebra. Given a constraint $(V,C)$ over the alphabet $\Z_k$, we write
$$\mc{A}(V,C)=\C\Z_k^{V}/\gen{\Phi_{V_i,\phi}}{\phi\notin C}.$$
This is isomorphic to the $C^\ast$-algebra of functions on the finite set $C$.

\begin{definition}\label{def:alg}
    Let $S=(X,\{(V_i,C_i)\}_{i=1}^m)$ be a $k$-ary CS. We define the following algebras:
    \begin{itemize}
        \item \textbf{The constraint-constraint algebra} $\mc{A}_{c-c}(S)=\bigast_{i=1}^m\mc{A}(V_i,C_i)$: Write the inclusion map $\sigma_i:\mc{A}(V_i,C_i)\rightarrow\mc{A}_{c-c}(S)$, and write $\Phi_{V_i,\phi}$ for $\sigma_i(\Phi_{V_i,\phi})$ where clear. Given a probability distribution $\pi$ on $[m]\times [m]$, define the weighted algebra $\mc{A}_{c-c}(S,\pi)=(\mc{A}_{c-c},\mu_{c-c,\pi})$, where the weight function $\mu_{c-c,\pi}(\Phi_{V_i,\phi}\Phi_{V_j,\psi})=\pi(i,j)$ for all $\phi \in C_i$, $\psi\in C_j$ with $\phi|_{V_i\cap V_j} \neq \psi|_{V_i\cap V_j}$, and $0$ on all other elements.
        \item The \textbf{inter-contextual algebra}: Given a probability distribution $\pi$ on $[m]\times[m]$, define the weighted algebra $\mcA_{inter}(S,\pi) = (\mc{A}_{c-c}(S),\mu_{inter,\pi})$, with weight function $\mu_{inter,\pi}(\sigma_i(x)^l-\sigma_j(x)^l) = \pi(i,j)$ for all $i\neq j \in [m]$, $l \in [k-1]$ and $x\in V_i\cap V_j$, and $0$ on all other elements.
        \item The \textbf{constraint-variable algebra} $\mc{A}_{c-v}(S)=\bigast_{i=1}^m\mc{A}(V_i,C_i)\ast\C\Z_k^{\ast X}$: As above, write the inclusion $\sigma_i:\mc{A}(V_i,C_i)\rightarrow\mc{A}_{c-v}(S)$, and write also the inclusion $\sigma':\C\Z_k^{\ast X}\rightarrow\mc{A}_{c-v}(S)$. For a probability distribution $\pi'$ on $[m]$, define the weighted algebra $\mc{A}_{c-v}(S,\pi')=(\mc{A}_{c-v}(S),\mu_{c-v,\pi'})$ where the weight function $\mu_{c-v}(\Phi_{V_i,\phi}(1-\Pi_{\phi(x)}(\sigma'(x))))=\frac{\pi'(i)}{|V_i|}$ for all $x\in V_i$, $\phi\in C_i$, and $0$ on all other elements.
        \item The \textbf{assignment algebra} $\mc{A}_a(S)=\C\Z_k^{\ast X}$: Given a probability distribution $\pi'$ on $[m]$, define the weighted algebra $\mc{A}_a(S,\pi')=(\mc{A}_a(S),\mu_{a,\pi'})$ where the weight function $\mu_{a,\pi'}(\Phi_{V_i,\phi})=\pi'(i)$ for all $\phi\notin C_i$, and $0$ on all other elements.

        \item The \textbf{assignment algebra with commutation}: Define $\mc{A}_{a+comm}(S,\pi')=(\mc{A}_(S),\mu_{a,\pi'}+\mu_{comm,\pi'})$, where the weight function $\mu_{comm,\pi'}([\Pi_a(x),\Pi_b(y)])=\sum_{i.\;x,y\in V_i}\pi'(i)$ for all $a,b\in\Z_k$ and $x,y\in X$, and $0$ on all other elements.
    \end{itemize}
\end{definition}

The goal of defining these algebras is that their traces correspond to synchronous strategies for constraint system games, or $\CS$ games. Let $S = (X,\{(V_i,C_i)\}_{i=1}^m)$ be a $k$-ary $\CS$, and let $\pi$ be a probability distribution on $[m]\times[m]$. The \textbf{CS game} $\ttt{G}(S,\pi)$ is the nonlocal game $([m],\{C_{i}\}_{i\in[m]}, \pi, V)$, where $V(\phi,\psi|i,j) = 1$ if $\phi|_{V_i\cap V_j} = \psi|_{V_i\cap V_j}$, and is $0$ otherwise. In the game $\ttt{G}(S,\pi)$, the players receive question pair $(i,j) \in [m]\times [m]$ with probability $\pi(i,j)$, and must answer with satisfying assignments $\phi\in C_i$ and $\psi\in C_j$. They win if their assignments agree on the variables in $V_i\cap V_j$.

Let $\pi'$ be a probability distribution on $[m]$, and define the probability distribution $\nu(i,x) = \pi'(i)/|V_i|$ for all $i\in [m]$ and $x \in V_i$, and $0$ otherwise. We can now also define the \textbf{constraint-variable CS game} $\ttt{G}_{c-v}(S,\pi')$, which is the nonlocal game $([m]\times X, \{C_{i}\times \Z_k\}_{i\in[m],x\in X}, \nu, V')$, where and $V'(\phi,a|i,x) = 1$ if $\phi(x) = a$, and is $0$ otherwise. Note that in this case, the players have different question sets, but the game can be symmetrised without changing the synchronous winning probability. In this game, one player receives $i\in[m]$ sampled from $\pi'$ and the other receives a uniformly random variable $x\in V_i$. To win, the first player must answer a satisfying assignment $\phi\in C_i$ and the second must answer $a\in\Z_k$ such that $a=\phi(x)$.

For a $2$-CS $S=(X,\{(V_i,C_i\}_{i=1}^m)$ and a probability distribution $\pi'$ on $[m]$, there is a natural way to associate a nonlocal game to the assignment algebra $\mc{A}_a(S,\pi')$. Define the \textbf{{$2$\nobreakdash-CS} game} $\ttt{G}_a(S,\pi')$ as the nonlocal game $(X,\Z_k,\nu_a,V_a)$, where $\nu_a(x,y)=\frac{\pi'(i)}{2}$ if $V_i=\{x,y\}$ and $0$ otherwise, and $V_a(a,b|x,y)=1$ iff there exists $\phi\in C_i$ such that $\phi(x)=a$ and $\phi(y)=b$. In this game, the referee samples $i$ from $\pi'$, and then asks each of the players one variable from the constraint. Each player responds with an assignment to the variable she received. They win if they have answered a satisfying assignment. On the other hand, there does not seem to be a natural way to associate a nonlocal game to the assignment algebra for any CS that is not a $2$-CS. Similarly, the assignment algebra with commutation is not directly associated to any nonlocal game, but it is technically very useful in connecting the constraint-variable and assignment algebras.

Unlike the other algebras in \Cref{def:alg}, there is no game associated directly to $\mcA_{inter}(S,\pi)$. Instead, the inter-contextual algebra is used as an intermediary step in the proof of soundness of the subdivision transformation in \Cref{sec:subdivision}.

Representations $\alpha$ of $\mcA_{c-c}(S)$ are in bijective correspondence with families of projective measurements $\{M_{\phi}^i\}_{\phi\in C_i}, i\in [m]$ via the relation $M_{\phi}^i = \alpha(\Phi_{V_i,\phi})$. If $(\{M^i_{\phi}\},\ket{v},\mcH)$ is a synchronous commuting operator strategy for $\ttt{G}(S,\pi)$, and $\alpha: \mcA_{c-c}(S) \to \mcB(\mcH)$ is the representation of $\mcA_{c-c}(S)$ with $\alpha(\Phi_{V_i,\phi}) = M^i_{\phi}$, then $a \to \bra{v}\alpha(a)\ket{v}$ is a tracial state on $\mcA_{c-c}(S)$. Conversely, if $\tau$ is a tracial state on $\mcA_{c-c}(S)$, then the GNS representation theorem implies that there is a synchronous commuting operator strategy $(\{M_{\phi}^i\},\ket{v}, \mcH)$ such that $\tau(a) = \bra{v}\alpha(a)\ket{v}$ where $\alpha$ is the representation of $\mcA_{c-c}(S)$ with $\alpha(\Phi_{V_i,\phi}) = M_{\phi}^i$. Note that the trace is faithful on the image of the GNS representation. As a result of this correspondence, we can use synchronous commuting operator strategies for $\ttt{G}(S,\pi)$ and tracial states on $\mcA_{c-c}(S)$ interchangeably. A correlation $p$ for the $\CS$ game $\ttt{G}(S,\pi)$ is in $C_{qc}$ if and only if there is a tracial state $\tau$ on $\mcA_{c-c}(S)$ such that $p(\phi,\psi|i,j) = \tau(\Phi_{V_i,\phi}\Phi_{V_j,\psi})$ for all $i,j,\phi$, and $\psi$. Finite-dimensional tracial states on $\mcA_{c-c}(S)$ can be used interchangeably with synchronous quantum strategies for $\ttt{G}(S,\pi)$, and $p\in C_q$ if and only if there is a finite-dimensional tracial state $\tau$ with $p(\phi,\psi|i,j) = \tau(\Phi_{V_i,\phi}\Phi_{V_j,\psi})$ for all $i,j,\phi$, and $\psi$. Similarly, $p \in C_{qa}$ if and only if there is a Connes-embeddable tracial state $\tau$ such that $p(\phi,\psi|i,j) = \tau(\Phi_{V_i,\phi}\Phi_{V_j,\psi})$ for all $i,j,\phi$, and $\psi$. Similarly, traces on $\mcA_{c-v}(S,\pi)$ or $\mc{A}_a(S,\pi)$ are interchangeable with synchronous strategies for the constraint-variable game or the $2$-CS game, respectively.

A correlation $p$ for the $\CS$ game $\ttt{G}(S,\pi)$ is perfect if $p(\phi,\psi|i,j) = 0$ whenever $\pi(i,j)>0$ whenever $\phi|_{V_i\cap V_j}\neq \psi|_{V_i\cap V_j}$. Thus, a tracial state $\tau$ is called \textbf{perfect} if and only if $\tau(\Phi_{V_i,\phi}\Phi_{V_j,\psi}) = 0$ whenever $\phi|_{V_i\cap V_j}\neq \psi|_{V_i\cap V_j}$. Consequently a tracial state on $\mcA_{c-c}(S)$ is perfect 
for $\ttt{G}(S,\pi)$ if and only if it is the pullback of a tracial state on the
\textbf{synchronous algebra} of $\ttt{G}(S,\pi)$, which is the quotient
\begin{align*}
    \SynAlg(S,\pi) = \mcA_{c-c}(S) / \ang{& \Phi_{V_i,\phi} \Phi_{V_j,\psi} = 0 \text{ for all }
            i,j \in [m] \text{ with } \pi(i,j) > 0 \\ & \text{ and } \phi \in C_i,
            \psi \in C_j \text{ with } \phi|_{V_i \cap V_j} \neq \psi|_{V_i \cap V_j}}.
\end{align*}
Thus, $\mcA_{c-c}(S,\pi)$ can be thought of as a model of $\SynAlg(S,\pi)$. In fact, assuming that for all $x\in V_i$ there exists some $j$ such that $x\in V_j$ and $\pi(i,j)>0$, each of the algebras defined above --- except in general the assignment algebra --- are models for $\SynAlg(S,\pi)$, since
\begin{align*}
    \SynAlg(S,\pi) &= \mcA_{c-c}(S)/\langle \supp(\mu_{c-c,\pi})\rangle = \mcA_{c-c}(S)/\langle \supp(\mu_{inter,\pi})\rangle 
    \\
    &= \mcA_{c-v}(S)/\langle \supp(\mu_{c-v,\pi'})\rangle = \mcA_a(S)/\langle \supp(\mu_{a,\pi'}+\mu_{comm,\pi'})\rangle,
\end{align*} 
where $\pi'(i) = \sum_j \pi(i,j)$. Thus if one of these algebras has a perfect trace then all of the others do. In the next section we will see what the relationship between imperfect traces on these algebras is.

The notion of a $\CS$ game leads to the following special version of an $\MIP$ protocol.
A \textbf{$\CS$-$\MIP$ protocol} is a family of CS games
$\ttt{G}(S_x,\pi_x)$, where $S_x = (X_x,\{(V_i^x,C_i^x)\}_{i=1}^{m_x})$, along
with a probabilistic Turing machine $Q$ and another Turing machine $C$, such
that 
\begin{enumerate} 
    \item on input $x$, $Q$ outputs $(i,j) \in [m_x] \times [m_x]$ with probability
$\pi_x(i,j)$, and
    \item on input $(x, \phi, i)$, $C$ outputs true if $\phi \in C_i^x$ and false
        otherwise.
\end{enumerate}
Technically, this definition should also include some way of computing the sets
$X_x$ and $V_i^x$. For instance, we might say that the integers $|X_x|$ and
$|V_i^x|$ are all computable, and there are computable order-preserving
injections $[|V_i^x|] \to [|X_x|]$. For simplicity we ignore this
aspect of the definition in what follows, and just assume that in any
$\CS$-$\MIP^*$ protocol, we have some efficient way of working with the sets
$X_x$ and $V_i^x$, the intersections $V_i^x \cap V_j^x$, and assignments
$\phi\in\Z_2^{V_i^x}$. 
A language $\mcL$ belongs to the complexity class $\CS$-$\MIP(s)$ if there is a
$\CS$-$\MIP$ protocol as above such that $\lceil \log m_x \rceil$ and $|V_i^x|$
are polynomial in $|x|$, $Q$ and $C$ run in polynomial time, if $x \in \mcL$ then
$\mfk{w}_c^s(\ttt{G}_x) = 1$, and if $x \not\in \mcL$ then $\mfk{w}_c^s(\ttt{G}_x) \leq s$.
The parameter $s$ is called the soundness. The classes $\CS$-$\MIP^*(s)$ and $\CS$-$\MIP^{co}(s)$ are defined equivalently to $\CS$-$\MIP(s)$, but with $\mfk{w}_c$ replaced with $\mfk{w}_q$ and $\mfk{w}_{qc}$, respectively. If the constraint systems in a $\CS-\MIP$ protocol are boolean then we call it a $\BCS-\MIP$ protocol. 

An analogous protocol can be constructed using the constraint-variable version of the CS games, which we call a \textbf{constraint-variable $\CS$-$\MIP$ protocol}. Here, since the constraint-variable game is not naturally synchronous, the only difference is that the verifier must randomly choose which player to ask the constraint question and which player to ask the variable question, and also ask consistency check questions with some constant probability. This transformation preserves constant gap and guarantees that the players can win near-optimally with synchronous strategies due to~\cite{marrakchi2023synchronous}, so we do not need to worry about it in practice.

\subsection{Relations between CS algebras}

\begin{figure}
    \centering
    \begin{tikzpicture}
        \node (cc) {$\mc{A}_{c-c}(S,\pi)$};
        \node (inter) [above = 2.5cm of cc]{$\mc{A}_{inter}(S,\pi)$};
        \node (cv) [right=1.5cm of cc] {$\mc{A}_{c-v}(S,\pi')$};
        \node (acomm) [right=1.5cm of cv] {$\mc{A}_{a+comm}(S,\pi')$};
        \node (a) [above=2.5cm of acomm] {$\mc{A}_{a}(S,\pi')$};
        \node (bcs) [left=1.5cm of cc] {$\mc{A}_{c-c}(B(S),\pi)$};
        \draw[-Latex] (cc) to [bend right] node[right]{$O(1)$} (inter);
        \draw[-Latex] (inter) to  [bend right] node[left]{$O(L)$} (cc);
        \draw[-Latex] (cc) to [bend right] node[below]{$O(L)$} (cv);
        \draw[-Latex] (cv) to [bend right] node[above]{$O(P)$} (cc);
        \draw[-Latex] (bcs) to [bend right] node[below]{$O(1)$} (cc);
        \draw[-Latex] (cc) to [bend right] node[above]{$O(1)$} (bcs);
        \draw[-Latex] (acomm) to [bend right] node[above]{$O(L^2)$} (cv);
        \draw[Latex-, dashed] (acomm) to [bend left] node[below]{$\poly(k^L)$} (cv);
        \draw[-Latex] (a) to [bend left] node[right]{$O(1)$} (acomm);
    \end{tikzpicture}
    \caption{$C$-homomorphisms (solid arrows) and trace-dependent mappings (dashed arrows) between the weighted algebras considered, for a $k$-ary CS $S=(X,\{(V_i,C_i)\}_{i=1}^m)$. Here $\pi'(i)=\sum_j\pi(i,j)$, $L=\max_i|V_i|$, $P=\max_{i,j.\;V_i\cap V_j\neq\varnothing}\frac{\pi'(i)}{\pi(i,j)}$, and $B(S)$ is the BCS defined in~\Cref{def:bs}.}
    \label{fig:alg-c-homs}
\end{figure}

In this section, we study the relationships between the weighted algebras coming from a CS. Fig.~\ref{fig:alg-c-homs} summarises the $C$-homomorphisms we find between these weighted algebras.
First, we note that any constraint-constraint algebra is equivalent to a related BCS algebra.

\begin{definition}\label{def:bs}
    Given a $k$-ary CS $S=(X,\{(V_i,C_i)\}_{i=1}^m)$, the \textbf{boolean form} of $S$ is defined as $B(S)=(X',\{(V_i',C_i')\}_{i=1}^m)$, where $X'=\set*{(x,a)}{x\in X,a\in\Z_k}$, $V_i'=\set*{(x,a)}{x\in V_i, a\in\Z_k}$, $C_i'=\set*{\phi'}{\phi\in C_i}$ where
    $$\phi'(x,a)=\begin{cases}1&\phi(x)=a\\0&\text{ else}\end{cases}.$$
\end{definition}

This corresponds to replacing each variable in $X$ with $k$ indicator variables indicating which value in $\Z_k$ is assigned to $x$.

\begin{lemma}\label{lem:bcs-to-kcs}
    Let $S=(X,\{(V_i,C_i)\}_{i=1}^m)$ be a $k$-ary CS and $\pi$ be a probability distribution on $[m]\times[m]$. There is a $\ast$-isomorphism $\alpha:\mc{A}_{c-c}(B(S))\rightarrow\mc{A}_{c-c}(S)$ such that $\alpha$ and $\alpha^{-1}$ are $1$-homomorphisms between $\mc{A}_{c-c}(B(S),\pi)$ and $\mc{A}_{c-c}(S,\pi)$.
\end{lemma}

\begin{proof}
    Write $B(S)=(X',\{(V_i',C_i')\}_{i=1}^m)$. Consider first the map $\alpha_i:\C\Z_2^{V_i'}\rightarrow\C\Z_k^{V_i}$ defined as the homomorphism such that
    $$\alpha((x,a))=1-2\Pi^{(k)}_a(x).$$
    Since for any $\phi\in\Z_2^{V_i'}$, $\Phi_{V_i',\phi}=\prod_{(x,a)\in V_i'}\frac{1+(-1)^{\phi(x,a)}(x,a)}{2}=\prod_{x\in V_i}\prod_{a\in\Z_k}\frac{1+(-1)^{\phi(x,a)}(x,a)}{2}$,
    \begin{align*}
        \alpha_i(\Phi_{V_i',\phi})&=\prod_{x\in V_i}\prod_{a\in\Z_k}\frac{1+(-1)^{\phi(x,a)}(1-2\Pi^{(k)}_a(x))}{2}=\prod_{x\in V_i}\prod_{\phi(x,a)=0}(1-\Pi^{(k)}_a(x))\prod_{\phi(x,a)=1}\Pi^{(k)}_a(x)\\
        &=\begin{cases}\prod_{x\in V_i}\Pi^{(k)}_{\psi(x)}(x)&\exists\,\psi\in\Z_k^{V_i}.\;\phi(x,b)=1\iff b=\psi(x)\;\forall\,x\in V_i\\0&\text{ otherwise}\end{cases}.
    \end{align*}
    This means that, if $\alpha_i(\Phi_{V_i',\phi})\neq 0$, there exists $\psi\in\Z_k^{V_i}$ such that $\phi=\psi'$. Further, in this case $\alpha_i(\Phi_{V_i',\phi})=\Phi_{V_i,\psi}^{(k)}$. As such, the map $\alpha_i$ factors through to a map $\bar{\alpha}_i:\mc{A}(V_i',C_i')\rightarrow\mc{A}(V_i,C_i)$. Since we have $\bar{\alpha}_i(\Phi_{V_i',\phi'})=\Phi_{V_i,\phi}^{(k)}$ for all $\phi\in C_i$, $\bar{\alpha}_i$ is an isomorphism. Finally, we take $\alpha:\mc{A}_{c-c}(B(S))\rightarrow\mc{A}_{c-c}(S)$ to be given by $\bar{\alpha}_i$ on the corresponding term of the free product. By construction, $\alpha$ is an isomorphism, and $\alpha$ and $\alpha^{-1}$ both preserve the weight function as they exchange $\Phi_{V_i',\phi'}$ and $\Phi_{V_i,\phi}^{(k)}$.
\end{proof}

Now we examine the homomorphisms between the $\mcA_{c-c}(S,\pi)$ and $\mcA_{inter}(S,\pi)$ algebras.

\begin{prop}\label{prop:inter}
    Suppose $S = (X, \{(V_i,C_i)\}_{i=1}^m)$ is a $k$-ary $\CS$, and $\pi$ is a
    probability distribution on $[m] \times [m]$. Then the identity map $\mcA_{c-c}(S) \to \mcA_{c-c}(S)$ gives an
    $O(1)$-homomorphism $\mcA_{c-c}(S,\pi)\to \mcA_{inter}(S,\pi)$, and an
    $O(L)$-homomorphism $\mcA_{inter}(S,\pi)\to \mcA_{c-c}(S,\pi)$, where
    $L = \max_{i,j} |V_i \cap V_j|$.
\end{prop}
Recall that $\sigma_i : \mcA(V_i,C_i) \to \mcA_{c-c}(S)$ is the natural inclusion of the $i$th factor.  
\begin{proof}
    Fix $1 \leq i,j \leq m$. Since $\Phi_{V_i,\phi}$ is a projection in $\mcA(V_i,C_i)$,
    $(\Phi_{V_i,\phi} \Phi_{V_j,\psi})^* (\Phi_{V_i,\phi} \Phi_{V_j,\psi})$ is cyclically
    equivalent to $\Phi_{V_i,\phi} \Phi_{V_j,\psi}$ for all $\phi \in C_i$, $\psi \in C_j$.
    For $x \in V_i \cap V_j$, let $R_x$ be the pairs $(\phi,\psi) \in C_i \times C_j$ such
    that $\phi(x) \neq \psi(x)$. Then 
    \begin{equation*}
        \sum_{\phi|_{V_i \cap V_j} \neq \psi|_{V_i \cap V_j}} \Phi_{V_i,\phi} \Phi_{V_j,\psi}
        \lesssim \sum_{x \in V_i \cap V_j} \sum_{(\phi,\psi) \in R_x} \Phi_{V_i,\phi} \Phi_{V_j,\psi},
    \end{equation*}
    and since $\phi|_{V_i \cap V_j}$ and $\psi|_{V_i \cap V_j}$ can disagree in at most $|V_i \cap V_j|$
    places, 
    \begin{equation*}
        \sum_{x \in V_i \cap V_j} \sum_{(\phi,\psi) \in R_x} \Phi_{V_i,\phi} \Phi_{V_j,\psi} \lesssim
        |V_i \cap V_j|
        \sum_{\phi|_{V_i \cap V_j} \neq \psi_{V_i \cap V_j}} \Phi_{V_i,\phi} \Phi_{V_j,\psi}.
    \end{equation*}
    Fix $x \in V_i \cap V_j$, and let $V_i' = V_i \setminus \{x\}$, $V_j' = V_j \setminus \{x\}$. 
    \begin{align*}
        \sum_{(\phi,\psi) \in R_x} \Phi_{V_i,\phi} \Phi_{V_j,\psi} & = 
        \sum_{\substack{\phi \in \Z_k^{V_i'}, \psi \in \Z_k^{V_j'}\\a\neq b}} \Phi_{V_i',\phi} 
            \left[\Pi_a^{(k)}(\sigma_i(x)) \Pi_b^{(k)}(\sigma_j(x))\right] \Phi_{V_j',\psi} \\
            & = \sum_{a \neq b}\Pi_a^{(k)}(\sigma_i(x)) \Pi_b^{(k)}(\sigma_j(x)),
    \end{align*}
    where the last equality holds because $\sum_{\phi \in \Z_k^{V_i'}} \Phi_{V_i',\phi}$ and
    $\sum_{\psi \in \Z_k^{V_j'}} \Phi_{V_i',\psi}$ are both equal to $1$. Notice that $\sum_{a\neq b}\Pi_a^{(k)}(\sigma_i(x)) \Pi_b^{(k)}(\sigma_j(x)) = \sum_a \Pi^{(k)}_a(\sigma_i(x))(1-\Pi^{(k)}_a(\sigma_j(x)))\lesssim \frac{1}{2}\sum_a \hsq*{\Pi^{(k)}_a(\sigma_i(x))-\Pi_a^{(k)}(\sigma_j(x))}$. Thus, we get
    \begin{align*}
        \sum_{(\phi,\psi) \in R_x} \Phi_{V_i,\phi} \Phi_{V_j,\psi} &\lesssim \frac{1}{2}\sum_a \hsq*{\Pi^{(k)}_a(\sigma_i(x))-\Pi_a^{(k)}(\sigma_j(x))}
        \\
        &= \frac{1}{2}\sum_a \hsq[\Big]{\frac{1}{k}\sum_{\ell = 0}^{k-1}a^{-\ell}(\sigma_i(x)^{\ell}-\sigma_j(x)^{\ell})}
        \\
        &\leq 2^{\ceil{\log k}-1}\sum_a \frac{1}{k^2}\sum_{\ell = 0}^{k-1}\hsq*{\sigma_i(x)^{\ell}-\sigma_j(x)^{\ell}}
        \\
        &\leq 2^{\ceil{\log k}-1}\frac{1}{k}\sum_{\ell = 0}^{k-1}\hsq*{\sigma_i(x)^{\ell}-\sigma_j(x)^{\ell}},
    \end{align*}
where the second last inequality is due to \Cref{lem:hermitiansquare}. Finally, notice that $\sum_{\ell = 0}^{k-1}\hsq*{\sigma_i^{\ell}(x)-\sigma_j^{\ell}(x)}$ is cyclically equivalent to $\sum_{a\in\Z_k}2\Pi_a^{(k)}(\sigma_i(x))(\Pi_a^{(k)}(\sigma_i(x))-\Pi_a^{(k)}(\sigma_j(x)))$. Then 
\begin{equation*}
    \Pi_a^{(k)}(\sigma_i(x))(\Pi_a^{(k)}(\sigma_i(x))-\Pi_a^{(k)}(\sigma_j(x))) = \Pi_a^{(k)}(\sigma_i(x))(1-\Pi_a^{(k)}(\sigma_j(x)))
\end{equation*}
gives the result.
\end{proof}

Next, we look at homomorphisms between constraint-constraint and constraint-variable algebras.

\begin{lemma}\label{lem:cc-to-cv}
    Let $S=(X,\{(V_i,C_i)\}_{i=1}^m)$ be a $k$-ary CS and $\pi$ be a symmetric probability distribution on $[m]\times[m]$. There is a $4L$-homomorphism $\alpha:\mc{A}_{c-c}(S,\pi)\rightarrow\mc{A}_{c-v}(S,\pi')$ where $\pi'(i)=\sum_{j}\pi(i,j)$ and $L=\max_{i}|V_i|$.
\end{lemma}

\begin{proof}
    Let $\alpha$ be the inclusion of $\mc{A}_{c-c}(S)$ in $\mc{A}_{c-v}(S)$. Then, similarly to Proposition 6.3 in~\cite{MS24}, for all $i,j$
    \begin{align*}
        \sum_{\substack{\phi\in C_i,\psi\in C_j\\\phi|_{V_i\cap V_j}=\psi|_{V_i\cap V_j}}}\hsq*{\Phi_{V_i,\phi}\Phi_{V_j,\psi}}&\lesssim\sum_{\substack{\phi\in C_i,\psi\in C_j\\\phi|_{V_i\cap V_j}=\psi|_{V_i\cap V_j}}}\Phi_{V_i,\phi}\Phi_{V_j,\psi}\lesssim\sum_{x\in V_i\cap V_j}\sum_{\substack{\phi\in C_i,\psi\in C_j\\\phi(x)\neq\psi(x)}}\Phi_{V_i,\phi}\Phi_{V_j,\psi}\\
        &=\sum_{x\in V_i\cap V_j}\sum_{\phi\in C_i,\psi\in C_j}\Phi_{V_i,\phi}\sum_{a\neq b}\Pi_a(\sigma_i(x))\Pi_b(\sigma_j(x))\Phi_{V_j,\psi}\\
        &=\sum_{x\in V_i\cap V_j}\sum_{a\neq b}\Pi_a(\sigma_i(x))\Pi_b(\sigma_j(x)).
    \end{align*}
    Now, noting that $\sum_{a\neq b}\Pi_a(\sigma_i(x))\Pi_b(\sigma_j(s))=\sum_a\Pi_a(\sigma_i(x))(1-\Pi_a(\sigma_j(x)))\lesssim\frac{1}{2}\sum_a\hsq*{\Pi_a(\sigma_i(x))-\Pi_a(\sigma_j(x))}$, we get that
    \begin{align*}
        \sum_{\substack{i,j\\\phi\in C_i,\psi\in C_j\\\phi|_{V_i\cap V_j}\neq\psi|_{V_i\cap V_j}}}\pi(i,j)\hsq*{\Phi_{V_i,\phi}\Phi_{V_j,\psi}}&\lesssim\frac{1}{2}\sum_{\substack{i,j\\x\in V_i\cap V_j}}\pi(i,j)\sum_a\hsq*{\Pi_a(\sigma_i(x))-\Pi_a(\sigma_j(x))}\\
        &\leq\sum_{\substack{i,j\\x\in V_i\cap V_j}}\pi(i,j)\sum_a\parens*{\hsq*{\Pi_a(\sigma_i(x))-\Pi_a(\sigma'(x))}+\hsq*{\Pi_a(\sigma_j(x))-\Pi_a(\sigma'(x))}}.
    \end{align*}
    Next, using the symmetry of $\pi$ and the fact that $L\geq|V_i|$ for all $i$,
    \begin{align*}
        \sum_{\substack{i,j\\\phi\in C_i,\psi\in C_j\\\phi|_{V_i\cap V_j}\neq\psi|_{V_i\cap V_j}}}\pi(i,j)\hsq*{\Phi_{V_i,\phi}\Phi_{V_j,\psi}}&\lesssim2\sum_{i,x\in V_i}\pi'(i)\sum_a\hsq*{\Pi_a(\sigma_i(x))-\Pi_a(\sigma'(x))}\\
        &\leq 2L\sum_{i,x\in V_i}\frac{\pi'(i)}{|V_i|}\sum_a\hsq*{\Pi_a(\sigma_i(x))-\Pi_a(\sigma'(x))}\\
        &\lesssim 4L\sum_{i,x\in V_i}\frac{\pi'(i)}{|V_i|}\sum_{a}\Pi_a(\sigma_i(x))(1-\Pi_a(\sigma'(x))).
    \end{align*}
    Finally, reintroducing the projectors $\Phi_{V_i,\phi}$,
    \begin{align*}
        \sum_{\substack{i,j\\\phi\in C_i,\psi\in C_j\\\phi|_{V_i\cap V_j}\neq\psi|_{V_i\cap V_j}}}\pi(i,j)\hsq*{\Phi_{V_i,\phi}\Phi_{V_j,\psi}}&\lesssim4L\sum_{i}\frac{\pi'(i)}{|V_i|}\sum_{\substack{x\in V_i\\\phi\in C_i}}\sum_{a}\Phi_{V_i,\phi}\Pi_a(\sigma_i(x))(1-\Pi_a(\sigma'(x)))\\
        &=4L\sum_{i}\frac{\pi'(i)}{|V_i|}\sum_{\substack{x\in V_i\\\phi\in C_i}}\Phi_{V_i,\phi}(1-\Pi_{\phi(x)}(\sigma'(x)))\\
        &\lesssim4L\sum_{i}\frac{\pi'(i)}{|V_i|}\sum_{\substack{x\in V_i\\\phi\in C_i}}\hsq*{\Phi_{V_i,\phi}(1-\Pi_{\phi(x)}(\sigma'(x)))}.\qedhere
    \end{align*}
\end{proof}

\begin{lemma}\label{lem:cv-to-cc}
     Let $S=(X,\{(V_i,C_i)\}_{i=1}^m)$ be a $k$-ary CS and $\pi$ be a probability distribution on $[m]\times[m]$. There is a $P$-homomorphism $\beta:\mc{A}_{c-v}(S,\pi')\rightarrow\mc{A}_{c-c}(S,\pi)$ where $\pi'(i)=\sum_{j}\pi(i,j)$ and $P=\max_{i,j.\;V_i\cap V_j\neq\varnothing}\frac{\pi'(i)}{\pi(i,j)}$.
\end{lemma}

Note that, for an arbitrary $\pi$, $P$ can be arbitrarily large. In our context, this $C$-homomorphism will nevertheless be useful in the case of perfect completeness (defect $0$).

\begin{proof}
    For each $x \in X$ choose $i_x \in [m]$ such that $x \in V_{i_x}$. Let $\beta$ be the $\ast$-homomorphism defined by 
    \begin{align*}
        \beta(\sigma_i(x)) &= \sigma_i(x), \text{ and }
        \\
        \beta(\sigma'(x)) &= \sigma_{i_x}(x)
    \end{align*}
    for all $x \in X$. Then, for all $x \in X$ and $\phi \in C_i$ 
    \begin{align*}
        \beta\squ*{\hsq*{\Phi_{V_i,\phi}(1-\Pi_{\phi(x)}(\sigma'(x)))}} &= \hsq*{\Phi_{V_i,\phi}(1-\Pi_{\phi(x)}(\sigma_{i_x}(x)))}\lesssim\Phi_{V_i,\phi}(1-\Pi_{\phi(x)}(\sigma_{i_x}(x)))\\
        &= \sum_{\substack{\psi\in \Z_k^{V_{i_x}}\\\phi(x)\neq \psi(x)}}\Phi_{V_i,\phi}\Phi_{V_{i_x},\psi}\lesssim\sum_{\substack{\psi\in \Z_k^{V_{i_x}}\\\phi(x)\neq \psi(x)}}\hsq*{\Phi_{V_i,\phi}\Phi_{V_{i_x},\psi}}\\
        &\leq\sum_{\substack{\psi\in C_{i_x}\\\phi|_{V_i\cap V_{i_x}}\neq \psi|_{V_i\cap V_{i_x}}}}\hsq*{\Phi_{V_i,\phi}\Phi_{V_{i_x},\psi}}.
    \end{align*}
    Thus, we get
    \begin{align*}
        \beta\squ[\Big]{\sum_i \frac{\pi'(i)}{|V_i|}\sum_{\substack{x\in V_i\\\phi\in C_i}}\hsq*{\Phi_{V_i,\phi}(1-\pi_{\phi(x)}(\sigma_x(x)))}}&\lesssim\sum_i \frac{\pi'(i)}{|V_i|}\sum_{x\in V_i}\sum_{\substack{\phi\in C_i,\psi\in C_{i_x}\\\phi|_{V_i\cap V_{i_x}}\neq\psi|_{V_i\cap V_{i_x}}}}\hsq*{\Phi_{V_i,\phi}\Phi_{V_{i_x},\psi}}\\
        &\leq \sum_{i,j}\pi'(i)\sum_{\substack{\phi\in C_i,\psi\in C_{j}\\\phi|_{V_i\cap V_{j}}\neq\psi|_{V_i\cap V_{j}}}}\hsq*{\Phi_{V_i,\phi}\Phi_{V_{j},\psi}}\\
        &\leq P\sum_{i,j}\pi(i,j)\sum_{\substack{\phi\in C_i,\psi\in C_{j}\\\phi|_{V_i\cap V_{j}}\neq\psi|_{V_i\cap V_{j}}}}\hsq*{\Phi_{V_i,\phi}\Phi_{V_{j},\psi}}.\qedhere
    \end{align*}
\end{proof}

Now, we look at homomorphisms between constraint-variable and assignment algebras. In order to preserve the oracularisability, we first consider the variant of the assignment algebra with commutation constraints added to the defect. In one direction, we are able to find a $C$-homomorphism, and in the other we find a trace-dependent mapping.

\begin{lemma}\label{lem:acomm-to-cv}
    Let $S=(X,\{(V_i,C_i)\}_{i=1}^m)$ be a $k$-ary CS and $\pi$ be a probability distribution on $[m]$. There is a $20L^2$-homomorphism $\alpha:\mc{A}_{a+comm}(S,\pi)\rightarrow\mc{A}_{c-v}(S,\pi)$, where $L=\max_i|V_i|$.
\end{lemma}

\begin{proof}
    Let $\alpha$ be the inclusion of $\mc{A}_{a}(S)$ in $\mc{A}_{c-v}(S)$. Fix $i$ and write $V_i=\{x_1,\ldots,x_n\}$ ordered according to the ordering of $X$. Also, for $j=1,\ldots,n$, write $x_{j,a}=\Pi_{a}(\sigma'(x_j))$ and $\bar{x}_{j,a}=\Pi_{a}(\sigma_{i}(x_j))$. Then, as $\sigma_i(\Phi_{V_i,\phi})=0$ for all $\phi\notin C_i$,
    \begin{align*}
        \alpha\squ[\Big]{\sum_{\phi\notin C_i}\hsq*{\Phi_{V_i,\phi}}}&=\sum_{\phi\notin C_i}\hsq*{x_{1,\phi(x_1)}\cdots x_{n,\phi(x_n)}}=\sum_{\phi\notin C_i}\hsq*{x_{1,\phi(x_1)}\cdots x_{n,\phi(x_n)}-\bar{x}_{1,\phi(x_1)}\cdots\bar{x}_{n,\phi(x_n)}}\\
        &\leq\sum_{a_1,\ldots,a_n}\hsq*{x_{1,a_1}\cdots x_{n,a_n}-\bar{x}_{1,a_1}\cdots\bar{x}_{n,a_n}}\\
        &=\sum_{a_1,\ldots,a_n}\hsq[\Big]{\sum_{l=1}^nx_{1,a_1}\cdots x_{l-1,a_{l-1}}(x_{l,a_l}-\bar{x}_{l,a_l})\bar{x}_{l+1,a_{l+1}}\cdots\bar{x}_{n,a_n}}\\
        &\leq2^{\ceil{\log(n)}}\sum_{a_1,\ldots,a_n}\sum_{l=1}^n\hsq*{x_{1,a_1}\cdots x_{l-1,a_{l-1}}(x_{l,a_l}-\bar{x}_{l,a_l})\bar{x}_{l+1,a_{l+1}}\cdots\bar{x}_{n,a_n}}\\
        &\leq 2n\sum_{a}\sum_{j=1}^n\hsq*{x_{j,a}-\bar{x}_{j,a}}=2|V_i|\sum_{x\in V_i}\sum_a\hsq*{\Pi_a(\sigma'(x))-\Pi_a(\sigma_i(x))}.
    \end{align*}
    Also, we find that
    \begin{align*}
        \alpha\squ[\Big]{\sum_{x,y\in V_i;\;a,b\in\Z_k}\hsq*{[\Pi_a(x),\Pi_b(y)]}}&=\sum_{j,k=1,\ldots,n;\;a,b\in\Z_k}\hsq*{[x_{j,a},x_{k,b}]}=\sum_{j,k=1,\ldots,n;\;a,b\in\Z_k}\hsq*{[x_{j,a},x_{k,b}]-[\bar{x}_{j,a},\bar{x}_{k,b}]}\\
        &=\sum_{j,k=1,\ldots,n;\;a,b\in\Z_k}\hsq*{[x_{j,a}-\bar{x}_{j,a},x_{k,b}]-[\bar{x}_{j,a},\bar{x}_{k,b}-x_{k,b}]}\\
        &\leq2\sum_{j,k=1,\ldots,n;\;a,b\in\Z_k}\hsq*{[x_{j,a}-\bar{x}_{j,a},x_{k,b}]}+\hsq*{[\bar{x}_{j,a},\bar{x}_{k,b}-x_{k,b}]}\\
        &\lesssim4\sum_{j,k=1,\ldots,n;\;a,b\in\Z_k}\hsq*{(x_{j,a}-\bar{x}_{j,a})x_{k,b}}+\hsq*{\bar{x}_{j,a}(\bar{x}_{k,b}-x_{k,b})}\\
        &\leq8n\sum_{j=1,\ldots,n;\;a\in\Z_k}\hsq*{x_{j,a}-\bar{x}_{j,a}}=8|V_i|\sum_{x\in V_i}\sum_a\hsq*{\Pi_a(\sigma'(x))-\Pi_a(\sigma_i(x))}.
    \end{align*}
    
    By the proof of \Cref{lem:cc-to-cv}, we know $$\sum_{x\in V_i}\sum_a\hsq*{\Pi_a(\sigma(x))-\Pi_a(\sigma_i(x))}\lesssim 2\sum_{\substack{x\in V_i\\\phi\in C_i}}\hsq*{\Phi_{V_i,\phi}(1-\Pi_{\phi(x)}(\sigma'(x)))}.$$ Thus, we get that
    \begin{align*}
        \alpha\parens[\Big]{\sum_{r\in\mc{A}_a(S)}(\mu_{a,\pi}(r)+\mu_{comm,\pi}(r))r^\ast r}&\lesssim \sum_{i}\pi(i)20|V_i|\sum_{\substack{x\in V_i\\\phi\in C_i}}\hsq*{\Phi_{V_i,\phi}(1-\Pi_{\phi(x)}(\sigma'(x)))}\\
        &\leq20L^2\sum_{i}\frac{\pi(i)}{|V_i|}\sum_{\substack{x\in V_i\\\phi\in C_i}}\hsq*{\Phi_{V_i,\phi}(1-\Pi_{\phi(x)}(\sigma'(x)))}.\qedhere
    \end{align*}
\end{proof}

\begin{lemma}\label{lem:cv-to-acomm}
    Let $S=(X,\{(V_i,C_i)\}_{i=1}^m)$ and let $\pi$ be a probability distribution on $[m]$. Then, for every trace $\tau$ on $\mc{A}_{a+comm}(S,\pi)$, there exists a trace $\tau'$ on $\mc{A}_{c-v}(S,\pi)$ such that $\defect(\tau')\leq \poly(k^L)\defect(\tau)$, where $L=\max_i|V_i|$.
\end{lemma}

\begin{proof}
    Let $\tau=\rho\circ\varphi$ be the GNS representation, where $\varphi:\mc{A}_a(S)\rightarrow\mc{M}\subseteq\mc{B}(\mc{H})$ is a $\ast$\nobreakdash-homomorphism, and $\rho:\mc{M}\rightarrow\C$ is a tracial state. Let $\varepsilon_i=\sum_{x,y\in V_i}\sum_{a,b\in\Z_k}\norm{[\Pi_a(\sigma'(x),\Pi_b(\sigma'(y))]}_\tau^2$. We have $\defect(\tau)=\defect(\tau;\mu_{a,\pi})+\sum_i\pi(i)\varepsilon_i$. Let $u_{x,a}=\varphi(1-2\Pi_a(x))$. Fix some $i\in[m]$. We have in particular $\norm{[u_{x,a},u_{y,b}]}_\rho^2\leq\varepsilon_i$ for all $x,y\in V_i$ and $a,b\in\Z_k$. So, using \Cref{lem:z2stab}, there exist commuting order-$2$ unitaries $v_{x,a}\in\mc{M}$ such that $\norm{v_{x,a}-u_{x,a}}_\rho^2\leq\poly(k|V_i|)\varepsilon_i$. Next, note that
    \begin{align*}
        \norm[\Big]{\sum_{a\in\Z_k}\Pi_{1}(v_{x,a})-1}_\rho^2&=\frac{1}{4}\norm[\Big]{\sum_{a\in\Z_k}(u_{x,a}-v_{x,a})}_\rho^2\leq\frac{1}{4}2^{\ceil{\log_2k}}\sum_{a\in\Z_k}\norm{u_{x,a}-v_{x,a}}_\rho^2\\
        &\leq k^2\poly(k|V_i|)\varepsilon_i=\poly(|V_i|,k)\varepsilon_i.
    \end{align*}
    Let $p$ be the projector onto the kernel of $\sum_{a\in\Z_k}\Pi_{1}(v_{x,a})-1$. Then, as the eigenvalues of $\sum_{a\in\Z_k}\Pi_{1}(v_{x,a})-1$ are integers, $\norm{1-p}_\rho^2\leq\poly(|V_i|,k)\varepsilon_i$. Take $w_{i,x}=\sum_{a\in\Z_k}\omega_k^a\Pi_{1}(v_{x,a})p+(1-p)$. As the supports of the $\Pi_{1}(v_{x,a}) p=\Pi_{1}(v_{x,a})\land p$ are disjoint, $w_{i,x}$ is an order-$k$ unitary. Further,
    \begin{align*}
        \norm{\Pi_a(w_{i,x})-\Pi_a(\varphi(x))}_\rho^2&\leq 4\parens*{\norm{\Pi_{1}(v_{x,a})(p-1)}_\rho^2+\norm{1-p}_\rho^2+\norm{\Pi_{1}(v_{x,a})-\Pi_{1}(u_{x,a})}_\rho^2}\\
        &\leq4\parens*{2\norm{1-p}_\rho^2+\tfrac{1}{4}\norm{v_{x,a}-u_{x,a}}_\rho^2}\leq\poly(|V_i|,k)\varepsilon_i.
    \end{align*}
    Now, take the PVM $\{P_{V_i,\phi}\}_{\phi\in\Z_k^{V_i}}$ as $P_{V_i,\phi}=\prod_{x\in V_i}\Pi_{\phi(x)}(w_{i,x})$. Fix some $\phi_i\in C_i$ and take the PVM $\{\tilde{P}_{V_i,\phi}\}_{\phi\in C_i}$ as $\tilde{P}_{V_i,\phi}=P_{V_i,\phi}+\delta_{\phi,\phi_i}\sum_{\psi\notin C_i}P_{V_i,\psi}$. Now, we can take the $\ast$-representation $\chi:\mc{A}_{c-v}(S,\pi)\rightarrow\mc{M}$ defined by $\chi(\sigma'(x))=\varphi(x)$ and $\chi(\Phi_{V_i,\phi})=\tilde{P}_{V_i,\phi}$; and take the tracial state $\tau'=\rho\circ\chi$. It remains to calculate the defect of $\tau'$. First, note that
    \begin{align*}
        \defect(\tau')&=\sum_{i=1}^m\frac{\pi(i)}{|V_i|}\sum_{x\in V_i,\phi\in C_i}\tau'(\Phi_{V_i,\phi}(1-\Pi_{\phi(x)}(\sigma'(x))))=\sum_{i=1}^m\frac{\pi(i)}{|V_i|}\sum_{x\in V_i,\phi\in C_i}\rho(\tilde{P}_{V_i,\phi}\Pi_{\phi(x)}(\varphi(x)))\\
        &=\sum_{i=1}^m\frac{\pi(i)}{|V_i|}\parens[\Bigg]{\sum_{x\in V_i,\phi\in\Z_k^{V_i}}\rho(P_{V_i,\phi}\Pi_{\phi(x)}(\varphi(x)))+\sum_{x\in V_i,\phi\notin C_i}\rho(P_{V_i,\phi}(\Pi_{\phi(x)}(\varphi(x))+\Pi_{\phi_i(x)}(\varphi(x))))}\\
        &\leq\sum_{i=1}^m\frac{\pi(i)}{|V_i|}\parens[\Bigg]{\frac{1}{4}\sum_{x\in V_i,a\in\Z_k}\norm{\Pi_a(w_{i,x})-\Pi_{a}(\varphi(x))}_\rho^2+2|V_i|\sum_{\phi\notin C_i}\norm{P_{V_i,\phi}}_\rho^2}.
    \end{align*}
    For fixed $i$ and $\phi\in\Z_k^{V_i}$, write $V_i=\{x_1,\ldots,x_{|V_i|}\}$, $p_{j}=\Pi_{\phi(x_j)}(x_j)$, and $\bar{p}_{j}=\Pi_{\phi(x_j)}(w_{i,x_j})$. So we have that
    \begin{align*}
        \norm{P_{V_i,\phi}}_\rho^2&=\norm{\bar{p}_1\cdots\bar{p}_{|V_i|}}_\rho^2\\
        &=\norm[\Big]{p_1\cdots p_{|V_i|}+\sum_{j=1}^{|V_i|}p_1\cdots p_{j-1}(\bar{p}_j-p_j)\bar{p}_{j+1}\cdots\bar{p}_{|V_i|}}_\rho^2\\
        &\leq 2^{\ceil{\log_2(|V_i|+1)}}\parens[\Bigg]{\norm{p_1\cdots p_{|V_i|}}_\rho^2+\sum_{j=1}^{|V_i|}\norm{\bar{p}_j-p_j}_\rho^2}\\
        &\leq 2(|V_i|+1)\parens[\Bigg]{\norm{\Phi_{V_i,\phi}}_\tau^2+\sum_{x\in V_i}\norm{\Pi_{\phi(x)}(w_{i,x})-\Pi_{\phi(x)}(\varphi(x))}_\rho^2}.
    \end{align*}
    As such, we get that the defect
    \begin{align*}
        \defect(\tau')&\leq\sum_{i=1}^m\frac{\pi(i)}{|V_i|}\Bigg(\frac{1}{4}\sum_{x\in V_i,a\in\Z_k}\norm{\Pi_a(w_{i,x})-\Pi_{a}(\varphi(x))}_\rho^2+4|V_i|(|V_i|+1)\sum_{\phi\notin C_i}\norm{\Phi_{V_i,\phi}}_\tau^2\\
        &\qquad\qquad+4|V_i|(|V_i|+1)\sum_{\phi\notin C_i}\sum_{x\in V_i}\norm{\Pi_{\phi(x)}(w_{i,x})-\Pi_{\phi(x)}(\varphi(x))}_\rho^2\Bigg)\\
        &\leq 4(L+1)\defect(\tau;\mu_{a,\pi})+\sum_{i=1}^m\pi(i)\parens*{\frac{1}{4|V_i|}+4(|V_i|+1)k^{|V_i|-1}}\sum_{x\in V_i,a\in\Z_k}\norm{\Pi_a(w_{i,x})-\Pi_{a}(\varphi(x))}_\rho^2\\
        &\leq 4(L+1)\defect(\tau;\mu_{a,\pi})+\sum_{i=1}^m\pi(i)\parens*{\frac{1}{4|V_i|}+4(|V_i|+1)k^{|V_i|-1}}k|V_i|\poly(|V_i|,k)\varepsilon_i\\
        &\leq\poly(k^L)\defect(\tau).\qedhere
    \end{align*}
\end{proof}

Finally, we consider the relationship between the assignment algebras with and without commutation. Here, we are only able to find a $C$-homomorphism in one direction, as a $C$-homomorphism in the other direction would imply oracularisability of the underlying assignment algebra, which does not hold for general CSs.

\begin{lemma}\label{lem:a-to-acomm}
    Let $S=(X,\{(V_i,C_i)\}_{i=1}^m)$ be a $k$-ary CS and $\pi$ be a probability distribution on $[m]$. There is a $1$-homomorphism $\alpha:\mc{A}_{a}(S,\pi)\rightarrow\mc{A}_{a+comm}(S,\pi)$.
\end{lemma}

\begin{proof}
    Let $\alpha$ be the identity map. Noting that $\mu_{a,\pi}+\mu_{comm,\pi}\geq\mu_{a,\pi}$ gives the result.
\end{proof}

\section{Constraint system languages}\label{sec:csp}

\subsection{Constraint satisfaction problems}

Informally, a constraint satisfaction problem is a collection of constraint systems that are constructed by filling a fixed set of constraints with variables arbitrarily. We formalise this with the notion of a pushforward constraint.

\begin{definition}
    Let $C\subseteq\Sigma^V$ and let $r:V\rightarrow W$. The \textbf{pushforward constraint of $C$ by $r$} is $r_\ast C=\set*{\phi\in\Sigma^W}{\phi\circ r\in C}$.
\end{definition}

\begin{definition}
    Let $\Gamma$ be a finite set of constraints over an alphabet $\Sigma$. The \textbf{constraint satisfaction problem (CSP)} of $\Gamma$ is the set of constraint systems
    $$\CSP(\Gamma)=\set*{(X,\{(W_i,{r_i}_\ast C_i)\}_{i=1}^m)}{W_i\subseteq X,\;(V_i,C_i)\in\Gamma,\;r_i:V_i\rightarrow W_i}.$$

    If every element of $\CSP(\Gamma)$ is a $2$-CS, we say it is a \textbf{$2$-CSP}. We also abuse notation and say $\Gamma$ is a $2$-CSP.
\end{definition}

\begin{example}\hphantom{}
    \begin{itemize}
        \item 3SAT is the CSP where $\Sigma=\set*{\Z_2^3\backslash\{t\}}{t\in\Z_2^3}$.
        \item $k$-colouring is the CSP where $\Sigma=\Z_k$ and $\Gamma=\{\neq_{\Z_k}\}$, for $\neq_{\Z_k}=\set*{(a,b)\in\Z_k^2}{a\neq b}$. This is a $2$-CSP.
    \end{itemize}
\end{example}

\subsection{Classical CSPs}

A natural complexity problem given a CSP is to decide the satisfiability of the CS instances. Much is known about the classical version of this problem, which we outline here.

\begin{definition}[Constraint languages]
    Let $\Gamma$ be a set of constraints over an alphabet $\Sigma$, and let $1\geq c\geq s\geq 0$.
    \begin{itemize}
        \item $\CSP(\Gamma)_{c,s}$ is the promise problem with instances that are CSs $S=(X,\{(V_i,C_i)\}_{i=1}^m)\in\CSP(\Gamma)$, where $S$ is a yes instance if there is an assignment $f:X\rightarrow\Sigma$ such that $f|_{V_i}\in C_i$ for at least $mc$ values of $i$, and $S$ is a no instance if, for every assignment $f$, $f|_{V_i}\in C_i$ for strictly less than $ms$ values of $i$.

        \item $\SuccinctCSP(\Gamma)_{c,s}$ is the promise problem with instances that are probabilistic Turing machines $M$ that sample the constraints of a CS $S=(X,\{(V_i,C_i)\}_{i=1}^m)\in\CSP(\Gamma)$ according to some probability distribution $\pi:[m]\rightarrow[0,1]$, where $M$ is a yes instance if there is an assignment $f:X\rightarrow\Sigma$ such that $\Pr_{i\leftarrow\pi}[f|_{V_i}\in C_i]\geq c$, and $M$ is a no instance if for every assignment $\Pr_{i\leftarrow\pi}[f|_{V_i}\in C_i]<s$.
    \end{itemize}
\end{definition}

\begin{definition}
    A mapping $f:\Sigma^k\rightarrow\Sigma$ induces a \textbf{polymorphism} $(\Sigma^V)^k\rightarrow\Sigma^V$ as
    $$f(\phi_1,\ldots,\phi_n)(v)=f(\phi_1(v),\ldots,\phi_n(v)).$$
    The polymorphism $f$ \textbf{preserves} $(V,C)$ if for all $\phi_1,\ldots,\phi_n\in C$, $f(\phi_1,\ldots,\phi_n)\in C$. For a set of constraints $\Gamma$, we say $f$ is a \textbf{$\Gamma$-homomorphism} if it preserves every constraint in $\Gamma$.
\end{definition}

\begin{definition}
    A map $f:\Sigma^k\rightarrow\Sigma$ is a \textbf{weak near-unanimity} if for all $a,b\in\Sigma$, $f(b,a,\ldots,a)=f(a,b,a,\ldots,a)=\cdots=f(a,\ldots,a,b)$.
\end{definition}

\begin{theorem}[\cite{Zhu17,Bul17}]
    For all finite alphabets $\Sigma$, $\CSP(\Gamma)_{1,1}$ is $\NP$-complete if there is no weak near-unanimity $\Gamma$-homomorphism; otherwise, $\CSP(\Gamma)_{1,1}\in\cP$.
\end{theorem}

\begin{corollary}
    If $\CSP(\Gamma)_{1,1}$ is $\NP$-complete, then there is a constant $s\in[0,1)$ such that $\SuccinctCSP(\Gamma)_{1,s}$ is $\NEXP$-complete. Otherwise, $\SuccinctCSP(\Gamma)_{1,s}\in\EXP$ for all $s$.
\end{corollary}

\begin{proof}[Proof sketch]
    Using the CSP version of the complexity class $\PCP[m,q]=\NEXP$ for $m=\exp(n)$ and $q=O(1)$, we know that $\SuccinctCSP(\Gamma_q)_{1,1/2}=\NEXP$, where $\Gamma_q$ is the set of all constraints on $q$ boolean variables~\cite{BFL91,Vid14}. From the proof of hardness part of the CSP dichotomy conjecture~\cite{Bul17,Zhu17}, we see that reduction of $\CSP(\Gamma_q)_{1,1}$ to $\CSP(\Gamma)_{1,1}$ is done constraint by constraint, by replacing each constraint by a conjunction of constraints from~$\Gamma$. As there is a finite number of constraints in $\Gamma_q$, there exists $M\in\N$ such that every constraint can be expressed by at most $M$ constraints from~$\Gamma$. We can apply the same reduction to any instance of $S\in\SuccinctCSP(\Gamma_q)_{1,1/2}$ in polynomial time to get a succinct description of an $S'\in\CSP(\Gamma)$. Note first that, if $S$ is a yes instance, all the constraints are satisfied, and therefore all the constraints of $S'$ are satisfied. Next, if $S$ is a no instance, suppose $\geq 1-\frac{1}{2M}$ constraints are satisfied by an assignment to $S'$. Then, $\leq \frac{1}{2M}$ constraints of $S'$ are not satisfied. Since each constraint of $S$ is mapped to a conjunction of at most $M$ constraints in $S'$, this means there is an assignment that satisfies $\geq\frac{1}{2}$ constraints of $S$, a contradiction. Hence, taking $s=1-\frac{1}{2M}$, if $S$ is a no instance of $\SuccinctCSP(\Gamma_q)_{1,1/2}$, $S'$ is a no instance of $\SuccinctCSP(\Gamma)_{1,s}$, as wanted.

    Conversely, if $\CSP(\Gamma)_{1,1}$ is not $\NP$-complete, by the CSP dichotomy theorem~\cite{Bul17,Zhu17} there is a polynomial-time algorithm for it. Hence, this translates to an exponential-time algorithm for $\SuccinctCSP(\Gamma)_{1,1}$. As every yes (\textbf{cf.} no) instance of $\SuccinctCSP(\Gamma)_{1,s}$ is a yes (\textbf{cf.} no) instance of $\SuccinctCSP(\Gamma)_{1,1}$, the algorithm also decides $\SuccinctCSP(\Gamma)_{1,s}$ in exponential time. Hence, $\SuccinctCSP(\Gamma)_{1,s}\in\EXP$.
\end{proof}

To finish this section, note that in the boolean case $\Sigma=\Z_2$, the set of weak near-unanimity polymorphisms is well known.

\begin{definition}
The following are the weak near unanimity polymorphisms for boolean constraint systems.
    \begin{itemize}
        \item The \textbf{constant $0$ polymorphism} is $0:\Z_2^0\rightarrow\Z_2$ with output $0$.
        \item The \textbf{constant $1$ polymorphism} is $1:\Z_2^0\rightarrow\Z_2$ with output $1$.
        \item The \textbf{AND polymorphism} is $\mathrm{AND}:\Z_2^2\rightarrow\Z_2$, $\mathrm{AND}(a,b)=a\land b$.
        \item The \textbf{OR polymorphism} is $\mathrm{OR}:\Z_2^2\rightarrow\Z_2$, $\mathrm{OR}(a,b)=a\lor b$.
        \item The \textbf{majority polymorphism} is $\mathrm{MAJ}:\Z_2^3\rightarrow\Z_2$, $\mathrm{MAJ}(a,b,c)=(a\land b)\lor(b\land c)\lor(c\land a)$.
        \item The \textbf{minority polymorphism} is $\mathrm{MIN}:\Z_2^3\rightarrow\Z_2$, $\mathrm{MIN}(a,b,c)=a\oplus b\oplus c$.
    \end{itemize}
\end{definition}

\begin{theorem}[Schaefer's dichotomy theorem~\cite{Sch78}]
    For $\Gamma$ a set of constraints over $\Z_2$, $\CSP(\Gamma)_{1,1}\in\cP$ if one of the polymorphisms $0,1,\mathrm{AND},\mathrm{OR},\mathrm{MAJ},
    \mathrm{MIN}$ is a $\Gamma$-homomorphism; else $\CSP(\Gamma)_{1,1}$ is $\NP$-complete.
\end{theorem}

\subsection{Quantum CSPs}

The quantum satisfiability problem for CSPs can be equivalently phrased in terms of algebra representations or nonlocal games.

\begin{definition}[Entangled constraint systems]
    Let $\Gamma$ be a set of constraints over $\Z_k$, let $w\in\{c-c,c-v,a,a+comm\}$, and let $1\geq c\geq s\geq 0$.
    \begin{itemize}
        \item $\CSP_w(\Gamma)_{c,s}^\ast$ is the promise problem with instances that are CSs $S=(X,\{(V_i,C_i)\}_{i=1}^m)\in\CSP(\Gamma)$, where $S$ is a yes instance if $$\inf_\tau\defect(\tau)\leq 1-c,$$ where the infimum is over all finite-dimensional traces $\tau$ on $\mc{A}_{w}(S,\mbb{u}_m)$, where $\mbb{u}_m$ is the uniform distribution on $[m]$; and $S$ is a no instance if $$\inf_\tau\defect(\tau)>1-s.$$

        \item $\SuccinctCSP_w(\Gamma)_{c,s}^\ast$ is the promise problem with instances that are probabilistic Turing machines $M$ that sample the constraints of a CS $S=(X,\{(V_i,C_i)\}_{i=1}^m)\in\CSP(\Gamma)$ according to some probability distribution $\pi:[m]\rightarrow[0,1]$, where $S$ is a yes instance if $$\inf_\tau\defect(\tau)\leq 1-c,$$ where the infimum is over all finite-dimensional traces $\tau$ on $\mc{A}_{w}(S,\pi)$; and $S$ is a no instance if $$\inf_\tau\defect(\tau)>1-s.$$
    \end{itemize}
\end{definition}

We drop the subscripts if they are not important. Note that we can always replace the infima over finite-dimensional traces by Connes-embeddable traces.

In this definition, the yes instances correspond exactly to CSs where the quantum synchronous value of the corresponding CS game, as defined in \Cref{sec:csp-defs}, is at least $c$, and the no instances correspond to CSs whose quantum synchronous value is less than $s$. If $w=c-c$, the corresponding game is the constraint-constraint CS game; if $w=c-v$, the corresponding game is the constraint variable CS game; and if $w=a$, there is only an associated game if $\Gamma$ is a $2$-CSP, in which case it is the $2$-CS game.

We make use the following property of CSs, which precludes the construction of simple commutation gadgets.

\begin{definition}\label{def:tvf}
    Let $C\subseteq\Sigma^V$ be a constraint over an alphabet $\Sigma$. $C$ is \textbf{two-variable falsifiable (TVF)} if for all $x\neq y\in V$, there exist $a,b\in\Sigma$ such that $\phi\notin C$ if $\phi(x)=a$ and $\phi(y)=b$.

    We say a CS $S$ or a set of constraints $\Gamma$ is TVF if every constraint in it is TVF.
\end{definition}

Two-variable falsifiability can seem like a very strong restriction on the constraints. However, the example of 1-in-3-SAT, brought to our attention by Alex Meiburg, shows that there are nontrivial CSPs that are TVF.

\begin{example} Consider the boolean constraint $$C=\set*{(0,0,1),(0,1,0),(1,0,0)}\subseteq\Z_2^3.$$ This is TVF as setting any two variables to $1$ makes the constraint unsatisfied. On the other hand, the language $\CSP(\{C\})_{1,1}$ is $\NP$-complete as a direct consequence of Schaefer's dichotomy theorem. 
\end{example}

\begin{theorem}[Main theorem]\label{thm:main-theorem}
    Let $\Gamma$ be a set of constraints over $\Z_k$ such that $\CSP(\Gamma)_{1,1}$ is $\NP$-complete, and one of the following holds:
    \begin{enumerate}
        \item $\Gamma$ is not TVF,
        \item $\Gamma$ is boolean, or
        \item $\Gamma=\{\neq_{\Z_3}\}$ is $3$-colouring.
    \end{enumerate}
    Then, there exists a constant $s\in[0,1)$ such that $\SuccinctCSP_{c-v}(\Gamma)^\ast_{1,s}$ is $\RE$-complete.
\end{theorem}

It remains open whether non-boolean TVF CSPs, except for $3$-colouring, are also $\RE$-complete. As an important example, we do not know if $k$-colouring for $k\geq 4$ is $\RE$-complete with entanglement. However, we do know that any $2$-CSP that contains an empty constraint (where all assignments are accepted), such as the language of all $2$-CSPs over $\Z_k$, is $\RE$-complete as it is not TVF.

For $\RE$-complete entangled CSPs, it also holds that $\SuccinctCSP_{c-c}(\Gamma)^\ast_{1,s}$ is $\RE$-complete, which we show using a mapping between the constraint-variable and constraint-constraint algebras given in \Cref{prop:cc-to-cv-complexity}. This is in contrast with \Cref{lem:cv-to-cc}, where the defect can scale very badly when mapping to the constraint-constraint algebra.

\begin{corollary}\label{cor:main-theorem-1}
    Let $\Gamma$ be a set of constraints satisfying the conditions of \Cref{thm:main-theorem}. Then, there exists a constant $s\in[0,1)$ such that $\SuccinctCSP_{c-c}(\Gamma)^\ast_{1,s}$ is $\RE$-complete.
\end{corollary}

For some CSPs, we can also show a similar $\RE$-hardness result for the assignment algebra via an oracularisability property. For $2$-CSPs, this corresponds to the associated $2$-CS game.

\begin{corollary}\label{cor:main-theorem-2}
    Let $\Gamma$ be a set of constraints over $\Z_k$ such that $\CSP(\Gamma)_{1,1}$ is $\NP$-complete, and $\Gamma$ is boolean TVF, $\Gamma=\{\neq_{\Z_3}\}$, or $\Gamma=\set*{C\subseteq\Z_k^2}$. Then, there exists a constant $s\in[0,1)$ such that $\SuccinctCSP_{a}(\Gamma)^\ast_{1,s}$ is $\RE$-complete.
\end{corollary}

More generally, the hardness extends to the assignment algebra with commutation.

\begin{corollary}\label{cor:main-theorem-3}
    Let $\Gamma$ be a set of constraints satisfying the conditions of \Cref{thm:main-theorem}. Then, there exists a constant $s\in[0,1)$ such that $\SuccinctCSP_{a+comm}(\Gamma)^\ast_{1,s}$ is $\RE$-complete.
\end{corollary}

This follows immediately from the mappings between the constraint-variable and assignment with commutation algebras given in \Cref{lem:cv-to-acomm,lem:acomm-to-cv}. By these, the values of the defects for these algebras are equal up to constant factors (for constraint systems with constant-size constraints and alphabets) and therefore hardness of deciding one directly implies hardness of the other.

Next, as noted in the introduction, we can extend the undecidability of entangled CSP languages to non-succinct languages, by considering computable reductions that are not polynomial-time.

\begin{corollary}\label{cor:main-theorem-4}
    Let $\Gamma$ be a set of constraints over $\Z_k$ such that $\CSP(\Gamma)_{1,1}$ is $\NP$-complete, and $\Gamma$ satisfies the conditions of \Cref{thm:main-theorem}. Then, there exists a constant $s\in[0,1)$ such that $\CSP_{c-v}(\Gamma)^\ast_{1,s}$, $\CSP_{c-c}(\Gamma)^\ast_{1,s}$, and $\CSP_{a+comm}(\Gamma)^\ast_{1,s}$ are $\RE$-complete with respect to exponential-time reductions. If $\Gamma$ additionally satisfies the conditions of \Cref{cor:main-theorem-2}, $\CSP_{a}(\Gamma)^\ast_{1,s}$ is $\RE$-complete with respect to exponential-time reductions.
\end{corollary}

The main theorem and earlier corollaries give a polynomial-time mapping from instances of the halting problem to instances of $\SuccinctCSP_w(\Gamma)_{1,s}^\ast$, preserving yes and no instances. In exponential time, the whole constraint-sampling algorithm can be described, and thus gives an exponential-time mapping to $\CSP_w(\Gamma)_{1,s}^\ast$. The only remaining issue to check is that the probability distribution can be made uniform while preserving the constant gap. First, in the constraint-variable game, the probability can be perturbed by mixing with a uniform distribution to ensure that the probability of any constraint is at least $\alpha/m$ for some small constant $\alpha\in(0,1)$, while preserving constant gap. Next, each constraint $i$ can be repeated with multiplicity $\floor{\pi(i)m/\alpha}$. These constraints are sampled from a uniform distribution, which preserves the constant soundness. Note that the defect cannot be made smaller by making the trace have different values on the different copies: there exists a trace on the original algebra corresponding to the minimal value of the trace over all copies. Also, this reduction to uniform distribution requires the number of constraints to be polynomial in order to be efficient, as it may be easy to sample from a distribution while being hard to compute the probabilities. An analogous argument applies for the other algebra models.

In \cite{MS24}, the authors use a modified version of a protocol due to Dwork, Feige, Kilian, Naor, Safra \cite{Dwork1992LowC2} to show that $\MIP^*$ admits two-prover one-round perfect zero knowledge proof systems with polynomial question and answer length. Since the zero knowledge tableau proof system from \cite{MS24} is a non-TVF $\CSP$ game, and 3SAT reduces to it, our main theorem has the following immediate corollary.
\begin{corollary}\label{cor:main-theorem-5}
    There is a perfect zero knowledge $\CS$-$\MIP^*(2,1,1,s)$ protocol for the halting problem in which the soundness parameter $s$ is constant, the questions have length $\poly(n)$ and the answers have constant length. Furthermore, if a game in the protocol has a perfect strategy, then it has a perfect synchronous quantum strategy. 
\end{corollary}

The following sections are devoted to the proof of the main theorem. First, in \Cref{sec:subdivision}, we improve the subdivision technique from~\cite{MS24}, which will allow us to ask constraints from a CSP as separate questions, rather than at the same time, while preserving constant gap. Then, we split the proof of \Cref{thm:main-theorem} into three parts. In \Cref{sec:non-TVF}, we prove \Cref{thm:main-theorem} for the case of non-TVF CSPs. Next, in \Cref{sec:TVF}, we provide the proof of \Cref{thm:main-theorem} for the case of boolean TVF CSPs. We also show \Cref{cor:main-theorem-2} for these CSPs there. In \Cref{sec:2-csps}, we prove \Cref{thm:main-theorem} for $3$-colouring, and \Cref{cor:main-theorem-2} for $3$-colouring and the language of all $2$-CSPs. Finally, in \Cref{sec:cv-to-cc}, we prove \Cref{cor:main-theorem-1}.

\section{Improved subdivision}\label{sec:subdivision}

Suppose we have a BCS where each constraint is a conjunction of subconstraints on
subsets of the variables (for instance, a 3SAT instance made up of 3SAT
clauses). In \cite{MS24}, the authors define a BCS transformation called subdivision that splits up the
contexts and constraints so that each subconstraint is in its own context. In this section we prove an improved version of Theorem 7.5 of \cite{MS24} which allows us to preserve a constant soundness gap in our reduction. 

Splitting up a context in the weighted algebra formalism changes the commutative
subalgebra corresponding to the context to a non-commutative subalgebra. The authors of \cite{MS24} deal with this using the stability of $\Z_2^k$, which we recall in \Cref{sec:vna}.

We now recall the definition of subdivision from \cite{MS24}.

\begin{definition}\label{def:subdivision}
    Let $B = (X,\{(V_i,C_i)\}_{i=1}^m)$ be a $\BCS$. Suppose that for all $1\leq i\leq m$ there exists a constant $m_i\geq1$ and a set of constraints $\{D_{ij}\}_{j=1}^{m_i}$ on variables $\{V_{ij}\}_{j=1}^{m_i}$ respectively, such that
    \begin{enumerate}
        \item $V_{ij} \subseteq V_i$ for all $i \in [m]$ and $j \in [m_i]$,
        \item for every $i\in [m]$ and $x,y \in V_i$, there is a $j \in [m_i]$ such that $x,y \in V_{ij}$, and
        \item $C_i = \AND_{j=1}^{m_i}D_{ij}$ for all $i \in [m]$, where $\AND$ is conjunction.
    \end{enumerate}
    The BCS $B' = (X,\{(V_{ij},D_{ij})\}_{i,j})$ is called a \textbf{subdivision} of $B$. When working with subdivisions, we refer to $D_{ij}$ as the \textbf{clauses} of constraint $C_i$, and we refer to $m_i$ as the \textbf{number of clauses} in constraint $i$. 
\end{definition}

Given a subdivision $B'$ of $B$ with $M = \sum_{i=1}^m m_i$, we pick a bijection between $[M]$ and the set of pairs $(i,j)$ with $1 \leq i\leq m$ and $1\leq j \leq m_i$. If $\pi$ is a probability distribution on $[M]\times [M]$ with $\pi_{sub}(ij,kl) = \pi(i,k)/m_im_k$. It is a result of \cite{MS24} that if $\ttt{G}(B,\pi)$ has a perfect quantum (resp. commuting operator) strategy if and only if $\ttt{G}(B',\pi_{sub})$ has a perfect quantum (resp. commuting operator) strategy. Theorem 7.5 of \cite{MS24} states that near perfect strategies for $\ttt{G}(B',\pi_{sub})$ can be pulled back to near perfect strategies for $\ttt{G}(B,\pi)$. The main result of this section is a new version of this result with an improved bound.

\begin{theorem}\label{thm:subdiv}
    Let $B = (X,\{(V_i,C_i)\}_{i=1}^m)$ be a $\BCS$ and let $B' = (X,\{(V_{ij},D_{ij})\}_{i,j})$ be a subdivision of $B$ with $m_i$ clauses in constraint $C_i$. Let $\pi$ be a probability distribution on $[m]\times[m]$ that is $C$-diagonally dominant for some $C>0$, and let $\pi_{sub}$ be the probability distribution defined from $\pi$ as above. If there is a trace $\tau$ on $\mcA_{c-c}(B',\pi_{sub})$, then there is a trace $\wtd{\tau}$ on $\mcA_{c-c}(B,\pi)$ with $\df(\wtd{\tau})\leq \poly(2^L,M,K)\df(\tau)$, where $L = \max_{i,j}|V_{ij}|$, $K = \max_i|V_i|$, and $M = \max m_i$.
\end{theorem}

To prove the theorem, we follow \cite{MS24} by considering other versions of the weighted BCS algebra, where $\mcA(V_i,C_i)$ is replaced by $\C\Z_2^{*V_i}$, and the defining relations of $\mcA(V_i,C_i)$ are moved into the weight function.

\begin{defn}\label{def:freeBCS}
    Let $B = \left(X,\{(V_i,C_i)\}_{i=1}^m\right)$ be a BCS with a probability
    distribution $\pi$ on $[m] \times [m]$, and let $B' =
    \left(X,\{(V_{ij},D_{ij})\}_{i,j}\right)$ be a subdivision, with $m_i$ clauses
    in constraint $C_i$ and probability
    distribution $\pi_{sub}$ induced by $\pi$.
    Let $\sigma_i : \C \Z_2^{*V_i} \to *_{i=1}^m \C \Z_2^{*V_i}$ denote the
    inclusion of the $i$th factor. Let $\mcA_{free}(B) := *_{i=1}^m \C \Z_2^{*V_i}$, and
    define weight functions $\mu_{inter}$, $\mu_{sat}$, $\mu_{clause}$, and $\mu_{comm}$ on $\mcA_{free}(B)$ by
    \begin{align*}
        & \mu_{inter}(\sigma_i(x) - \sigma_j(x)) = \pi(i,j) \text{ for all } i \neq j \in [m] \text{ and } x \in V_i \cap V_j, 
            \\
        & \mu_{sat}(\Phi_{V_i,\phi}) = \pi(i,i) \text{ for all } i \in [m] \text{ and } \phi \in \Z_2^{V_{i}} \setminus C_{i},
        	\\
        & \mu_{clause}(\Phi_{V_{ij},\phi}) = \pi(i,i)/m_i^2 \text{ for all } (i,j) \in [m]\times[m_i] \text{ and } \phi \in \Z_2^{V_{ij}} \setminus D_{ij}, 
            \text{ and }\\ 
        & \mu_{comm}([\sigma_i(x),\sigma_i(y)]) = \pi(i,i) \text{ for all } i \in [m] \text{ and } x,y \in V_i,
    \end{align*}
    and $\mu_{inter}(r) = 0$, $\mu_{sat}(r)= 0$, $\mu_{clause}(r)=0$, and $\mu_{comm}(r)=0$ for any
    elements $r$ other than those listed. 
    Let $\mcA_{free}(B,B',\pi)$ be the weighted algebra $(\mcA_{free}(B), \mu_{all})$,
    where $\mu_{all} := \mu_{inter} + \mu_{clause} + \mu_{comm}$.
\end{defn}
Note that $\mu_{inter}$ is the same as the weight function of the algebra
$\mcA_{inter}(B,\pi)$ defined in \Cref{def:alg}, except that it's defined on
$\mcA_{free}(B)$ rather than $\mcA_{c-c}(B)$. The weight function $\mu_{sat}$ comes
from the defining relations for $\mcA_{c-c}(B)$, while $\mu_{clause}$ comes from the
defining relations for $\mcA_{c-c}(B')$, so $\mcA_{free}(B,B',\pi)$ is a mix of relations
from $\mcA_{inter}(B,\pi)$ and $\mcA_{inter}(B',\pi)$. As mentioned previously,
the context $V_i$ has an order inherited from $X$, and this is used for the
order of the product when talking about $\Phi_{V_i,\phi}$ and
$\Phi_{V_{ij},\phi}$ in $\mcA_{free}(B)$. In particular, the order on $V_{ij}$
is compatible with the order on $V_{i}$.  

The weight functions $\mu_{inter}$, $\mu_{sat}$ and $\mu_{clause}$ can also be
defined on $\ast_{i=1}^m \C \Z_2^{V_i}$ using the same formulae as in
\Cref{def:freeBCS}, and we use the same notation for both versions. Lemma 7.7 from \cite{MS24} shows that we can relax $\mcA_{inter}(B,\pi)$ to $(\ast_{i=1}^m
\C \Z_{2}^{V_i}, \mu_{inter} + \mu_{clause})$ when $\pi$ is maximized on the diagonal. A probability distribution is maximized on the diagonal when $\pi(i,i) \geq \pi(i,j)$ and $\pi(i,i) \geq \pi(j,i)$. We use a stronger condition on $\pi$, namely that it is $C$-diagonally dominant, and prove a version of the lemma that removes the dependence on the connectivity of the BCS at the expense of a dependence on the maximum clause size.

\begin{lemma}\label{lem:correctsat}
    Let $B = \left(X,\{(V_i,C_i)\}_{i=1}^m\right)$ be a $\BCS$, and let $\pi$ be a
    probability distribution on $[m]\times[m]$ which is $C$-diagonally dominant for some $C>0$. Let $B' = \left(X,\{(V_{ij},D_{ij})\}_{i,j}\right)$ be a subdivision of $B$.
    Let $\mu_{inter}$, $\mu_{sat}$ and $\mu_{clause}$ be the weight functions defined above with respect to $\pi$. Then there is an
    $O(K)$-homomorphism $\mcA_{inter}(B,\pi) \to (\ast_{i=1}^m \C \Z_2^{V_i}, \mu_{inter}+\mu_{sat})$,
    where $K = \max_{i}|V_i|$. 
    Furthermore, there is an $M^2$-homomorphism $(\ast_{i=1}^m \C \Z_2^{V_i}, \mu_{inter}+\mu_{sat}) \to (\ast_{i=1}^m \C \Z_2^{V_i}, \mu_{inter}+\mu_{clause})$, where $M = \max_i m_i$ is the maximum number of clauses
    $m_i$ in constraint $i$.
\end{lemma}
\begin{proof}
    Since $C_i$ is non-empty by convention, we can choose $\psi_i\in C_i$ for every $1\leq
    i\leq m$. Define the homomorphism $\alpha:\mcA_{inter}(B,\pi) \to (\ast_{i=1}^m
    \C \Z_2^{V_i}, \mu_{inter}+\mu_{sat})$ by 
	\begin{equation*}
		\alpha(\sigma_i(x)) = \sum_{\varphi\in C_i}\Phi_{V_i,\varphi}\sigma_i(x)+\sum_{\mathclap{\varphi\in \Z^{V_i}_2\setminus C_i}}\Phi_{V_i,\varphi}(-1)^{\psi_i(x)}.
	\end{equation*} 
    Let $\Phi_i = \sum_{\varphi\in C_i}\Phi_{V_i,\varphi}$, and recall that $\hsq*{a} =
    a^*a$ denotes the hermitian square of $a$. Then
	\begin{equation*}
		\begin{split}
			\alpha\left[\hsq*{\sigma_i(x)-\sigma_j(x)}\right]&=\hsq*{\Phi_i\sigma_i(x)+(1-\Phi_i)(-1)^{\psi_i(x)}-\Phi_j\sigma_j(x)-(1-\Phi_j)(-1)^{\psi_j(x)}}
			\\
			&\leq4\hsq*{\Phi_i\sigma_i(x)+(1-\Phi_i)(-1)^{\psi_i(x)}-\sigma_i(x)}
			\\
			&\quad +4\hsq*{\Phi_j\sigma_j(x)+(1-\Phi_j)(-1)^{\psi_j(x)}-\sigma_j(x)}+4\hsq*{\sigma_i(x)-\sigma_j(x)}.
		\end{split}
	\end{equation*}
	Observe that $\sigma_i(x) = \sum_{\varphi\in \Z_2^{V_i}}\Phi_{V_i,\varphi}(-1)^{\varphi(x)}$, so
	\begin{equation*}
        \hsq*{\Phi_i\sigma_i(x)+(1-\Phi_i)(-1)^{\psi_i(x)}-\sigma_i(x)} =
            \sum_{\mathclap{\varphi\in \Z_2^{V_i}\setminus C_i}}\Phi_{V_i,\varphi}((-1)^{\psi_i(x)} 
            -(-1)^{\varphi(x)})^2\leq 4\sum_{\mathclap{\varphi\in \Z_2^{V_i}\setminus C_i}}\Phi_{V_i,\varphi}. 
	\end{equation*}
	Thus 
    {
        \small
	\begin{align*}
        \alpha\Bigl(\sum_{\substack{1\leq i \neq j\leq m\\x\in V_i\cap V_j}} \pi(i,j)\hsq*{\sigma_i(x)-\sigma_j(x)}\Bigr)&
            \leq \sum_{\substack{1\leq i \neq j\leq m\\x\in V_i\cap V_j}} \pi(i,j)\Bigl(16\sum_{\mathclap{\varphi\in \Z_2^{V_i}\setminus C_i}}\Phi_{V_i,\varphi} + 16\sum_{\mathclap{\varphi\in \Z_2^{V_j}\setminus C_j}}\Phi_{V_j,\varphi}+4\hsq*{\sigma_i(x)-\sigma_j(x)}\Bigr)
			\\
        & \leq \sum_{a \in \ast_{i=1}^m \C \Z_2^{V_i}} 4 \mu_{inter}(a) a^* a + \sum_{a \in \ast_{i=1}^m \C \Z_2^{V_i}} 32 \frac{K}{C}\mu_{sat}(a) a^* a \\
        &\leq O(K)\sum_{a\in \ast_{i=1}^m \C \Z_2^{V_i}}(\mu_{inter}(a)+\mu_{sat}(a))a^*a,
	\end{align*}
    }
	since $\pi$ is $C$-diagonally dominant.

    Next, suppose $B'$ is a subdivision of $B$. If $\phi \in \Z_2^{V_i}
    \setminus C_i$, then we can choose $j_{\phi} \in [m_i]$ such that
    $\phi|_{V_{i{j_\phi}}} \not\in D_{i{j_\phi}}$. Since $\displaystyle\sum_{\phi :
    \phi|_{V_{ij}} = \phi'} \Phi_{V_i,\phi} = \Phi_{V_{ij},\phi'}$,
    \begin{equation*}
        \sum_{\phi \not\in C_i} \Phi_{V_i,\phi} 
        = \sum_{1 \leq j \leq m_i} \sum_{\phi : j_{\phi} = j} \Phi_{V_i,\phi}
        \leq \sum_{1 \leq j \leq m_i} \sum_{\phi : \phi|_{V_{ij}} \not\in D_{ij}} \Phi_{V_i,\phi}
        = \sum_{1 \leq j \leq m_i} \sum_{\phi' \not\in D_{ij}} \Phi_{V_{ij},\phi'}.
    \end{equation*}
    Hence
    \begin{equation*}
        \sum_{r} \mu_{sat}(r) r^* r \leq M^2 \sum_{r} \mu_{clause}(r) r^* r,
    \end{equation*}
    where the $M^2$ comes from the fact that we divide by $m_i^2$ in the definition of $\mu_{clause}$.
    Thus the identity map $(\ast_{i=1}^m \C \Z_2^{V_i},\mu_{inter}+\mu_{sat})
    \to (\ast_{i=1}^m \C \Z_2^{V_i}, \mu_{inter}+\mu_{clause})$ is an $M^2$-homomorphism.
\end{proof}

Proposition 7.8 of \cite{MS24} shows how to construct tracial states on $\mcA_{inter}(B,\pi)$ from tracial states on $\mcA_{free}(B,B',\pi)$. Their proof assumes $\pi$ is maximized on the diagonal, and obtains a soundness drop-off that is polynomial in the number of contexts $m$ in $B$. We are working with exponentially many contexts, so using this result will give an exponential drop-off in soundness after subdivision. Using a probability distribution that is $C$-diagonally dominant instead allows us to remove the dependence on $m$.
\begin{prop}\label{prop:correction}
    Let $B = \left(X,\{(V_i,C_i)\}_{i=1}^m\right)$ be a BCS, and let $\pi$ be a
    probability distribution on $[m] \times [m]$ which is $C$-diagonally dominant for some $C>0$. Let $B' = \left(X,\{(V_{ij},D_{ij})\}_{i,j}\right)$ be a subdivision of $B$ with
    $m_i$ clauses in constraint $C_i$. 
    If $\tau$ is a trace on $\mcA_{free}(B,B',\pi)$, then there is a
    trace $\wtd{\tau}$ on $\mcA_{inter}(B,\pi)$ such that $\df(\wtd{\tau}) \leq
    \poly(2^L,M,K)\df(\tau)$, where $L = \max_{ij}|V_{ij}|$, $K = \max_{i}|V_i|$, and $M = \max_{i} m_i$. Furthermore, if $\tau$ is finite-dimensional then so is $\wtd{\tau}$. 
\end{prop}
\begin{proof}
    Since $\pi$ is $C$-diagonally dominant, if $\pi(i,i) = 0$ then $\pi(i,j)
    = \pi(j,i) = 0$ for all $j \in [m]$, and the variables in $V_i$ do not appear
    in $\supp(\mu_{inter})$. Thus we may assume without loss of generality that
    $\pi(i,i) > 0$ for all $i \in [m]$. Let $\tau$ be a trace on
    $\mcA_{free}(B,B',\pi)$. By the GNS construction there is
    a $*$-representation $\rho$ of $\mcA_{free}(B,B',\pi)$ acting on a Hilbert
    space $\mcH_0$ with a unit cyclic vector $\psi$ such that $\tau(a) =
    \langle\psi|\rho(a)|\psi\rangle$ for all $a \in \mcA_{free}(B)$. Let
    $\mcM_0 = \overline{\rho(\mcA_{free}(B))}$ be the weak operator closure of the image of $\rho$, and let $\tau_0$
    be the faithful normal tracial state on $\mcM_0$ corresponding to
    $\ket{\psi}$ (so $\tau_0 \circ \rho = \tau)$. 

	 Let $\sum_{a\in\Z^{V_i}_2}\mu_{comm}(a)\|a\|^2_{\tau} = \varepsilon_i$ for all $i \in [m]$. The restriction of $\rho$ to $\mbZ_2^{*V_i}$ is an $\eps_i$\nobreakdash-homomorphism from $\mbZ_2^{V_i}$ into $(\mcM_0,\tau_0)$, so by \Cref{lem:z2stab} there is a representation $\rho_i: \mbZ_2^{V_i}\rightarrow \mcU(\mcM_0)$ such that 
	 \begin{equation}\label{eq:stability}
        \|\rho_i(x_j)-\rho(x_j)\|_{\tau_0}^2\leq \poly(K)\eps_i
	\end{equation}
	for all generators $x_j \in \mbZ_2^{V_i}$. Let $\wtd{\rho} : \ast_{i=1}^m \C \Z_2^{V_i} \to \mcM_0$ be
    the homomorphism defined by $\wtd{\rho}(x) = \rho_i(x)$ for $x\in \mbZ_2^{V_i}$. Suppose $x \in V_i \cap V_j$, then
    \begin{align*}
        \|\wtd{\rho}(\sigma_i(x)-\sigma_j(x))\|_{\tau_0}^2
        & \leq 4 \|\wtd{\rho}(\sigma_i(x)) - \rho(\sigma_i(x))\|_{\tau_0}^2
            + 4 \|\wtd{\rho}(\sigma_j(x)) - \rho(\sigma_j(x))\|_{\tau_0}^2 \\
            & \quad\quad\quad + 4 \|\rho(\sigma_i(x)-\sigma_j(x))\|_{\tau_0}^2 \\
        & \leq \poly(K)(\eps_i+\eps_j) + 4 \|\sigma_i(x) - \sigma_j(x)\|_{\tau}^2.
    \end{align*}
	We conclude that
    \begin{align*}
        \df(\tau_0 \circ \wtd{\rho}; \mu_{inter}) & \leq \sum_{i\neq j} \sum_{x \in V_i \cap V_j} \pi(i,j)\left( \poly(K)(\eps_i+\eps_j)
        + 4 \|\sigma_i(x) - \sigma_j(x)\|_{\tau}^2\right) \\ 
        & \leq\sum_{i} \sum_{x \in V_i \cap V_j} \frac{\pi(i,i)}{C}\left( \poly(K)(\eps_i+\eps_j)
        + 4 \|\sigma_i(x) - \sigma_j(x)\|_{\tau}^2\right)\\
        & \leq  O(\poly(K) \df(\tau;\mu_{comm}) + \df(\tau;\mu_{inter})).
    \end{align*}

    For any $S\subseteq V_i$, let $x_S := \prod_{x \in S} x \in \Z_2^{*V_i}$, where the
    order of the product is inherited from the order on $X$. By \Cref{eq:stability},
    \begin{equation*}
    	\|\wtd{\rho}(x_S)-\rho(x_S)\|_{\tau_0}^2 \leq \poly(K)\eps_i,
    \end{equation*}
    where the degree of the polynomial $\poly(K)$ has increased by one. Since $\Phi_{V_{ij},\phi} =
    \tfrac{1}{2^{|V_{ij}|}} \sum_{S \subseteq V_{ij}} (-1)^{\phi(x_S)}x_S$, we get that
  \begin{equation*}
  	\|\wtd{\rho}(\Phi_{V_{ij},\phi}) - \rho(\Phi_{V_{ij},\phi})\|_{\tau_0}^2
  	\leq \frac{1}{2^{|V_{ij}|}} \sum_{S \subseteq V_{ij}} \|\wtd{\rho}(x_S) - \rho(x_S)\|_{\tau_0}^2
  	\leq \poly(K)\eps_i.
  \end{equation*}
    If $1 \leq i \leq m$, $1 \leq j \leq m_i$, and $\phi \not\in D_{ij}$, then 
    \begin{equation*}
    	\|\wtd{\rho}(\Phi_{V_{ij},\phi})\|_{\tau_0}^2 \leq
    	2 \|\wtd{\rho}(\Phi_{V_{ij},\phi}) - \rho(\Phi_{V_{ij},\phi})\|_{\tau_0}^2 + 
    	2 \|\rho(\Phi_{V_{ij},\phi})\|_{\tau_0}^2, 
    \end{equation*}
    and hence 
    \begin{align*}
        \df(\tau_0 \circ \wtd{\rho}; \mu_{clause}) & = \sum_{i,j} \frac{\pi(i,i)}{m_i^2}
            \sum_{\phi \not\in D_{ij}} \|\wtd{\rho}(\Phi_{V_{ij},\phi})\|_{\tau_0}^2 \\
        & \leq \sum_{i,j} \sum_{\phi \not\in D_{ij}} \frac{\pi(i,i)}{m_i^2} \left( \poly(K)\eps_i + 2 \|\Phi_{V_{ij},\phi}\|^2_{\tau}\right) \\
        & \leq \sum_{i} 2^L\frac{\pi(i,i)}{m_i} \poly(K)\eps_i + 2 \df(\tau;\mu_{clause})\\
        & \leq 2^L\df(\tau;\mu_{comm})+2 \df(\tau;\mu_{clause}).
    \end{align*}
    We conclude that $\wtd{\tau} = \tau_0 \circ
    \wtd{\rho}$ is a tracial state on $\ast_{i=1}^m \C \Z_2^{V_i}$ with $\df(\wtd{\tau};\mu_{inter}+\mu_{clause})$
    bounded by 
    \begin{align*}
         O(\df(\tau;\mu_{inter}) + \df(\tau;\mu_{clause}) + 2^L\poly(K) \df(\tau;\mu_{comm})).
    \end{align*}
    We conclude that 
    \begin{equation*}
        \df(\wtd{\tau};\mu_{inter}+\mu_{clause}) \leq \poly(2^L,K) \df(\tau; \mu_{inter}+\mu_{clause}+\mu_{comm}).
    \end{equation*}
     By \Cref{lem:correctsat}, there is a
    $O(KM^2)$-homomorphism $\mcA_{inter}(B,\pi) \to
    (\ast^m_{i=1}\C\Z^{V_i}_2,\mu_{inter}+\mu_{clause})$, and pulling
    $\wtd{\tau}$ back by this homomorphism gives the proposition.
\end{proof}
The following proposition allows us to pull back tracial states from the subdivision algebra $\mcA_{inter}(B',\pi_{sub})$ to traces on $\mcA_{free}(B,B',\pi)$.
\begin{prop}[\cite{MS24} Proposition 7.9]\label{prop:subdivhom}
    Let $B = \left(X,\{(V_i,C_i)\}_{i=1}^m\right)$ be a BCS, and let $B' =
    \left(X,\{(V_{ij},D_{ij})\}_{i,j}\right)$ be a subdivision of $B$.
    Let $\pi$ be a
    probability distribution on $[m] \times [m]$, and let $\pi_{sub}$ be the
    probability distribution defined from $\pi$ as above. Then there is a $\poly(M,2^{C})$-homomorphism $\mcA_{free}(B,B',\pi) \to \mcA_{inter}(B',\pi_{sub})$,
    where $C = \max_{ij}|V_{ij}|$ and $M = \max_{i}m_i$.
\end{prop}

\begin{proof}[Proof of \Cref{thm:subdiv}]
    Applying \Cref{prop:subdivhom}, \Cref{prop:correction} and \Cref{prop:inter} gives the result.
\end{proof}

\section{Hardness of non-TVF CSPs}\label{sec:non-TVF}

In this section, we prove the first part of \Cref{thm:main-theorem}, proving $\RE$-hardness for noncommutative non-TVF CSPs. The main obstacle in doing so is condition (2) of \Cref{def:subdivision}: in order to subdivide a constraint, each pair of variables must appear in one of the subconstraints. But, in general, putting two variables in the same constraint puts nontrivial restrictions on the values that they may take. A naive way to get around this is by means of empty constraints.

\begin{definition}
    We call a constraint $(V,C)$ over an alphabet $\Sigma$ \textbf{empty} if $C=\Sigma^V$.
\end{definition}
Empty constraints impose no conditions on the variables involved, but they can be used to guarantee that variables commute. Hence, they are useful in subdivision to make sure every pair of variables appears in at least one constraint. However, if we wish to reduce to a CSP that does not contain any empty constraints, we need to find a way to replace these empty constraints by some constraint system coming from the CSP.

If a set of constraints is non-TVF (\Cref{def:tvf}), we can replace any empty constraint by a simple commutativity gadget coming from a non-TVF constraint.

\begin{proposition} \label{prop:remove-empty-constraints-non-tvf}
    Let $\Gamma$ be a non-TVF set of $k$-ary constraints. Then, for every CS $S=(X,\{(V_i,{r_i}_\ast C_i)\}_{i=1}^m)\in\CSP(\Gamma\cup\{\Z_k^2\})$, there exists a CS $S'=(X',\{(V_i',{r_i'}_\ast C_i')\}_{i=1}^{m})\in\CSP(\Gamma)$ such that there is a $\frac{L}{2}$-homomorphism $\alpha:\mc{A}_{c-v}(S,\pi)\rightarrow\mc{A}_{c-v}(S',\pi)$, where $L=\max_{(V,C)\in\Gamma}|V|$; and there is a $1$-homomorphism $\beta:\mc{A}_{c-v}(S',\pi)\rightarrow\mc{A}_{c-v}(S,\pi)$, for every probability distribution $\pi$ on $[m]$.
\end{proposition}

\begin{proof}
    Without loss of generality, we may assume that there exists $1\leq n\leq m$ such that $C_i$ is a constraint from $\Gamma$ for all $i\leq n$ and $C_i$ is an empty constraint for all $i>n$. Since $\Gamma$ is not TVF, there exists a constraint $(V,C)\in\Gamma$ such that $C\subseteq\Z_k^l$ and $u\neq v\in V$ such that for all $a,b\in\Z_k$, there exists $\phi_{a,b}\in C$ with $\phi_{a,b}(u)=a$ and $\phi_{a,b}(v)=b$. For each $i>n$, and $w\in V\backslash\{u,v\}$ let $z_{iw}$ be a variable and take $X'=X\cup\set*{z_{iw}}{i>n,w\in V\backslash\{u,v\}}$. Also, if $i\leq n$, let $(V_i',{r_i'}_\ast C_i')=(V_i,{r_i}_\ast C_i)$; and if $i>n$, take $V_i'=V_i\cup\set{z_{iw}}{w\in V\backslash\{u,v\}}$, $C_i'=C$, and $r_i'$ a bijection such that $r_i'(w)=z_{iw}$ for all $w\in V\backslash\{u,v\}$.

    Now, let $\alpha$ be the inclusion of $\mc{A}_{c-v}(S)$ in $\mc{A}_{c-v}(S')$. First, note that for every $i>n$,
    \begin{align*}
        \alpha\parens[\Big]{\sum_{\substack{x\in V_i\\\phi\in {r_i}_\ast C_i}}\hsq*{\Phi_{V_i,\phi}(1-\Pi_{\phi(x)}(\sigma'(x)))}}&=\sum_{\substack{x\in V_i\\\phi\in\Z_k^{V_i}}}\hsq[\Big]{\prod_{y\in V_i}\Pi_{\phi(y)}(\sigma_i(y))(1-\Pi_{\phi(x)}(\sigma'(x)))}\\
        &=\sum_{\substack{x\in V_i\\\phi\in\Z_k^{V_i'}}}\hsq[\Big]{\prod_{y\in V_i'}\Pi_{\phi(y)}(\sigma_i(y))(1-\Pi_{\phi(x)}(\sigma'(x)))}\\
        &=\sum_{\substack{x\in V_i\\\phi\in {r_i'}_\ast C_i'}}\hsq*{\Phi_{V_i',\phi}(1-\Pi_{\phi(x)}(\sigma'(x)))}\\
        &\leq\sum_{\substack{x\in V_i'\\\phi\in {r_i'}_\ast C_i'}}\hsq*{\Phi_{V_i',\phi}(1-\Pi_{\phi(x)}(\sigma'(x)))}.
    \end{align*}
    Hence, we get that
    \begin{align*}
        &\alpha\parens[\Big]{\sum_{i=1}^m\sum_{\substack{x\in V_i\\\phi\in {r_i}_\ast C_i}}\frac{\pi(i)}{|V_i|}\hsq*{\Phi_{V_i,\phi}(1-\Pi_{\phi(x)}(\sigma'(x)))}}\\
        &\leq\sum_{i=1}^n\frac{\pi(i)}{|V_i|}\sum_{\substack{x\in V_i'\\\phi\in {r_i'}_\ast C_i'}}\hsq*{\Phi_{V_i',\phi}(1-\Pi_{\phi(x)}(\sigma'(x)))}+\sum_{i=n+1}^m\frac{\pi(i)}{2}\sum_{\substack{x\in V_i'\\\phi\in {r_i'}_\ast C_i'}}\hsq*{\Phi_{V_i',\phi}(1-\Pi_{\phi(x)}(\sigma'(x)))}\\
        &\leq\frac{|V|}{2}\sum_{\substack{i,x\in V_i'\\\phi\in {r_i'}_\ast C_i'}}\frac{\pi(i)}{|V_i'|}\hsq*{\Phi_{V_i',\phi}(1-\Pi_{\phi(x)}(\sigma'(x)))}\leq\frac{L}{2}\sum_{\substack{i,x\in V_i'\\\phi\in {r_i'}_\ast C_i'}}\frac{\pi(i)}{|V_i'|}\hsq*{\Phi_{V_i',\phi}(1-\Pi_{\phi(x)}(\sigma'(x)))}
    \end{align*}

    For the converse, let $\beta$ be defined as the map acting on $x\in X$ as $\sigma_i(x)\mapsto\sigma_i(x)$, $\sigma'(x)\mapsto\sigma'(x)$; and on $z_{iw}$ as $$\beta(\sigma_i(z_{iw}))=\beta(\sigma'(z_{iw}))=\sum_{a,b\in\Z_k}\omega_k^{\phi_{a,b}(w)}\Pi_a(\sigma_i(x))\Pi_b(\sigma_i(y)).$$
    Then, for all $i\leq n$, $\beta(\Phi_{V_i',\phi})=\Phi_{V_i,\phi}$, and for $i>n$, writing $V_i=\{x,y\}$,
    \begin{align*}
        \beta(\Phi_{V_i',\phi})&=\Pi_{\phi(x)}(\sigma_i(x))\Pi_{\phi(y)}(\sigma_i(y))\prod_{w\in V\backslash\{u,v\}}\sum_{a,b.\,\phi_{a,b}(w)=\phi(z_{iw})}\Pi_a(\sigma_i(x))\Pi_b(\sigma_i(y))\\
        &=\begin{cases}\Pi_{\phi(x)}(\sigma_i(x))\Pi_{\phi(y)}(\sigma_i(y))&\forall\,w\in V\backslash\{u,v\}\;\phi_{\phi(x),\phi(y)}(w)=\phi(z_{iw})\\0&\text{ else}\end{cases}
    \end{align*}
    That is, $\beta(\Phi_{V_i',\phi})=\Phi_{V_i,\phi|_{V_i}}$ if and only if $\phi=\phi_{\phi(x),\phi(y)}\circ r_i^{-1}$, and otherwise $0$. In particular, it maps every $\Phi_{V_i',\phi}$ for $\phi\notin {r_i}_\ast C_i'$ to $0$, meaning $\beta$ is a $\ast$-homomorphism as needed. Also, in the case $\beta(\Phi_{V_i',\phi})\neq 0$, we have that $\beta(\Phi_{V_i',\phi}(1-\Pi_{\phi(z_{iw})}(\sigma'(z_{iw}))))=0$, so
    \begin{align*}
        &\beta\parens[\Big]{\sum_{\substack{i,x\in V_i'\\\phi\in {r_i'}_\ast C_i'}}\frac{\pi(i)}{|V_i'|}\hsq*{\Phi_{V_i',\phi}(1-\Pi_{\phi(x)}(\sigma'(x)))}}\\
        &=\sum_{i=1}^n\frac{\pi(i)}{|V_i|}\sum_{\substack{x\in V_i\\\phi\in {r_i}_\ast C_i}}\hsq*{\Phi_{V_i,\phi}(1-\Pi_{\phi(x)}(\sigma'(x)))}+\sum_{i=n+1}^m\frac{\pi(i)}{|V|}\sum_{\substack{x\in V_i\\\phi\in {r_i}_\ast C_i}}\hsq*{\Phi_{V_i,\phi}(1-\Pi_{\phi(x)}(\sigma'(x)))}\\
        &\leq\sum_{i}\frac{\pi(i)}{|V_i|}\sum_{\substack{x\in V_i\\\phi\in {r_i}_\ast C_i}}\hsq*{\Phi_{V_i,\phi}(1-\Pi_{\phi(x)}(\sigma'(x)))}.
    \end{align*}
\end{proof}
To prove our main theorem, we first need to transform the \Cref{thm:constanswer} protocol into a $\BCS$-$\MIP^*$ protocol. To do this, we follow \cite{MS24} and use the oracularization $\ttt{G}^{orac}$ of a nonlocal game $\ttt{G} = (I,\{O_i\},\pi,V)$. There are many versions of oracularization in the literature that are all closely related; we use the version from \cite{natarajan2019neexp}. In $\ttt{G}^{orac}$ the verifier samples a question pair $(i_1,i_2) \in I$ according to $\pi$, and then picks $a,b,c\in\{1,2\}$ uniformly at random. If $a=1$, then the verifier sends player $b$ both questions $(i_1,i_2)$, and sends the other player question $(i_c)$. Player $b$ must respond with $a_1 \in O_{i_1}$ and $a_{2}\in O_{i_2}$ such that $V(a_1,a_2|i_1,i_2) = 1$, and the other player must respond with $b\in O_{i_c}$. The players win if $a_c = b$. when $a = 2$, both players are sent $(i_1,i_2)$ and must respond with $(a_1,a_2)$ and $(b_1,b_2)$ in $O_{i_1}\times O_{i_2}$ with $V(a_1,a_2|i_1,i_2) = 1 = V(b_1,b_2|i_1,i_2)$. The players win if $(a_1,b_2) = (b_1,b_2)$. If the questions and answers in $\ttt{G}$ have lengths $q$ and $a$ respectively, then $\ttt{G}^{orac}$ has questions of length $2q$ and answers $2a$. The following lemma shows that this construction is sound.

\begin{lemma}\label{lem:oracsound}(\cite{natarajan2019neexp, ji2022mipre})
    Let $\ttt{G}$ be a nonlocal game. If $\ttt{G}$ has a perfect oracularizable strategy, then $\ttt{G}$ has a perfect synchronous strategy. Conversely, if $\mfk{w}_q(\ttt{G}^{orac}) = 1-\varepsilon$, then $\mfk{w}_q(\ttt{G}) \geq 1-\poly(\varepsilon)$.
\end{lemma}

\begin{proof}
    The proof follows the same lines as Theorem 9.3 of \cite{ji2022mipre}.
\end{proof}

Given a synchronous game $\ttt{G} = (I, \{O_i\}, \pi, V)$ where $I \subseteq
\{0,1\}^n$ and $O_i\subseteq \{0,1\}^{m_i}$, construct a constraint system
$B$ as follows. Take $X$ to be the set of variables $x_{ij}$, where $i \in I$
and $1 \leq j \leq m_i$. Let $V_i = \{x_{ij}, 1 \leq j \leq m_i\}$, and
identify $\Z_2^{V_i}$ with bit strings $\{0,1\}^{m_i}$, where the assignment
to $x_{ij}$ corresponds to the $j$th bit, and let $C_i \subseteq \Z_2^{V_i}$
be the subset corresponding to $O_i$. Let $P = \{(i,j) \in I \times I : 
\pi(i,j) > 0\}$. For $(i,j) \in P$, let $V_{ij} = V_i \cup V_j$, and let $C_{ij}
\subset \Z_2^{V_{ij}} = \Z_2^{V_i} \times \Z_2^{V_j}$ be the set of pairs
of strings $(a,b)$ such that $a \in O_i$, $b \in O_j$, and $V(a,b|i,j) = 1$.
Then $B$ is the constraint system with variables $X$ and constraints
$\{(V_i,C_i)\}_{i \in I}$ and $\{(V_{ij},C_{ij})\}_{(i,j) \in P}$. Let
$I' = I \cup P$ and $\pi^{orac}$ be the probability distribution on 
$I' \times I'$ such that
\begin{equation*}
    \pi^{orac}(i',j') = \begin{cases} 
             \tfrac{1}{8}\pi(i,j)       & i' = (i,j), j' = i \\
             \tfrac{1}{8}\pi(i,j)       & i' = (i,j), j' = j \\
             \tfrac{1}{8}\pi(i,j)       & i' = i, j' = (i,j) \\
             \tfrac{1}{8}\pi(i,j)       & i' = j, j' = (i,j) \\
            \tfrac{1}{2} \pi(i,j) & i' = j' = (i,j)  \\
            0 & \text{ otherwise}
    \end{cases}.
\end{equation*}
Then $\ttt{G}(B,\pi^{orac}) = \ttt{G}^{orac}$, so the oracularization of a synchronous game
is a BCS game. As a result, \Cref{thm:constanswer} has the following corollary:
\begin{corollary}\label{cor:BCSprotocol}
    There is a $\BCS-\MIP^\ast$ protocol $(\ttt{G}(B_x,\pi_x),S,C)$ for the halting problem with constant soundness $s<1$, where $B_x$ has exponentially many contexts of constant size.
\end{corollary}

\begin{proof}
    In \cite{DFNQXY23}, the authors construct a two-prover one-round $\MIP^*$ protocol with polynomial length questions and constant length answers for the halting problem by first oracularizing the $\MIP^*$ protocol for the halting problem from \cite{ji2022mipre}, and then applying a modified answer reduction protocol using the Hadamard code rather than the low degree code used in \cite{natarajan2019neexp}. If a game is oracularizable then so is its oracularization. The completeness argument in \cite{DFNQXY23} Theorem 6.6 is the same as in \cite{natarajan2019neexp} Theorem 17.10 which preserves oracularizability. Hence, the answer reduction step preserves oracularizability, and the protocol in \Cref{thm:constanswer} is oracularizable. The corollary follows by oracularization of this protocol.
\end{proof}
Two types of $\BCS$ transformations are used to prove the main result of \cite{MS24}, subdivisions and classical homomorphisms.
\begin{definition}\label{def:classhom}
    Let $B = (X, \{(V_i,C_i)\}_{i=1}^m)$ and $B' = (X',\{(W_i,D_i)\}_{i=1}^m)$ be boolean constraint
    systems. A homomorphism $\sigma : \mcA(B) \to \mcA(B')$ is a \textbf{classical homomorphism}
    if 
	\begin{enumerate}
		\item $\sigma(\mcA(V_i,C_i))\subseteq \mcA(W_i,D_i)$ for all $1\leq i\leq m$, and
		\item if $\sigma(\Phi_{V_i,\phi_i}) = \sum_k\Phi_{W_i,\psi_{ik}}$,  $\sigma(\Phi_{V_j,\phi_j}) = \sum_k\Phi_{W_j,\psi_{jl}}$, and $\phi_i|_{V_i\cap V_j}\neq \phi_j|_{V_i\cap V_j}$ then $\psi_{ik}|_{W_i\cap W_j}\neq\psi_{jl}|_{W_i\cap W_j}$ for all $k,l$.
	\end{enumerate}
\end{definition}
Classical homomorphisms are constraintwise maps of boolean constraint systems that map satisfying assignments to satisfying assignments and preserve assignments to constraints agreeing on overlapping variables. They also preserve the defect.
\begin{lemma}[\cite{MS24}]\label{lem:classhom}
    Let $B = \left(X,\{(V_i,C_i)\}_{i=1}^m\right)$ and $B' =
    \left(Y,\{(W_i,D_i)\}_{i=1}^m\right)$ be boolean constraint systems, and let $\pi$
    be a probability distribution on $[m] \times [m]$. If $\sigma : \mcA(B) \to
    \mcA(B')$ is a classical homomorphism, then $\sigma$ is a $1$-homomorphism
    $\mcA(B,\pi) \to \mcA(B',\pi)$. 
\end{lemma}
We need the following lemma.
\begin{corollary}[\cite{MS24}]\label{cor:BCStoCSP}
    Let $B = \left(X,\{(V_i,C_i)\}_{i=1}^m\right)$ be a BCS, and let $B' =
    \left(X',\{(W_i,D_i)\}_{i=1}^m\right)$ be a BCS with $X \subset X'$, $V_i
    \subseteq W_i$ for all $1 \leq i \leq m$, and $W_i \cap W_j = V_i \cap V_j$ for
    all $1 \leq i,j \leq m$. Suppose that for all $i \in [m]$, $\phi \in C_i$ if
    and only if there exists $\psi \in D_i$ with $\psi|_{V_i} = \phi$.  Then
    for any probability distribution $\pi$ on $[m] \times [m]$, the
    homomorphism
    \begin{equation*}
		\sigma:\mcA(B) \to \mcA(B') : \sigma_i(x) \mapsto \sigma_i(x) \text{ for } i \in [m], x \in V_i
	\end{equation*}
    defined by the inclusions $V_i \subseteq W_i$ is a $1$-homomorphism
    $\mcA(B,\pi) \to \mcA(B',\pi)$, and there is another $1$-homomorphism
    $\sigma' : \mcA(B',\pi) \to \mcA(B,\pi)$.
\end{corollary}
We are now ready to prove the hardness of non-TVF CSPs.
\begin{theorem}[Part 1 of \Cref{thm:main-theorem}]\label{thm:main-theorem-part-1}
    Let $\Gamma$ be a non-TVF set of $k$-ary constraints such that $\CSP(\Gamma)_{1,1}$ $\NP$-complete. Then there exists a constant $s\in [0,1)$ such that $\SuccinctCSP_{c-v}(\Gamma)^{\ast}_{1,s}$ is $\RE$-complete. 
\end{theorem}
\begin{proof}
    By \Cref{cor:BCSprotocol}, there is a $\BCS-\MIP^\ast$ protocol $(\ttt{G}(B_x,\pi_x),S,C)$ for the halting problem with constant soundness $0\leq s<1$, where $B_x = (X_x,\{(V_i^x,C_i^x)\}_{i=1}^{m_x})$, $m_x$ is exponential in $|x|$, and $|V_i^x| = O(1)$. By the $\NP$-completeness of $\CSP(\Gamma)_{1,1}$, there is a $\BCS$ $B' = (Y_x,\{(W_i^x,D_i^x)\}_{i = 1}^{n_x})$ as in \Cref{cor:BCStoCSP}, where $|W_i^x| = O(1)$, $n_x$ exponential in $x$, and $D_i^x$ is the boolean form of a $\CSP(\Gamma)$ instance. By \Cref{lem:Chom}, there is a $\BCS-\MIP^*$ protocol $(\ttt{G}(B',\pi_x),S,C')$ for the halting problem with the same soundness. Since $|W_i^x| = O(1)$, \Cref{thm:subdiv} implies that there is a a $\BCS-\MIP^*$ protocol $(\ttt{G}(B(S_x),\pi_{sub}^x),\wtd{S},\wtd{C})$ for the halting problem with constant soundness, where $S_x$ is a $\CSP(\Gamma)$ instance that may contain empty constraints. By subdividing further, we may always assume that the empty constraints are on two variables. Applying \Cref{lem:cv-to-cc,lem:bcs-to-kcs,lem:cc-to-cv} gives a constraint-variable $\CS-\MIP^*$ protocol $(\ttt{G}(S_x,\pi_{sub}^x),\wtd{S},\wtd{C}')$ for the halting problem with constant soundness, where $S_x$ is a $\CSP(\Gamma)$ instance that may contain empty constraints. Finally, \Cref{prop:remove-empty-constraints-non-tvf} gives the result.
\end{proof}

\section{Hardness of boolean TVF CSPs}\label{sec:TVF}
\begin{figure}[b]
    \centering
    \pgfdeclarelayer{triangle}
    \pgfsetlayers{triangle,main}
    \begin{tikzpicture}
        %Add nodes
        \draw[fill=black] (0,0) circle[radius=2.5pt] node(x)  {};
        \draw[fill=black] (-1,-1.5) circle[radius=2.5pt] node(u)  {};
        \draw[fill=black] (1,-1.5) circle[radius=2.5pt] node(v)  {};
        \draw[fill=black] (-2,-3) circle[radius=2.5pt] node(y)  {};
        \draw[fill=black] (0,-3) circle[radius=2.5pt] node(w)  {};
        \draw[fill=black] (2,-3) circle[radius=2.5pt] node(z)  {};
        \tikzstyle{inner triangle} = [line width=6mm,rounded corners=5pt]
        \tikzstyle{outer triangle} = [line width=2mm,rounded corners=5pt]
        %Draw triangles
        \begin{pgfonlayer}{triangle}
            \draw[rounded corners=20pt, blue] 
            ([xshift=-10pt,yshift=-4pt]u.south west) -- ([xshift=10pt,yshift=-4pt]v.south east) -- ([xshift=0pt,yshift=12pt]x.north) -- cycle;
            \draw[rounded corners=20pt, red] 
            ([xshift=-10pt,yshift=-4pt]y.south west) -- ([xshift=10pt,yshift=-4pt]w.south east) -- ([xshift=0pt,yshift=12pt]u.north) -- cycle;
            \draw[rounded corners=20pt, teal] 
            ([xshift=-10pt,yshift=-4pt]w.south west) -- ([xshift=10pt,yshift=-4pt]z.south east) -- ([xshift=0pt,yshift=12pt]v.north) -- cycle;
            %\draw[outer triangle, blue, draw opacity=0.5, fill opacity=0.5]
            %(x.center)  -- (u.center) -- (v.center) -- cycle;
            %\filldraw[inner triangle, white]
            %(x.center)  -- (u.center) -- (v.center) -- cycle;
        \end{pgfonlayer}
        %Add labels
        \node[above left=3pt, font=\footnotesize] at (y) {y};
        \node[above left=3pt, font=\footnotesize] at (w) {w};
        \node[above left=3pt, font=\footnotesize] at (z) {z};
        \node[above left=3pt, font=\footnotesize] at (x) {x};
        \node[above left=3pt, font=\footnotesize] at (u) {u};
        \node[above left=3pt, font=\footnotesize] at (v) {v};
    \end{tikzpicture}
    \caption{The basic commutativity gadget for TVF boolean constraint systems. Exactly one variable in each triangle must be assigned value 1. These constraints bound the commutator $[x,y]$, and any assignment to $x$ and $y$ may be extended to an assignment to all three constraints.}
    \label{fig:flux-capacitor}
\end{figure}

In this section we examine the hardness of boolean CSPs that are two-variable falsifiable. $\NP$-complete classical boolean CSPs can emulate any other constraint system. What we show in this section is that this emulation can be done in a way that is sound against quantum provers. We find that boolean TVF CSPs are $\NP$-complete if and only if they allow for a commutativity gadget similar to that constructed in \cite{Ji13} for $1$-in-$3$SAT. We prove the $\RE$-hardness of the entangled version of these CSPs by showing that the commutativity gadget is quantum sound. 

\subsection{The basic commutativity gadget}

We begin with a description of the commutativity gadget, and the proof of its quantum soundness.

\begin{lemma}\label{lem:basic-commutativity-gadget}
    Let $C=\set*{(0,0,1),(0,1,0),(1,0,0)}\in\Z_2^{[3]}$. Let $X=\{u,v,w,x,y,z\}$ be a set of variables; let $V_1=\{x,u,v\}$, $V_2=\{y,u,w\}$, and $V_3=\{z,v,w\}$; and let $r_i:[3]\rightarrow V_i$ be bijections. Consider the BCS $B=\{X,\{(V_i,{r_i}_\ast C)\}_{i=1}^3\}$. For all $a_x,a_y\in\Z_2$, there exists a classical satisfying assignment $\phi:X\rightarrow\Z_2$, $\phi|_{V_i}\in {r_i}_\ast C$ such that $\phi(x)=a_x$ and $\phi(y)=a_y$. Also, in the algebra $\mc{A}_{c-v}(B)$,
    \begin{align*}
        \hsq*{[\sigma'(x),\sigma'(y)]}\lesssim512\sum_{i=1}^3\sum_{\phi\in {r_i}_\ast C}\sum_{z\in V_i}\hsq*{\Phi_{V_i,\phi}(1-\Pi_{\phi(z)}(\sigma'(z)))}.
    \end{align*}
\end{lemma}

This is a robust version of Lemma 5 in~\cite{Ji13}.

The commutativity gadget defined in this lemma is illustrated in \Cref{fig:flux-capacitor}. Note also that an identical commutativity gadget can be constructed from $C$ with any of the variables negated, one of the constraints $C'=\{(0,0,0),(0,1,1),(1,0,1)\}$, $C''=\{(0,1,0),(0,0,1),(1,1,1)\}$, or $C'''=\{(1,1,0),(1,0,1),(0,1,1)\}$, where the last one, two, or three variable are negated, respectively. To achieve this, it suffices to connect three constraints in the same way, additionally taking care that negated variables are only connected to other negated variables. This is illustrated in \Cref{fig:flux-capacitor-2}.

\begin{proof}
    For the first part, we can work out all the cases: for $X$ ordered as in the lemma statement, we get the satisfying assignments $(1,1,1,0,0,0)$ for $(a_x,a_y)=(0,0)$, $(0,1,0,0,1,0)$ for $(a_x,a_y)=(0,1)$, $(0,0,1,1,0,0)$ for $(a_x,a_y)=(1,0)$, and $(0,0,0,1,1,1)$ for $(a_x,a_y)=(1,1)$. For the second part, write $x_1^i=\Pi_1(\sigma_i(x))$ and $x_0^i=\Pi_0(\sigma_i(x))$ and similarly for the other variables. Write $x_1=\Pi_1(\sigma'(x))$ and $x_0=\Pi_0(\sigma'(x))$ and similarly for the other variables. First, note that $x^1_1u^1_1=x^1_1v^1_1=u^1_1v^1_1=0$, so $x^1_1+u^1_1+v^1_1\leq 1$, and
    \begin{align*}
        1=x_1^1u^1_0v^1_0+x_0^1u^1_1v^1_0+x_0^1u^1_0v^1_1\leq x^1_1+u^1_1+v^1_1,
    \end{align*}
    giving equality $x^1_1+u^1_1+v^1_1=1$. In the same way, $y_1^2+u_1^2+w_1^2=1$. Therefore, the commutator
    \begin{align*}
        [x_1^1,y_1^2]&=[1-u_1^1-v_1^1,1-u_1^2-w_1^2]=[u_1^1+v_1^1,u_1^2+w_1^2]
        \\
        &=[u_1^1,u_1^2]+[u_1^1,w_1^2]+[v_1^1,u_1^2]+[v_1^1,w_1^2].
    \end{align*}
    Noting that $[u^1_1,u^1_1]=[u^2_1,w^2_1]=[v^1_1,u^1_1]=[v^3_1,w^3_1]=0$, we can write
    \begin{align*}
        [x_1^1,y_1^2]&=[u_1^1,u_1^2-u_1^1]+[u_1^1-u_1^2,w_1^2]+[v_1^1,u_1^2-u^1_1]+[v_1^1-v^3_1,w_1^2]+[v^3_1,w^2_1-w^3_1]\\
        &=[u_1^1+v^1_1+w_1^2,u_1^2-u_1^1]-[w^2_1,v_1^1-v_1^3]+[v^3_1,w^2_1-w_1^3].
    \end{align*}
    From here, we can expand
    \begin{align*}
        [x_1,y_1]&=[x^1_1,y^2_1]+[x_1-x^1_1,y^2_1]+[x_1,y_1-y^2_1]\\
        &=[u_1^1+v^1_1+w_1^2,u_1^2-u_1^1]-[w^2_1,v_1^1-v_1^3]+[v^3_1,w^2_1-w_1^3]+[y^2_1,x_1^1-x_1]-[x_1,y^2_1-y_1]\\
        &=[u_1^1+v^1_1+w_1^2,u_1^2-u_1]-[u_1^1+v^1_1+w_1^2,u_1^1-u_1]-[w^2_1,v_1^1-v_1]+[w^2_1,v_1^3-v_1]\\
        &\qquad+[v^3_1,w^2_1-w_1]-[v^3_1,w^3_1-w_1]+[y^2_1,x_1^1-x_1]-[x_1,y^2_1-y_1].
    \end{align*}
    Hence, we get the bound on the hermitian square of the commutator
    \begin{align*}
        \hsq*{[x_1,y_1]}&\leq2^{\ceil{\log 8}}\big(\hsq*{[u_1^1+v^1_1+w_1^2,u_1^2-u_1]}+\hsq*{[u_1^1+v^1_1+w_1^2,u_1^1-u_1]}+\hsq*{[w^2_1,v_1^1-v_1]}\\
        &\qquad+\hsq*{[w^2_1,v_1^3-v_1]}+\hsq*{[v^3_1,w^2_1-w_1]}+\hsq*{[v^3_1,w^3_1-w_1]}+\hsq*{[y^2_1,x_1^1-x_1]}+\hsq*{[x_1,y^2_1-y_1]}\big)\\
        &\lesssim 16\big(\hsq*{u_1^2-u_1}+\hsq*{u_1^1-u_1}+\hsq*{v_1^1-v_1}+\hsq*{v_1^3-v_1}+\hsq*{w^2_1-w_1}\\
        &\qquad+\hsq*{w^3_1-w_1}+\hsq*{x_1^1-x_1}+\hsq*{y^2_1-y_1}\big)\\
        &\leq16\sum_{i=1}^3\sum_{t\in V_i}\hsq*{t_1^i-t_1}=4\sum_{i=1}^3\sum_{t\in V_i}\hsq*{1-t^it}=16\sum_{i=1}^3\sum_{t\in V_i}\hsq*{t^i_0t_1+t^i_1t_0}\\
        &\leq32\sum_{i=1}^3\sum_{t\in V_i}\sum_{a\in\Z_2}\hsq*{t^i_a(1-t_a)}=32\sum_{i=1}^3\sum_{\phi\in {r_i}_\ast C}\sum_{z\in V_i}\hsq*{\Phi_{V_i,\phi}(1-\Pi_{\phi(z)}(\sigma'(z)))}.
    \end{align*}
    Noting that $[x_1,y_1]=\frac{1}{4}[x,y]$ finishes the proof.
\end{proof} 

\begin{figure}[b]
    \centering
    \pgfdeclarelayer{triangle}
    \pgfsetlayers{triangle,main}
    \begin{tikzpicture}[scale=0.8]
        %Add nodes
        \draw[fill=black] (0,0) circle[radius=2.5pt] node(x)  {};
        \draw (-1,-1.5) circle[radius=2.5pt] node(u)  {};
        \draw[fill=black] (1,-1.5) circle[radius=2.5pt] node(v)  {};
        \draw[fill=black] (-2,-3) circle[radius=2.5pt] node(y)  {};
        \draw[fill=black] (0,-3) circle[radius=2.5pt] node(w)  {};
        \draw (2,-3) circle[radius=2.5pt] node(z)  {};
        \tikzstyle{inner triangle} = [line width=6mm,rounded corners=5pt]
        \tikzstyle{outer triangle} = [line width=2mm,rounded corners=5pt]
        %Draw triangles
        \begin{pgfonlayer}{triangle}
            \draw[rounded corners=20pt, blue] 
            ([xshift=-10pt,yshift=-4pt]u.south west) -- ([xshift=10pt,yshift=-4pt]v.south east) -- ([xshift=0pt,yshift=12pt]x.north) -- cycle;
            \draw[rounded corners=20pt, red] 
            ([xshift=-10pt,yshift=-4pt]y.south west) -- ([xshift=10pt,yshift=-4pt]w.south east) -- ([xshift=0pt,yshift=12pt]u.north) -- cycle;
            \draw[rounded corners=20pt, teal] 
            ([xshift=-10pt,yshift=-4pt]w.south west) -- ([xshift=10pt,yshift=-4pt]z.south east) -- ([xshift=0pt,yshift=12pt]v.north) -- cycle;
            %\draw[outer triangle, blue, draw opacity=0.5, fill opacity=0.5]
            %(x.center)  -- (u.center) -- (v.center) -- cycle;
            %\filldraw[inner triangle, white]
            %(x.center)  -- (u.center) -- (v.center) -- cycle;
        \end{pgfonlayer}
        %Add labels
        \node[above left=3pt, font=\footnotesize] at (y) {y};
        \node[above left=3pt, font=\footnotesize] at (w) {w};
        \node[above left=3pt, font=\footnotesize] at (z) {z};
        \node[above left=3pt, font=\footnotesize] at (x) {x};
        \node[above left=3pt, font=\footnotesize] at (u) {u};
        \node[above left=3pt, font=\footnotesize] at (v) {v};
    \end{tikzpicture}
    \begin{tikzpicture}[scale=0.8]
        %Add nodes
        \draw[fill=black] (0,0) circle[radius=2.5pt] node(x)  {};
        \draw (-1,-1.5) circle[radius=2.5pt] node(u)  {};
        \draw (1,-1.5) circle[radius=2.5pt] node(v)  {};
        \draw[fill=black] (-2,-3) circle[radius=2.5pt] node(y)  {};
        \draw (0,-3) circle[radius=2.5pt] node(w)  {};
        \draw[fill=black] (2,-3) circle[radius=2.5pt] node(z)  {};
        \tikzstyle{inner triangle} = [line width=6mm,rounded corners=5pt]
        \tikzstyle{outer triangle} = [line width=2mm,rounded corners=5pt]
        %Draw triangles
        \begin{pgfonlayer}{triangle}
            \draw[rounded corners=20pt, blue] 
            ([xshift=-10pt,yshift=-4pt]u.south west) -- ([xshift=10pt,yshift=-4pt]v.south east) -- ([xshift=0pt,yshift=12pt]x.north) -- cycle;
            \draw[rounded corners=20pt, red] 
            ([xshift=-10pt,yshift=-4pt]y.south west) -- ([xshift=10pt,yshift=-4pt]w.south east) -- ([xshift=0pt,yshift=12pt]u.north) -- cycle;
            \draw[rounded corners=20pt, teal] 
            ([xshift=-10pt,yshift=-4pt]w.south west) -- ([xshift=10pt,yshift=-4pt]z.south east) -- ([xshift=0pt,yshift=12pt]v.north) -- cycle;
            %\draw[outer triangle, blue, draw opacity=0.5, fill opacity=0.5]
            %(x.center)  -- (u.center) -- (v.center) -- cycle;
            %\filldraw[inner triangle, white]
            %(x.center)  -- (u.center) -- (v.center) -- cycle;
        \end{pgfonlayer}
        %Add labels
        \node[above left=3pt, font=\footnotesize] at (y) {y};
        \node[above left=3pt, font=\footnotesize] at (w) {w};
        \node[above left=3pt, font=\footnotesize] at (z) {z};
        \node[above left=3pt, font=\footnotesize] at (x) {x};
        \node[above left=3pt, font=\footnotesize] at (u) {u};
        \node[above left=3pt, font=\footnotesize] at (v) {v};
    \end{tikzpicture}\begin{tikzpicture}[scale=0.8]
        %Add nodes
        \draw (0,0) circle[radius=2.5pt] node(x)  {};
        \draw (-1,-1.5) circle[radius=2.5pt] node(u)  {};
        \draw (1,-1.5) circle[radius=2.5pt] node(v)  {};
        \draw (-2,-3) circle[radius=2.5pt] node(y)  {};
        \draw (0,-3) circle[radius=2.5pt] node(w)  {};
        \draw (2,-3) circle[radius=2.5pt] node(z)  {};
        \tikzstyle{inner triangle} = [line width=6mm,rounded corners=5pt]
        \tikzstyle{outer triangle} = [line width=2mm,rounded corners=5pt]
        %Draw triangles
        \begin{pgfonlayer}{triangle}
            \draw[rounded corners=20pt, blue] 
            ([xshift=-10pt,yshift=-4pt]u.south west) -- ([xshift=10pt,yshift=-4pt]v.south east) -- ([xshift=0pt,yshift=12pt]x.north) -- cycle;
            \draw[rounded corners=20pt, red] 
            ([xshift=-10pt,yshift=-4pt]y.south west) -- ([xshift=10pt,yshift=-4pt]w.south east) -- ([xshift=0pt,yshift=12pt]u.north) -- cycle;
            \draw[rounded corners=20pt, teal] 
            ([xshift=-10pt,yshift=-4pt]w.south west) -- ([xshift=10pt,yshift=-4pt]z.south east) -- ([xshift=0pt,yshift=12pt]v.north) -- cycle;
            %\draw[outer triangle, blue, draw opacity=0.5, fill opacity=0.5]
            %(x.center)  -- (u.center) -- (v.center) -- cycle;
            %\filldraw[inner triangle, white]
            %(x.center)  -- (u.center) -- (v.center) -- cycle;
        \end{pgfonlayer}
        %Add labels
        \node[above left=3pt, font=\footnotesize] at (y) {y};
        \node[above left=3pt, font=\footnotesize] at (w) {w};
        \node[above left=3pt, font=\footnotesize] at (z) {z};
        \node[above left=3pt, font=\footnotesize] at (x) {x};
        \node[above left=3pt, font=\footnotesize] at (u) {u};
        \node[above left=3pt, font=\footnotesize] at (v) {v};
    \end{tikzpicture}
    \caption{Basic commutativity gadgets with one, two, or three variables per constraint negated. The white vertices indicate the negated variables: note that negated variables must be only connected amongst themselves to construct the gadget.}
    \label{fig:flux-capacitor-2}
\end{figure}

\subsection{Compression and simulation: building the needed constraints}

In this section, we analyse the structure of boolean TVF constraints. In particular, we show that, given an $\NP$-complete set of boolean TVF constraints, we can recover the 1-in-3SAT constraint from \Cref{lem:basic-commutativity-gadget}, or its negation on a subset of variables. To do so, we need to study the combinatorial structure of TVF constraints, which will first allow us to simplify the sets of constraints we work with, and then simulate the wanted constraint.

\begin{definition}
    A \textbf{(boolean) TVF graph} is a graph with two types of undirected edges, labelled $00$ and $11$, and one type of directed edge, labelled $01$. We denote a TVF graph $G=(V,E_{00}\sqcup E_{11}\sqcup E_{01})$. Given a vertex $v\in V$, we say that an edge $e$ of $G$ is \textbf{$0$ on $v$} if there exists $u\in V$ such that $e=\{u,v\}\in E_{00}$ or $e=(v,u)\in E_{01}$. In the same way, we say that $e$ is \textbf{$1$ on $v$} if there exists $u\in V$ such that $e=\{u,v\}\in E_{11}$ or $e=(u,v)\in E_{01}$.
    
    Let $V$ be a finite set and let $C\subseteq\Z_2^V$ be a boolean constraint on $V$. The \textbf{TVF graph of $C$} is a graph $G_{TVF}(C)=(V,E_{00}\sqcup E_{11}\sqcup E_{01})$, where $\{u,v\}\in E_{00}$ is an edge if $\phi(u)=\phi(v)=0$ implies $\phi\notin C$; $\{u,v\}\in E_{11}$ is an edge if $\phi(u)=\phi(v)=1$ implies $\phi\notin C$; and $(u,v)\in E_{01}$ is a directed edge if $\phi(u)=0$ and $\phi(v)=1$ implies $\phi\notin C$.
\end{definition}

A constraint $C$ is TVF if and only if the TVF graph of $C$ is complete. We can study TVF graphs independently of the constraints that generate them, although any TVF graph is generated by some constraint.

In a TVF graph, we sometimes denote an undirected edge $\{a,b\}$ in the same way as a directed edge, $(a,b)$, if we want to be more general about which edge set the edge belongs to.

\begin{definition}
    An \textbf{assignment} to a TVF graph $G=(V,E_{00}\sqcup E_{11}\sqcup E_{01})$ is a function $\phi:V\rightarrow\Z_2$ such that $(\phi(u),\phi(v))\neq(0,0)$ for all $\{u,v\}\in E_{00}$, $(\phi(u),\phi(v))\neq(1,1)$ for all $\{u,v\}\in E_{11}$, and $(\phi(u),\phi(v))\neq(0,1)$ for all $(u,v)\in E_{01}$.
\end{definition}

Every satisfying assignment to $C$ induces an assignment to $G_{TVF}(C)$, but the converse is not necessarily true. 

\begin{definition}
    Let $G=(V,E_{00}\sqcup E_{11}\sqcup E_{01})$ be a TVF graph. We say that $G$ is \textbf{compressible} to $U\subseteq V$ if for all $v\in V\backslash U$ and all assignments $\phi$ to $G$, either $\exists b\in\Z_2$ such that $\phi(v)=b$, $\exists u\in U$ such that $\phi(v)=\phi(u)$, or $\exists u\in U$ such that $\phi(v)=\lnot\phi(u)$.

    In the first case we say that $G$ \textbf{compresses by a constant} to $b$ at $v$, in the second case we say that $G$ \textbf{compresses by equality} to $u$ at $v$, and in the third case we say that $G$ \textbf{compresses by negation} to $u$ at $v$.

    We say $G$ is \textbf{incompressible} if $G$ is only compressible to $U\subseteq V$ when $U = V$.
    
    We use the same notation for a constraint $C\subseteq\Z_2^V$ when the properties hold for its TVF graph $G_{TVF}(C)$.
\end{definition}

It is easy to see that if a subgraph of $G$ is compressible, so is $G$. Now, we give some important examples of compressible TVF graphs.

\begin{lemma}\label{lem:compressible}
    \begin{enumerate}[(i)]
        \item A TVF graph with a $00$ or $11$ edge as a loop is compressible.
        \item A TVF graph with a double edge between two distinct vertices is compressible.
        \item The following TVF graph is compressible for $n>1$ and all $b_1,\dots,b_{n+1}\in\Z_2$: $G=(\{x_1,\ldots,x_n\},E_{00}\sqcup E_{11}\sqcup E_{01})$ with $(x_{i},x_{i+1})\in E_{\lnot b_ib_{i+1}}$ for $i=1,\ldots,n-1$ and $(x_n,x_1)\in E_{\lnot b_nb_{n+1}}$.\label{lem:compressible-iii}
    \end{enumerate}
\end{lemma}

Note that $01$ loops are redundant as they do not affect the assignments. As such we may suppose that TVF graphs have no $01$ loops. Hence, this lemma tells us in particular that incompressible TVF graphs have no loops or multiple edges, \emph{i.e.} they are simple graphs. The cycle graph described in (iii) is illustrated in \cref{fig:compressible-cycle}.

\begin{proof}
\begin{enumerate}[(i)]
    \item Consider the one-vertex TVF graph $G$ with $V=\{x\}$ and $E_{bb}=\{\{x,x\}\}$. Then, we know that for every assignment $\phi$ to $G$, if $\phi(x)=b$, then $\phi(x)\neq b$. Hence, we must have that $\phi(x)=\lnot b$. As such, $G$ compresses by a constant at $x$.

    \item Consider the two-vertex TVF graph $G=(\{x,y\},E_{00}\sqcup E_{11}\sqcup E_{01})$. Up to relabelling of $x$ and $y$, there are $4$ possible double edges. In the case $E_{00}=E_{11}=\{\{x,y\}\}$ and $E_{01}=\varnothing$, we have that $\phi(x)\neq\phi(y)$ for every assignment $\phi$, so $G$ is compressible by negation. In the case $E_{00}=\{\{x,y\}\}$, $E_{11}=\varnothing$, and $E_{01}=\{(x,y)\}$, we have that $\phi(x)=1$ for every assignment $\phi$, so $G$ is compressible by a constant. In the case $E_{00}=\varnothing$, $E_{11}=\{\{x,y\}\}$, and $E_{01}=\{(x,y)\}$, we have that $\phi(y)=0$ for every assignment $\phi$, so $G$ is compressible by a constant. And in the case $E_{00}=E_{11}=\varnothing$ and $E_{01}=\{(x,y),(y,x)\}$, we have that $\phi(x)=\phi(y)$ for every assignment $\phi$, so $G$ is compressible by equality.

    \item If $n=2$, we are in the case of (ii), so we know that $G$ is compressible. In the case $n\geq 3$, we consider two cases depending on the value of $\phi(x_2)$ for an assignment $\phi$. First note that, if $\phi(x_2)=b_2$, then $\phi(x_1)=b_1$. Also, if $\phi(x_2)=\lnot b_2$, then $\phi(x_3)=\lnot b_3$, and by induction $\phi(x_i)=\lnot b_i$ for $i=2,\ldots,n$. This implies $\phi(x_1)=\lnot b_{n+1}$. If $b_1=\lnot b_{n+1}$, then $\phi(x_1)=b_1$ in both possible cases, so $G$ is compressible by a constant. If $b_1 = b_{n+1}$, we get in the first case that $\phi(x_1)=b_1=b_{n+1}$, so $\phi(x_n)=b_n$ and by induction $\phi(x_i)=b_i$ for all $i=2,\ldots,n$. As $n\geq 3$, there exist $i\neq j=1,\ldots,n$ such that $b_i=b_j$, and hence $\phi(x_i)=\phi(x_j)$ in both cases, giving that $G$ is compressible by equality.
\end{enumerate}
\end{proof}

\begin{figure}
    \centering
    \begin{tikzpicture}[scale=1.2]
        \draw (150:3cm) -- (120:3cm) node[pos=0.5,above,sloped] {$\lnot b_{n-1}b_n$} -- (90:3cm) node[pos=0.5,above,sloped] {$\lnot b_{n}b_{n+1}$} -- (60:3cm) node[pos=0.5,above,sloped] {$\lnot b_{1}b_2$} (-30:3cm) -- (-60:3cm) node[pos=0.5,below,sloped] {\rotatebox{180}{$\lnot b_ib_{i+1}$}};% -- (30:3cm) node[pos=0.5,above,sloped] {$\lnot b_{2}b_3$};
        \draw[shorten <=20] (180:3cm) -- (150:3cm);
        \draw[shorten <=20]  (30:3cm) -- (60:3cm);
        \draw[shorten <=20]  (-90:3cm) -- (-60:3cm);
        \draw[shorten <=20]  (0:3cm) -- (-30:3cm);
        \draw[dashed,thin] (180:3cm) -- (150:3cm) (60:3cm) -- (30:3cm) (-90:3cm) -- (-60:3cm) (0:3cm) -- (-30:3cm);
        \fill (150:3cm) circle (3pt) node[below right] {$x_{n-1}$} (120:3cm) circle (3pt) node[below right] {$x_{n}$} (90:3cm) circle (3pt) node[below=3pt] {$x_{1}$} (60:3cm) circle (3pt) node[below left] {$x_{2}$} (-30:3cm) circle (3pt) node[above left] {$x_{i}$} (-60:3cm) circle (3pt) node[above left] {$x_{i+1}$};% (30:3cm) circle (3pt) node[below left] {$x_{3}$};
    \end{tikzpicture}
    \caption{The compressible cycle TVF graph from \Cref{lem:compressible}.\ref{lem:compressible-iii}. }
    \label{fig:compressible-cycle}
\end{figure}

If a given set of constraints has compressible elements, there may be significant redundancy in the variables. To remedy this, we show that we can always reduce to the case of incompressible constraints, while preserving the $\NP$-hardness. We will then show that all compressed $\NP$-hard constraint systems allow for the construction of commutativity gadgets. 

\begin{lemma}\label{lem:compression-transitive}
    Let $C\subseteq\Z_2^V$ be a constraint. Suppose $C$ is compressible to $U\subseteq V$, and that $C|_U$ is compressible to $W$. Then, $C$ is compressible to $W$.
\end{lemma}

\begin{proof}
    Let $v\in V\backslash W$. If $v\in U$, we know, since $C|_U$ is compressible to $W$, that either $\phi(v)=b$ for some $b\in\Z_2$, $\phi(v)=\phi(u)$ for some $u\in W$, or $\phi(v)=\lnot\phi(u)$ for some $u\in W$, for all $\phi\in C$. If $v\notin U$, since $C$ is compressible to $U$, either $\phi(v)=b$ for some $b\in\Z_2$, $\phi(v)=\phi(u)$ for some $u\in U$, or $\phi(v)=\lnot\phi(u)$ for some $u\in U$, for all $\phi\in C$. In the latter two cases, if $u\in W$ we are done. If not, knowing that $C|_U$ compresses at $u$, we have that either $\phi(u)=b$, $\phi(u)=\phi(w)$, or $\phi(u)=\lnot\phi(w)$, for some $w\in W$. As such, in the two cases, we get either $\phi(v)=b$, $\phi(v)=\phi(w)$, or $\phi(v)=\lnot\phi(w)$; or $\phi(v)=\lnot b$, $\phi(v)=\lnot\phi(w)$, or $\phi(v)=\phi(w)$.
\end{proof}

\begin{proposition}
    Every constraint $C\subseteq\Z_2^V$ is compressible to an incompressible constraint.
\end{proposition}

\begin{proof}
    We proceed recursively. If for every nonempty $U\subset V$, $C$ does not compress to $U$, then $C$ is incompressible. Otherwise, $C$ compresses to some $U_1$. Now, consider $C|_{U_1}$, and continue recursively. We get a descending sequence of subsets $V\supset U_1\supset U_2\supset\ldots$ such that $C$ compresses to $U_1$, $C|_{U_1}$ compresses to $U_2$, and so on. As these inclusions are strict, there will be some $k\in\N$ such that $C|_{U_k}$ is incompressible. But by \Cref{lem:compression-transitive}, $C$ compresses to $U_k$.
\end{proof}

\begin{definition}
    Let $\Gamma$ be a set of constraints. For each $(V,C)\in\Gamma$, let $U_C\subseteq V$ be a set of variables such that $C$ is compressible to $U_C$ and $C|_{U_C}$ is incompressible. Then, the \textbf{maximal compression} of $\Gamma$ is $\Gamma_{\max}=\Gamma_{\mathrm{comp}}\cup\Gamma_{\mathrm{aux}}$, where $\Gamma_{\mathrm{comp}}=\set*{(U_C,C|_{U_C})}{C\in\Gamma}$ and
    \begin{itemize}
        \item $\{b\}\in\Gamma_{\mathrm{aux}}$ (constant constraint) iff, for some $(V,C)\in\Gamma$, there exists $v\in V\backslash U_C$ that compresses by a constant to $b\in\Z_2$;

        \item $C_==\{(0,0),(1,1)\}\in\Gamma_{\mathrm{aux}}$ (equality constraint) iff, for some $(V,C)\in\Gamma$, there exists $v\in V\backslash U_C$ that compresses by equality; and

        \item $C_{\neq}=\{(0,1),(1,0)\}\in\Gamma_{\mathrm{aux}}$ (negation constraint) iff, for some $(V,C)\in\Gamma$, there exists $v\in V\backslash U_C$ that compresses by negation.
    \end{itemize}
\end{definition}

Note that all the constraints in $\Gamma_{\mathrm{comp}}$ are incompressible, but the constraints in $\Gamma_{\mathrm{aux}}$ are compressible. However, they are important in making sure the hardness of the CSP is preserved, both in the quantum and classical cases.

\begin{lemma}\label{lem:max-compression-np-hard}
    Let $\Gamma$ be a set of constraints and let $\Gamma_{\max}$ be its maximal compression. If $\Gamma_{\max}$ satisfies one of the polymorphisms $0$, $1$, $\mathrm{AND}$, $\mathrm{OR}$, $\mathrm{MAJ}$, or $\mathrm{MIN}$, so does $\Gamma$.
\end{lemma}

By contrapositive and Schaefer's dichotomy theorem, we get that if $\CSP(\Gamma)_{1,1}$ is $\NP$-complete, so is $\CSP(\Gamma_{\max})_{1,1}$.

\begin{proof}
    Suppose first that $\Gamma_{\max}$ satisfies the constant polymorphism $0$. Then, by construction we know that if $\Gamma$ compresses by a constant, it must compress to $0$, and it cannot compress by negation, as the negation constraint does not satisfy a constant polymorphism. Thus, all the compressed variables compress by a constant to $0$ or by equality. As such, we know that $0\in C|_{U_C}$ implies that $0\in C$, so $\Gamma$ satisfies the constant polymorphism $0$. The same argument holds for the constant polymorphism $1$.

    Next, suppose that $\Gamma_{\max}$ satisfies the polymorphism $\mathrm{AND}$. By construction, we know that $\Gamma$ does not compress by negation, as the negation constraint does not satisfy $\mathrm{AND}$. As such, every variable must compress by equality or a constant. Thus, if $C|_{U_C}$ compressed in this way satisfies $\mathrm{AND}$, so does $C$. The same argument holds for $\mathrm{OR}$.
    
    Suppose now that $\Gamma_{\max}$ satisfies the majority polymorphism $\mathrm{MAJ}$. If $C|_{U_C}$ satisfies $\mathrm{MAJ}$, then so does $C$, as every assignment to the variables in $U_C$ is constant, or equal to, or the negation of one of the remaining variables, and majority commutes with negation. So $\Gamma$ must also satisfy $\mathrm{MAJ}$. The same argument holds for $\mathrm{MIN}$.
\end{proof}

In what follows, we can always assume that we are working with a maximally compressed set of constraints, hence with either an incompressible constraint, a constant constraint, an equality constraint, or a negation constraint. Note that, of these constraints, only the incompressible constraints can have more than three satisfying assignments, which is a necessary condition to not satisfy the majority polymorphism. Now, we study the structure of incompressible TVF constraints, and use it to construct commutativity gadgets.

\begin{definition}
    Let $G=(V,E_{00}\sqcup E_{11}\sqcup E_{01})$ be a TVF graph. The \textbf{constraint generated by $G$}, $C_{TVF}(G)$, is the set of all assignments to $G$. The \textbf{TVF completion} of a constraint $C$ is $C_{TVF}(G_{TVF}(C))$.
\end{definition}

\begin{lemma}
    For any TVF graph $G$, $G_{TVF}(C_{TVF}(G))=G$.
\end{lemma}

\begin{proof}
    Let $G'=G_{TVF}(C_{TVF}(G))$. By definition, $G$ and $G'$ have the same vertex sets. Let $(u,v)$ be an $ab$ edge of $G$. This is iff for all $\phi\in C_{TVF}(G)$, $(\phi(u),\phi(v))\neq (a,b)$. By definition, this is iff $(u,v)$ is an $ab$ edge of $G'$.
\end{proof}

\begin{lemma}\label{lem:TVF-size}
     Let $G=(V,E_{00}\sqcup E_{11}\sqcup E_{01})$ be a complete TVF graph. Then, $|C_{TVF}(G)|\leq |V|+1$.
\end{lemma}

\begin{proof}
    We proceed by induction on $|V|$. If $|V|=0$, there are no edges, so no constraints, and only one assignment, the vacuous one. Hence $|C_{TVF}(G)|=1=|V|+1$.

    Now, suppose the induction hypothesis holds for $|V|\leq k$. Now let $|V|=k+1$ and let $v\in V$. The vertex $v$ is connected to every element of $V\backslash\{v\}$. Let $\ell\leq k$ be the number of edges that are $0$ on $v$. Hence, given an assignment $\phi$ to $G$, if $\phi(v)=0$, there are $\ell$ vertices $v_1,\ldots,v_\ell$ whose value of $\phi$ is fixed; and if $\phi(v)=1$, the value on the remaining $k-\ell$ vertices $v_{\ell+1},\ldots,v_k$ is fixed. As such,
    \begin{align*}
        |C_{TVF}(G)|&=\abs*{\set*{\phi\in C_{TVF}(G)}{\phi(v)=0}}+\abs*{\set*{\phi\in C_{TVF}(G)}{\phi(v)=1}}\\
        &\leq|C_{TVF}(G\backslash\{v,v_1,\ldots,v_\ell\})|+|C_{TVF}(G\backslash\{v,v_{\ell+1},\ldots,v_k\})|\\
        &\leq (k-\ell+1)+(\ell+1)=k+2=|V|+1,
    \end{align*}
    by induction hypothesis.
\end{proof}

\begin{definition}
    Let $C\subseteq \Z_2^V$ be a constraint. The \textbf{tableau form} of $C$ with respect to orderings $V=\{v_1,\ldots,v_k\}$ and $C=\{\phi_1,\ldots,\phi_n\}$ is the matrix
    $$\begin{bmatrix}\phi_1(v_1)&\phi_1(v_2)&\cdots&\phi_1(v_k)\\\phi_2(v_1)&\phi_2(v_2)&\cdots&\phi_2(v_k)\\\vdots&&\ddots&\vdots\\\phi_n(v_1)&\phi_n(v_2)&\cdots&\phi_n(v_k)\end{bmatrix}.$$
    We say $C$ has a tableau form $M$ if $M$ is the tableau form for some ordering of the variables and satisfied constraints.

    We say that $M$ is \textbf{upper triangular} if it is upper triangular as a matrix.
\end{definition}

%By the above lemma, note that an incompressible constraint $C_{TVF}(G)$ can only have an upper triangular tableau form if it has an all-zero assignment.

\begin{lemma}\label{lem:directed-graph}
    Let $G=(V,E)$ be a directed complete graph with no multiple edges or loops, \emph{i.e.} for all $x,y\in V$ distinct, either $(x,y)\in E$ or $(y,x)\in E$, but not both. Suppose that every vertex of $G$ has an incoming edge, and $|V|\geq 3$. Then $G$ has a cycle.
\end{lemma}

\begin{proof}
    We proceed by induction on $|V|$. For the base case, we have $|V|=3$. If there is a vertex $x\in V$ with two incoming edges $(y,x),(z,x)\in E$. Then, there is an edge between $y$ and $z$, which we may assume without loss of generality is $(y,z)\in E$. Then, $y$ has no incoming edges, which contradicts the hypothesis. As such, every vertex has must have one incoming edge and one outgoing edge. So $G$ is a directed $3$-cycle, and thus contains a cycle.

    Now let $|V|>3$. If every vertex of $G$ has both an incoming and an outgoing edge, then it has a subgraph that is a cycle. Otherwise, there exists a vertex with only incoming edges. Consider the subgraph $H$ of $G$ with that vertex removed. Then, every vertex of $H$ has at least one incoming edge and $H$ has $|V|-1$ vertices. By induction hypothesis, $H$ has a cycle, and therefore so does $G$.
\end{proof}

\begin{proposition}\label{prop:01-edges}
    Let $G$ be an incompressible complete TVF graph with no 00 or 11 edges. Then the constraint $C_{TVF}(G)$ has a tableau form 
    \begin{align*}
        \begin{bmatrix}1&1&1&\cdots&1&1\\0&1&1&\cdots &1&1\\0&0&1&\cdots&1&1\\\vdots&&&\ddots&&\vdots\\0&0&0&\cdots&1&1\\0&0&0&\cdots&0&1\\0&0&0&\cdots&0&0\end{bmatrix}
    \end{align*}
\end{proposition}

As a consequence, for any TVF constraint $C\subseteq\Z_2^V$ whose TVF graph has no $00$ or $11$ edges, there is a tableau form of $C$ that is a subset of the rows of the above matrix.

\begin{proof}
    We construct an ordering on $V$ and $C=C_{TVF}(G)$ recursively. We say that $G_{TVF}(C)=G=(V,E_{01})$ has a directed edge from $x$ to $y$ if $(y,x)\in E_{01}$. Then, note that there must always be a vertex $v_1\in V$ such that all the edges incident to $v_1$ must be pointing outwards. Otherwise, every vertex has at least one edge pointing inwards. Then, due to \Cref{lem:directed-graph}, we know that $G$ has a cycle, and due to \Cref{lem:compressible} this cycle is a compressible subgraph. Then, $G$ is compressible, a contradiction. As such, for each $v\neq v_1$, $(v_1,v)\in E_{01}$. Consider the assignment $\phi_1$ to $G$ such that $\phi_1(v_1)=1$, then $\phi_1(v)=1$ for all $v\in V$. As $\phi_1$ is an assignment to $G$, it is in $C$ by definition. For every other $\phi\in C$, we have $\phi(v_1)=0$. Now, we can continue this recursively with the subgraph of $G$ constructed by removing vertex $v_1$. Suppose we have $v_1,\ldots,v_k$ and some $\phi_1,\ldots,\phi_k$ such that $\phi_j(v_{i})=0$ for $i<j$ and $\phi_j(v)=1$ for all other $v\in V$, and for all other $\phi\in C$, $\phi(v_i)=0$. Then, take $G_k=G\backslash\{v_1,\ldots,v_k\}$. As above there must be a vertex $v_{k+1}$ with only outgoing edges. By maximality of $C$, there exists a unique $\phi_{k+1}\in C$ with $\phi_{k+1}(v_{k+1})=1$ and $\phi_{k+1}(v_i)=0$ for $i<k+1$, and we must have $\phi_{k+1}(v)=1$ for all $v\notin \{v_1,\ldots,v_k\}$. We have that $\phi(v_{k+1})=0$ for all $\phi\in C\backslash\{\phi_1,\ldots,\phi_{k+1}\}$. Hence, we can continue recursively until we have exhausted all the elements of $V$.

    At the end, we have orderings $v_1,\ldots,v_{|V|}$ and $\phi_1,\ldots,\phi_{|V|}$ such that $\phi_j(v_i)=1$ if $i\geq j$ and $0$ otherwise. We can then construct the assignment $\phi_{|V|+1}(v)=0$. This brings the number of elements of $C$ to $|V|+1$, the maximum possible by \Cref{lem:TVF-size}, finishing the proof.
\end{proof}

\begin{corollary}\label{cor:np-hard-edges}
    Let $\Gamma$ be a set of boolean TVF constraints. If $\CSP(\Gamma)_{1,1}$ is $\NP$-complete, then there exists at least one $C\in\Gamma_{\mathrm{comp}}$ whose TVF graph has a $00$ or $11$ edge.
\end{corollary}

\begin{proof}
    Suppose otherwise that every edge of the TVF graph of every $C\in\Gamma_{\mathrm{comp}}$ is a $01$ edge. Then, the TVF completion of $C$ has a tableau form as given in \Cref{prop:01-edges}. Such constraints satisfy the majority polymorphism. In fact, if $i\geq j\geq k$ with $\phi_i,\phi_j,\phi_k\in C$, $\mathrm{MAJ}(\phi_i,\phi_j,\phi_k)=\phi_j\in C$. Therefore, by \Cref{lem:max-compression-np-hard}, every constraint in $\Gamma$ must also satisfy the majority polymorphism. Hence, $\CSP(\Gamma)$ is in $\cP$ by Schaefer's dichotomy theorem. By contrapositive, we get the desired result.
\end{proof}

\begin{lemma}\label{lem:upper-triangular-no-00}
    Suppose that $G=(V,E_{00}\sqcup E_{11}\sqcup E_{01})$ is an incompressible complete TVF graph. If $G$ has no $00$ edges, then $C_{TVF}(G)$ has an upper triangular tableau form.
\end{lemma}

\begin{proof}
    The proof follows the structure of \Cref{prop:01-edges}; that is, we recursively construct an order on the variables and the constraints of $C=C_{TVF}(G)$ to get the wanted form. For the base case, we first show that there is a vertex $v_1\in V$ such that every edge of $G$ incident to $v_1$ is $1$ on it. Suppose otherwise that for every vertex of $v$, there is a $01$ edge that is $0$ on $v$. Start with an arbitrary vertex $v$ and follow one of the $01$ edges which is $0$ on $v$ to the next vertex. Then repeat this procedure. Because the graph is finite, eventually we visit a vertex twice. But this induces a cycle satisfying the conditions of \Cref{lem:compressible}, so $G$ is compressible, a contradiction. Looking at the subgraph on the vertex set $V\backslash\{v_1\}$ we can apply the same reasoning to find a vertex $v_2$. Then, continuing recursively, we find vertices $v_1,\ldots,v_{|V|}$ such that the (unique) edge from $v_i$ to $v_j$ is $1$ on $v_i$ iff $i<j$. Thus, for any assignment $\phi\in C$, if $\phi(v_i)=1$, then the value $\phi(v_j)$ is fixed for all $j>i$. As such, for each $i$, there is exactly one $\phi_i\in C$ such that $\phi_i(v_k)=0$ for all $k<i$, and $\phi_i(v_i)=1$. Finally, the zero assignment $\phi_{|V|+1}(v)=0$ gives the remaining element of $C$.
\end{proof}

\begin{definition}
    Let $r:V\rightarrow W\cup\Sigma$. Given a map $\phi\in\Sigma^W$, the \textbf{augmented composition} of $\phi$ with $r$ is $\phi\circ r\in\Sigma^V$ defined as
    $$(\phi\circ r)(v)=\begin{cases}\phi(r(v))&r(v)\in W\\r(v)&r(v)\in\Sigma\end{cases}.$$
    The \textbf{augmented pushforward} of a constraint $C\in\Sigma^V$ by $r$ is $r_\ast C=\set*{\phi\in\Sigma^W}{\phi\circ r\in C}$. We say that $C$ \textbf{simulates} a constraint $C'\subseteq\Sigma^W$ if there exists $W'\supseteq W$ and $r:V\rightarrow W'\cup\Sigma$ such that $C'=r_\ast C|_W$.

    Let $r:V\rightarrow W\cup\lnot W\cup\Z_2$. Given a map $\phi\in\Z_2^W$, the \textbf{augmented composition (with negation)} of $\phi$ with $r$ is $\phi\circ r\in\Z_2^V$ defined as
    $$(\phi\circ r)(v)=\begin{cases}\phi(r(v))&r(v)\in W\\\lnot\phi(w)&r(v)=\lnot w\in\lnot W\\r(v)&r(v)\in\Z_2\end{cases}.$$
    The \textbf{augmented pushforward (with negation)} of a constraint $C\in\Z_2^V$ by $r$ is $r_\ast C=\set*{\phi\in\Z_2^W}{\phi\circ r\in C}$. We say that $C$ \textbf{simulates $C'\subseteq\Z_2^W$ with negation} if there exists $W'\supseteq W$ and $r:V\rightarrow W'\cup\lnot W'\cup\Z_2$ such that $C'=r_\ast C|_W$.
\end{definition}

\begin{proposition}\label{prop:TVF-simulates}
    Suppose that $C\subseteq\Z_2^V$ is a TVF constraint that does not satisfy the majority polymorphism and whose TVF graph does not have a $00$ edge. Then, $C$ either simulates $\{(1,0,0), (0,1,0), (0,0,1)\}$ or $\{(1,0,1), (0,1,1), (0,0,0)\}$.
\end{proposition}

\begin{proof}
    By \Cref{lem:upper-triangular-no-00}, the TVF completion of $C$ has an upper triangular tableau form, induced by orderings $V=\{v_1,\ldots,v_{|V|}\}$ and $C_{TVF}(G_{TVF}(C))=\{\phi_1,\ldots,\phi_{|V|+1}\}$. Let $I\subseteq[|V|+1]$ be the set of indices of elements of $C_{TVF}(G_{TVF}(C))$ that are in $C$. By construction of the upper triangular tableau, for all rows $i,j,k\in[|V|+1]$ there exists $r_{ijk}\in[|V|+1]$ such that $\mathrm{MAJ}(\phi_i,\phi_j,\phi_k)=\phi_{r_{ijk}}$. By hypothesis, there exist three rows $i<j<k$ in $I$ such that $r_{ijk}\notin I$. First, we know that for all $l<j$, $\phi_j(v_l)=\phi_k(v_l)=0$, so $\mathrm{MAJ}(\phi_i,\phi_j,\phi_k)(v_l)=0$. We know that $\phi_j(v_j)=1$ and $\phi_k(v_j)=0$. We claim that $\phi_i(v_j)=0$ as well. In fact, if $\phi_i(v_j)=1$, then by construction we know that $\phi_i(v_l)=\phi_j(v_l)$ for all $l\geq j$. Hence, the majority $\mathrm{MAJ}(\phi_i,\phi_j,\phi_k)=\phi_j$, a contradiction.
    
    Now consider two cases. Suppose first that there is some $j<h<k$ such that $\phi_i(v_h)=\phi_j(v_h)=1$. We can also suppose that $h$ is the smallest index satisfying this property. Then, we know that $\phi_i(v_l)=\phi_j(v_l)$ for $l\geq h$. As such, $r_{ijk}=h$, so $h\notin I$. Now, define a map $r:V\rightarrow\{x,y,z,z'\}\cup\Z_2$ as follows:
    \begin{equation*}
        r(v_l) = \begin{cases}
                    \phi_i(v_l)& \text{if }\phi_i(v_l)=\phi_j(v_l)=\phi_k(v_l)
                    \\
                    x& \text{if }\phi_i(v_l)=1 \text{ and } \phi_j(v_l)=\phi_k(v_l)=0
                    \\
                    y& \text{if }\phi_j(v_l)=1 \text{ and } \phi_i(v_l) = \phi_k(v_l) = 0
                    \\
                    z'& \text{if }\phi_k(v_l)=1 \text{ and }\phi_i(v_l)=\phi_j(v_l)=0
                    \\
                    z& \text{if }\phi_i(v_l)=\phi_j(v_l)=1\text{ and }\phi_k(v_l)=0
                \end{cases}.
    \end{equation*}
    Note that no other cases are possible as $\phi_i(v_l)=\phi_j(v_l)$ for all $l\geq h$, while $\phi_k(v_l)=0$ for all $l<h$. Let $\phi\in\Z_2^{\{x,y,z,z'\}}$ be such that $\psi=\phi\circ r\in C$. Since $\psi(v_l)=0$ for all $l<i$, $\psi=\phi_t$ for some $t\geq i$. First, we claim that if $\phi(x)=1$, then $\phi(y)=0$ and $\phi(z)=1$. Since if $\phi(x)=1$, then $\psi(v_{i})=1$, so $\psi=\phi_i$. As such, $\psi(v_{j})=0$ and $\psi(v_h)=1$ as wanted. Next, suppose that $\phi(x)=0$. Then, if $\phi(y)=1$, we have that $\psi(v_{j})=1$ and $\psi(v_l)=0$ for all $l<j$. As such, $\psi=\phi_j$ and $\psi(v_h)=1$, giving that $\phi(z)=1$. On the other hand, if $\phi(x)=\phi(y)=0$ and $\phi(z)=1$ we would have that $\psi(v_h)=1$ and $\psi(v_l)=0$ for all $l<h$. This does not correspond to any $\phi_t$ for $t\in I$ and therefore we must have $\phi(z)=0$. Hence, restricting to the variables $\{x,y,z\}$, we find that $C$ simulates $\{(1,0,1), (0,1,1), (0,0,0)\}$.
    
    Now, suppose that $\phi_i(v_l)$ and $\phi_j(v_l)$ are not both $1$ for all $j<l<k$. Note that we must have $\phi_i(v_k)=\phi_j(v_k)=0$, as otherwise $\mathrm{MAJ}(\phi_i,\phi_j,\phi_k)=\phi_k$. Let $h$ be the minimal $l$ such that two of $\phi_i(v_l),\phi_j(v_l),\phi_k(v_l)$ are equal to $1$. Then, $r_{ijk}=h$ and $h\notin I$. 
    %Let $m$ be the minimal $l$ such that $M_{il}=M_{jl}=M_{kl}=1$.
    Now, define $r:V\rightarrow\{x,y,z,z'\}\cup\Z_2$ as follows:
    \begin{equation*}
        r(v_l) = \begin{cases}
                    \phi_i(v_l)& \text{if }\phi_i(v_l)=\phi_j(v_l)=\phi_k(v_l)
                    \\
                    x& \text{if }\phi_i(v_l)=1 \text{ and } \phi_j(v_l)=\phi_k(v_l)=0
                    \\
                    y& \text{if }\phi_j(v_l)=1 \text{ and } \phi_i(v_l)=\phi_k(v_l)=0
                    \\
                    z& \text{if }\phi_k(v_l)=1 \text{ and } \phi_i(v_l)=\phi_j(v_l)=0
                    \\
                    z'& \text{if only two of }\phi_i(v_l),\phi_j(v_l) \text{ and }\phi_k(v_l) \text{ are equal to } 1
                \end{cases}.
    \end{equation*}
    Let $\phi\in\Z_2^{\{x,y,z,z'\}}$ be such that $\psi=\phi\circ r\in C$. We claim that at most one of $\phi(x), \phi(y), \phi(z)$ can be $1$. If $\phi(x)=1$, then $\psi=\phi_i$, so $\phi(y)=\phi(z)=0$. Next, if $\phi(y)=1$ we must have $\phi(x)=0$ to not contradict the above, and therefore $\phi=\phi_j$ so $\phi(z)=0$. Finally, if $\phi(z)=1$, we have by the above that $\phi(x)=\phi(y)=0$. To complete the argument consider two cases. Suppose that there is no $\phi\in r_\ast C$ such that $\phi(x)=\phi(y)=\phi(z)=0$. Then, via $r$, $C$ simulates $r_\ast C|_{\{x,y,z\}}=\{(1,0,0),(0,1,0), (0,0,1)\}$. For the second case, suppose that there is such an element $\phi\in r_\ast C$. For this element, note that if $\phi(z')=1$, $\psi=\phi_h$. As $h\notin I$, this implies that we must have $\phi(z')=0$. Therefore, we have that $r_\ast C$ is one of the following constraints: $\{(1,0,0,0),(0,1,0,1), (0,0,1,1),(0,0,0,0)\}$, $\{(1,0,0,1),(0,1,0,0), (0,0,1,1),(0,0,0,0)\}$, $\{(1,0,0,1),(0,1,0,1), (0,0,1,0),(0,0,0,0)\}$. Take $s:\{x,y,z,z'\}\cup\Z_2\rightarrow\{x,y,z\}\cup\Z_2$, where $s(a)=a$ for $a\in\Z_2$ and $s(z')=z$. In the first case, $s(x)=0$, $s(y)=x$, $s(z)=y$; in the second case $s(x)=x$, $s(y)=0$, $s(z)=y$; and in the third case, $s(x)=x$, $s(y)=y$, $s(z)=0$. In all three cases, via $s\circ r$, $C$ simulates $\{(1,0,1),(0,1,1),(0,0,0)\}$.
\end{proof}

\begin{definition}
    Let $\phi\in\Z_2^V$ and $U\subseteq V$. The \textbf{negation} of $\phi$ at $U$ is the map $\phi_{\lnot U}\in\Z_2^V$ defined as
    $$\phi_{\lnot U}(v)=\begin{cases}\lnot\phi(v)&v\in U\\\phi(v)&\text{otherwise}\end{cases}.$$
    
    The \textbf{negation} of a constraint $C\subseteq\Z_2^V$ at $U$ is the constraint $C_{\lnot U}=\set*{\phi_{\lnot U}}{\phi\in C}.$
    
    %We say that a set of boolean constraints $\Gamma$ is \textbf{closed under negation} if for every $(V,C)\in\Gamma$ and $U\subseteq V$, $(V,C_{\lnot U})\in\Gamma$. The \textbf{negation closure} of a set of boolean constraints $\Gamma$ is the smallest set of constraints closed under negation that contains $\Gamma$.
\end{definition}

\begin{lemma} \label{lem:triangular-negation}
    Let $G=(V,E_{00}\sqcup E_{11}\sqcup E_{01})$ be an incompressible complete TVF graph. There exists $U\subseteq V$ such that $C_{TVF}(G)_{\lnot U}$ has an upper triangular tableau form.
\end{lemma}

\begin{proof}
    The proof again follows the structure of \Cref{prop:01-edges}; that is, we recursively construct an order on the variables and the constraints of $C=C_{TVF}(G)$ to get the wanted form. Here, we also construct the subset $U\subseteq V$ at the same time. For the base case, we first show that there is a vertex $v_1\in V$ such that all edges of $G$ incident to $v_1$ are $0$ on $v$ or $1$ on $v$. Suppose otherwise that for every vertex $v$ of $G$, there is an edge that is $0$ on $v$ and an edge that is $1$ on $v$. Start with an arbtitrary vertex $u_1$ and follow any edge to the next vertex $u_2$. This edge is $b$ on $u_2$ for some $b\in\Z_2$, so we can pick an edge that is $\lnot b$ on $u_2$, connecting to the next vertex $u_3$. Then we repeat this procedure. Because the graph is finite, eventually we visit a vertex twice. But this induces a cycle satisfying the conditions of \Cref{lem:compressible}, so $G$ is compressible, a contradiction. If every edge incident to $v_1$ is $0$ on $v_1$, we pass to the negated constraint $C_{\lnot\{v_1\}}$, so that now the edges are $1$ on $v_1$. Looking at the subgraph on the vertex set $V\backslash\{v_1\}$ we can apply the same reasoning to find a vertex $v_2$. Then, continuing recursively, we find vertices $v_1,\ldots,v_{|V|}$ such that the edge from $v_i$ to $v_j$ in the TVF graph of $C_{\lnot U}$ is $1$ on $v_i$ if $i<j$. Thus, for any assignment $\phi\in C_{\lnot U}$, if $\phi(v_i)=1$, then $\phi(v_j)$ can only take one value for all $j>i$. As such, for each $i$, there is exactly one $\phi_i\in C_{\lnot U}$ such that $\phi_i(v_k)=0$ for all $k<i$ and $\phi_i(v_i)=1$. Finally, the zero assignment $\phi_{|V|+1}(v)=0$ gives the remaining element of $C_{\lnot U}$.
\end{proof}

\begin{proposition}\label{prop:TVF-simulates-negation}
    Suppose that $C\subseteq\Z_2^V$ is a TVF constraint that does not satisfy the majority polymorphism. Then, $C$ simulates $\{(1,0,0), (0,1,0), (0,0,1)\}$ with negation.
\end{proposition}

\begin{proof}
    Let $U\subseteq V$ be such that the TVF completion of $C_{\lnot U}$ has an upper triangular tableau form induced by orderings $V=\{v_1,\ldots,v_{|V|}\}$ and $C_{TVF}(G_{TVF}(C))_{\lnot U}=\{\phi_{1},\ldots,\phi_{|V|+1}\}$, guaranteed by \Cref{lem:triangular-negation}. Let $I\subseteq[|V|+1]$ be set of indices of elements of $C_{TVF}(G_{TVF}(C))_{\lnot U}$ that are in $C_{\lnot U}$. By construction, we know that for all $i,j,k\in[|V|+1]$, there exists $r_{ijk}$ such that $\mathrm{MAJ}(\phi_i,\phi_j,\phi_k)=\phi_{r_{ijk}}$. Since $C$ does not satisfy the majority polymorphism, neither does $C_{\lnot U}$, and hence there exist $i<j<k$ in $I$ such that $r_{ijk}\notin I$. Define $r:V\rightarrow\{x,y,z,\lnot x,\lnot y,\lnot z\}\cup\Z_2$ as follows: 
    
    $$r(v_l)=\begin{cases}\phi_{i}(v_l)&\text{if }\phi_{i}(v_l)=\phi_{j}(v_l)=\phi_{k}(v_l)\text{ and }v_l\notin U
    \\
    \lnot \phi_{i}(v_l)&\text{if }\phi_{i}(v_l)=\phi_{j}(v_l)=\phi_{k}(v_l)\text{ and }v_l\in U
    \\
    x&\text{if }(\phi_i(v_l)=1, \phi_j(v_l)=\phi_k(v_l)=0, v_l\notin U)\text{ or }(\phi_{i}(v_l)=0, \phi_{j}(v_l)=\phi_{k}(v_l)=1, v_l\in U)
    \\
    \lnot x&\text{if }(\phi_{i}(v_l)=0, \phi_{j}(v_l)=\phi_{k}(v_l)=1, v_l\notin U)\text{ or }(\phi_{i}(v_l)=1, \phi_{j}(v_l)=\phi_{k}(v_l)=0, v_l\in U)
    \\
    y&\text{if }(\phi_{j}(v_l)=1, \phi_{k}(v_l)=\phi_{i}(v_l)=0, v_l\notin U)\text{ or }(\phi_{j}(v_l)=0, \phi_{k}(v_l)=\phi_{i}(v_l)=1, v_l\in U)
    \\
    \lnot y&\text{if }(\phi_{j}(v_l)=0, \phi_{k}(v_l)=\phi_{i}(v_l)=1, v_l\notin U)\text{ or }(\phi_{j}(v_l)=1, \phi_{k}(v_l)=\phi_{i}(v_l)=0, v_l\in U)
    \\
    z&\text{if }(\phi_{k}(v_l)=1, \phi_{i}(v_l)=\phi_{j}(v_l)=0, v_l\notin U)\text{ or }(\phi_{k}(v_l)=0, \phi_{i}(v_l)=\phi_{j}(v_l)=1, v_l\in U)
    \\
    \lnot z&\text{if }(\phi_{k}(v_l)=0, \phi_{i}(v_l)=\phi_{j}(v_l)=1, v_l\notin U)\text{ or }(\phi_{k}(v_l)=1, \phi_{i}(v_l)=\phi_{j}(v_l)=0, v_l\in U)\end{cases}.$$
    First, let $\phi=(0,0,0)\in\Z_2^{\{x,y,z\}}$. We claim that $\phi\notin r_\ast C$. Let $\psi=\phi\circ r$. Then, for $v_l\notin U$, $\psi(v_l)=1$ if two of $\phi_{i}(v_l),\phi_{j}(v_l),\phi_{k}(v_l)$ are $1$; and for $v_l\in U$, $\psi(v_l)=1$ if two of $\phi_{i}(v_l),\phi_{j}(v_l),\phi_{k}(v_l)$ are $0$. Hence, $\psi_{\lnot U}=\phi_{r_{ijk}}$, and thus $\psi\notin C$. So $(0,0,0)\notin r_\ast C$. Now suppose that $\phi\in r_\ast C$ with $\phi(x)=1$. We claim that $\phi(y)=\phi(z)=0$. Since $\phi(x)=1$, $\phi\circ r=\phi_i$. Since $r_{ijk}\neq j$, we know that $\phi_i(v_j)=0$ so $\phi(y)=0$. If there exists $j<h<k$ such that $\phi_i(v_h)=\phi_j(v_h)=1$, then $\lnot\phi(z)=1$ so $\phi(z)=0$. Otherwise, we know that $r_{ijk}\neq k$, so $\phi_i(v_k)=\phi_j(v_k)=0$, giving $\phi(z)=0$ as well. Next, suppose $\phi(y)=1$. By the above, we must have $\phi(x)=0$. Then, we have that $\phi\circ r=\phi_j$, so by the same two-case argument as above $\phi(z)=0$. Finally, note that it is possible to have $\phi(z)=1$ and $\phi(x)=\phi(y)=0$, as $\phi\circ r=\phi_k$ in that case. By the above, this is the only possible $\phi\in r_\ast C$ with $\phi(z)=1$. As such, we have that $r_\ast C=\{(1,0,0),(0,1,0),(0,0,1)\}$, as wanted.
\end{proof}

\subsection{The general commutativity gadget}

In this section, we show that the simplification and simulation arguments of the previous section are quantum-sound, and use this to construct a general commutativity gadget modelled on that of \Cref{lem:basic-commutativity-gadget}.

First, we want to show the constraints in the maximal compression can be expressed in terms of the original set of constraints, in a quantum-sound way. Then, we can work only with the maximal compression to construct gadgets.

\begin{lemma}\label{lem:the-bends}
    Let $\Gamma$ be a set of constraints and let $\Gamma_{\max}$ be its maximal compression. For each CS $S=(X,\{(V_i,{r_i}_\ast C_i)_{i=1}^m)\in\CSP(\Gamma_{\max})$ and probability distribution $\pi$ on $[m]$, there exists a CS $S'=(X',\{(V_i',{r_i'}_\ast C_i')_{i=1}^{m})\in\CSP(\Gamma)$ such that there exists a $L$-homomorphism $\alpha:\mc{A}_{c-v}(S,\pi)\rightarrow\mc{A}_{c-v}(S',\pi)$, where $L=\max_{(V,C)\in\Gamma}|V|$.
\end{lemma}

\begin{proof}
    By definition of $\Gamma_{\max}$, for each $C_i$, there exists a $C_i'\in\Gamma$ such that $C_i'|_{V_{C_i}}=C_i$. Then, let $X'=X\cup\set*{x_{i,v}}{v\in V_{C_i'}\backslash V_{C_i}}$, let $V_i'=V_i\cup\set*{x_{i,v}}{v\in V_{C_i'}\backslash V_{C_i}}$, and let $r_i'|_{V_{C_i}}=r_i$ and for $v\in V_{C_i'}\backslash V_{C_i}$, $r_i'(v)=x_{i,v}$. Now, define $\alpha$ on $\mc{A}(V_i,{r_i}_{\ast}C_i)$ via $\alpha(\Phi_{V_i,\phi})=\sum_{\psi\in {r_i'}_\ast C_i',\psi|_{V_i}=\phi}\Phi_{V_i',\psi}$ and as identity on $\C\Z_2^{\ast X}$. Then,
    \begin{align*}
        &\alpha\Big(\sum_{i=1}^m\frac{\pi(i)}{|V_i|}\sum_{x\in V_i,\phi\in {r_i}_\ast C_i}\hsq*{\Phi_{V_i,\phi}(1-\Pi_{\phi(x)}(\sigma'(x)))}\Big)
        \\
        &=\sum_{i=1}^m\frac{\pi(i)}{|V_i|}\sum_{\substack{x\in V_i,\phi\in {r_i}_\ast C_i\\\psi\in{r_i}_\ast C_i',\;\psi|_{V_i}=\phi}}\hsq*{\Phi_{V_i',\psi}(1-\Pi_{\phi(x)}(\sigma'(x)))}
        \\
        &=\sum_{i=1}^m\frac{\pi(i)}{|V_i|}\sum_{x\in V_i,\phi\in {r_i'}_\ast C_i'}\hsq*{\Phi_{V_i',\phi}(1-\Pi_{\phi(x)}(\sigma'(x)))}
        \\
        &\leq L \sum_{i=1}^m\frac{\pi(i)}{|V_i'|}\sum_{x\in V_i',\phi\in {r_i'}_\ast C_i'}\hsq*{\Phi_{V_i',\phi}(1-\Pi_{\phi(x)}(\sigma'(x)))}.
    \end{align*}
\end{proof}

Next, we show that if a constraint can be simulated, it can also be done in a quantum-sound way, assuming that there are constraints that set variables to constants. Similarly, we also show that using a negation constraint, any negation of a constraint can be simulated in a quantum-sound way.

\begin{lemma}\label{lem:simulation-quantum-sound}
    Suppose that $C\subseteq\Z_2^V$ simulates $C'\subseteq\Z_2^W$ via $r:V\rightarrow W'\cup\Z_2$. Consider the BCS $S=(W'\cup\{x_0,x_1\},\{(V_i,C_i)\}_{i=1}^3)$, where $V_1=W'\cup\{x_0,x_1\}$, $C_1=s_\ast C$ with $s(v)=r(v)$ for $r(v)\in W'$ and $s(v)=x_{r(v)}$ for $r(v)\not\in W'$, $V_2=\{x_0\}$, $C_2=\{0\}$, $V_3=\{x_1\}$, and $C_3=\{1\}$. There exists a $24(|W'|+2)$-homomorphism $\alpha:\mc{A}_{c-v}((W,\{(W,C')\}),\mbb{u}_1)\rightarrow\mc{A}_{c-v}(S,\mbb{u}_3)$.
\end{lemma}

\begin{proof}
    Let $\alpha$ be the natural embedding $\mc{A}_{c-v}((W,\{(W,C')\}))\hookrightarrow\mc{A}_{c-v}(S)$. First, note that as $C'=r_\ast C|_W$, $\Phi_{W,\phi}=\sum_{\psi\in r_\ast C,\psi|_W=\phi}\Phi_{W',\psi}$. Therefore,
    \begin{align*}
        \sum_{x\in W,\phi\in C'}\hsq*{\Phi_{W,\phi}(1-\Pi_{\phi(x)}(\sigma'(x)))}&\leq\sum_{x\in W,\phi\in r_\ast C}\hsq*{\Phi_{W',\phi}(1-\Pi_{\phi(x)}(\sigma'(x)))}.
    \end{align*}
    Next, note that any $\phi\in r_\ast C$ admits an extension to $\phi'\in s_\ast C$ by setting $\phi'(x_0)=0$, $\phi'(x_1)=1$, and $\phi'|_{W'}=\phi$. Then, $\Phi_{V_1,\phi'}=\Phi_{W',\phi}\Pi_{0}(\sigma_1(x_0))\Pi_1(\sigma_1(x_1))$. Then, since $\Pi_0(\sigma_2(x_0))=1$, we get 
    \begin{align*}
        \hsq*{\Pi_0(\sigma_1(x_0))-1}&\leq 2\hsq*{\Pi_0(\sigma_1(x_0))-\Pi_0(\sigma'(x_0))}+2\hsq*{\Pi_0(\sigma'(x_0))-1}\\
        &\lesssim 2\sum_b\hsq*{\Pi_b(\sigma_1(x_0))(1-\Pi_b(\sigma'(x_0)))}+2\hsq*{\Pi_0(\sigma_2(x_0))(1-\Pi_0(\sigma'(x_0)))}\\
        &=2\sum_{\phi\in C_1}\hsq*{\Phi_{V_1,\phi}(1-\Pi_{\phi(x_0)}(\sigma'(x_0)))}+2\sum_{x\in V_2,\phi\in C_2}\hsq*{\Phi_{V_2,\phi}(1-\Pi_{\phi(x)}(\sigma'(x)))}.
    \end{align*}
    In the same way, $$\hsq*{\Pi_1(\sigma_1(x_1))-1}\lesssim2\sum_{\phi\in C_1}\hsq*{\Phi_{V_1,\phi}(1-\Pi_{\phi(x_1)}(\sigma'(x_1)))}+2\sum_{x\in V_3,\phi\in C_3}\hsq*{\Phi_{V_3,\phi}(1-\Pi_{\phi(x)}(\sigma'(x)))}.$$
    Putting these together,
    \begin{align*}
        \hsq*{\Pi_0(\sigma_1(x_0))\Pi_1(\sigma_1(x_1))-1}&\leq2\hsq*{\Pi_0(\sigma_1(x_0))\Pi_1(\sigma_1(x_1))-\Pi_0(\sigma_1(x_0))}+2\hsq*{\Pi_0(\sigma_1(x_0))-1}\\
        &\leq 2\hsq*{\Pi_1(\sigma_1(x_1))-1}+2\hsq*{\Pi_0(\sigma_1(x_0))-1}\\
        &\lesssim 4\sum_{i=1}^3\sum_{\substack{\phi\in C_i\\x=x_0,x_1}}\hsq*{\Phi_{V_i,\phi}(1-\Pi_{\phi(x)}(\sigma'(x)))}.
    \end{align*}
    As such, we get the result
    \begin{align*}
        &\sum_{x\in W,\phi\in C'}\hsq*{\Phi_{W,\phi}(1-\Pi_{\phi(x)}(\sigma'(x)))}\leq\sum_{x\in W,\phi\in r_\ast C}\hsq*{\Phi_{W',\phi}(1-\Pi_{\phi(x)}(\sigma'(x)))}\\
        &\leq\sum_{x\in W,\phi\in r_\ast C}\hsq*{(\Phi_{V_1,\phi'}-\Phi_{W',\phi}(\Pi_0(\sigma_1(x_0))\Pi_1(\sigma_1(x_1))-1))(1-\Pi_{\phi(x)}(\sigma'(x)))}\\
        &\lesssim 2\sum_{x\in W,\phi\in r_\ast C}\hsq*{\Phi_{V_1,\phi'}(1-\Pi_{\phi'(x)}(\sigma'(x)))}+2|W|\hsq*{\Pi_0(\sigma_1(x_0))\Pi_1(\sigma_1(x_1))-1}\\
        &\lesssim 2\sum_{x\in W',\phi\in C_1}\hsq*{\Phi_{V_1,\phi}(1-\Pi_{\phi(x)}(\sigma'(x)))}+8|W|\sum_{i=1}^3\sum_{\substack{\phi\in C_i\\x=x_0,x_1,}}\hsq*{\Phi_{V_i,\phi}(1-\Pi_{\phi(x)}(\sigma'(x)))}\\
        &\leq 8|W|\sum_{i=1}^3\sum_{x\in V_i,\phi\in C_i}\hsq*{\Phi_{V_i,\phi}(1-\Pi_{\phi(x)}(\sigma'(x)))}\\
        &\leq 8|W||V_1|\sum_{i=1}^3\frac{1}{|V_i|}\sum_{x\in V_i,\phi\in C_i}\hsq*{\Phi_{V_i,\phi}(1-\Pi_{\phi(x)}(\sigma'(x)))}.
    \end{align*}
\end{proof}

\begin{lemma}\label{lem:negation-quantum-sound}
    Let $C\subseteq\Z_2^V$ be a constraint and $U\subseteq V$, and suppose $V$ is ordered as $V=\{v_1,\ldots,v_{|V|}\}$. Let $I=\set*{i}{v_i\in U}$ Consider the constraint system $S=(\{v_1,w_1,\ldots,v_{|V|},w_{|V|}\}, \{(V_i,C_i)\}_{i=0}^{|V|})$ where $V_0=V$, $C_0=C$, and $V_i=\{v_i,w_i\}$ and $C_i=C_{\neq}$ for $i>0$. There exists a $4(|V|+1)$-homomorphism $\alpha:\mc{A}_{c-v}((V,\{(V,C_{\lnot U})\}),\mbb{u}_1)\rightarrow\mc{A}_{c-v}(S,\mbb{u}_{|V|+1})$.
\end{lemma}

\begin{proof}
    Define $\alpha$ on $\mc{A}(V,C)$ as $\alpha(v_i)=\sigma_0(v_i)$ if $i\notin I$ and $\alpha(v_i)=-\sigma_0(v_i)$ if $i\in I$; let $\alpha(\sigma'(v_i))=\sigma'(v_i)$ if $i\notin I$ and $\alpha(\sigma'(v_i))=\sigma'(w_i)$ if $i\in I$. First, we get that
    \begin{align*}
        \alpha\Big(\sum_{x\in V,\phi\in C_{\lnot U}}&\hsq*{\Phi_{V,\phi}(1-\Pi_{\phi(x)}(\sigma'(x)))}\Big)\\
        &=\sum_{i\in I,\phi\in C_{\lnot U}}\hsq*{\Phi_{V_0,\phi_{\lnot U}}(1-\Pi_{\phi(v_i)}(\sigma'(w_i)))}+\sum_{i\notin I,\phi\in C_{\lnot U}}\hsq*{\Phi_{V_0,\phi_{\lnot U}}(1-\Pi_{\phi(v_i)}(\sigma'(v_i)))}\\
        &=\sum_{i\in I,\phi\in C}\hsq*{\Phi_{V_0,\phi}(1+\Pi_{\phi(v_i)}(\sigma'(w_i)))}+\sum_{i\notin I,\phi\in C}\hsq*{\Phi_{V_0,\phi}(1-\Pi_{\phi(v_i)}(\sigma'(v_i)))}
    \end{align*}
    Next, in $\mc{A}_{c-v}(S)$, we have that for $i>0$, $\sigma_i(v_i)=-\sigma_i(w_i)$, so
    \begin{align*}
        \hsq*{\sigma'(v_i)+\sigma'(w_i)}&\leq4\hsq*{\sigma_i(v_i)-\sigma'(v_i)}+4\hsq*{\sigma_i(w_i)-\sigma'(w_i)}\\
        &\leq16\sum_{x\in V_i,\phi\in C_i}\hsq*{\Phi_{V_i,\phi}(1-\Pi_{\phi(x)}(\sigma'(x)))}.
    \end{align*}
    Putting these together,
    \begin{align*}
        &\sum_{i\in I,\phi\in C}\hsq*{\Phi_{V_0,\phi}(1+\Pi_{\phi(v_i)}(\sigma'(w_i)))}\leq2\sum_{i\in I,\phi\in C}\parens*{\hsq*{\Phi_{V_0,\phi}(1-\Pi_{\phi(v_i)}(\sigma'(v_i)))}+\frac{1}{4}\hsq*{\Phi_{V_0,\phi}(\sigma'(v_i)+\sigma'(w_i))}}\\
        &\leq2\sum_{i\in I,\phi\in C}\hsq*{\Phi_{V_0,\phi}(1-\Pi_{\phi(v_i)}(\sigma'(v_i)))}+4\sum_{i\in I}\sum_{x\in V_i,\phi\in C_i}\hsq*{\Phi_{V_i,\phi}(1-\Pi_{\phi(x)}(\sigma'(x)))}.
    \end{align*}
\end{proof}

Now, in order to construct the constant and negation constraints, we appeal to the structure of TVF graphs of incompressible constraints.

\begin{lemma}\label{lem:zero-quantum-sound}
    Let $C\subseteq\Z_2^V$ be an incompressible TVF constraint whose TVF graph has a $11$ edge between vertices $u,v\in V$. Let $r:V\rightarrow V\backslash\{v\}$ be defined $r(v)=u$ and $r(w)=w$ for all $w\neq v$. Then, there is a $(|V|-1)$-homomorphism $\alpha:\mc{A}_{c-v}((\{x\},\{(\{x\},\{0\})\}),\mbb{u}_1)\rightarrow\mc{A}_{c-v}(\{V\backslash\{v\},\{(V\backslash\{v\},r_\ast C)\}),\mbb{u}_1)$.
\end{lemma}

\begin{proof}
    Noting that $\sigma_1(x)=1$, define $\alpha$ by $\alpha(\sigma'(x))=\sigma'(u)$. Then, we have that
    \begin{align*}
        \alpha\Big(\hsq*{1-\Pi_0(\sigma'(x))}\Big)=\hsq*{1-\Pi_0(\sigma'(u))}.
    \end{align*}
    On the other hand, noting that $\phi\in r_\ast C$ implies that $\phi(u)\neq 1$,
    \begin{align*}
        \sum_{w\in V\backslash\{v\},\phi\in r_\ast C}\hsq*{\Phi_{V\backslash\{v\},\phi}(1-\Pi_{\phi(w)}(\sigma'(w)))}&\geq\sum_{\phi\in r_\ast C}\hsq*{\Phi_{V\backslash\{v\},\phi}(1-\Pi_{\phi(u)}(\sigma'(u)))}\\
        &=\hsq*{1-\Pi_{0}(\sigma'(u))},
    \end{align*}
    giving the wanted result.
\end{proof}

\begin{lemma}\label{lem:one-quantum-sound}
    Let $C\subseteq\Z_2^V$ and $r$ be as in the previous lemma, and let $C'\subseteq\Z_2^W$ be an incompressible nonempty constraint that does not contain the all-$0$ assignment. In particular, there exists $W_0\subset W$ and $\phi_0\in C'$ such that $\phi_0(w)=0$ if $w\in W_0$ and $\phi_0(w)=1$ otherwise. Consider the constraint system $S=(V\backslash\{v\}\cup\{u'\},\{(V\backslash\{v\},r_\ast C),(\{u,u'\},s_\ast C')\})$ where $s(w)=u$ if $w\in W_0$ and $s(w)=u'$ otherwise. Then, there exists a $8(|V|-1)$-homomorphism $\alpha:\mc{A}_{c-v}((\{y\},\{(\{y\},\{1\})\}),\mbb{u}_1)\rightarrow\mc{A}_{c-v}(S,\mbb{u}_2)$.
\end{lemma}

\begin{proof}
    As in the previous lemma, $\sigma_1(y)=-1$, so we define $\alpha$ by $\alpha(\sigma'(y))=\sigma'(u')$, giving that $\alpha\Big(\hsq*{1-\Pi_1(\sigma'(y))}\Big)=\hsq*{1-\Pi_1(\sigma'(u'))}$. As before, we have $$\sum_{\substack{w\in V\backslash\{v\}\\\phi\in r_\ast C}}\hsq*{\Phi_{V\backslash\{v\},\phi}(1-\Pi_{\phi(w)}(\sigma'(w)))}\geq \hsq*{1-\Pi_{0}(\sigma'(u))}.$$ Next, note that $\Phi_{\{u,u'\},(0,0)}=0$ by hypothesis so $\Pi_1(\sigma_2(u'))\geq\Pi_0(\sigma_2(u))$. Then,
    \begin{align*}
        \hsq*{1-\Pi_1(\sigma'(u'))}&\leq 2\hsq*{1-\Pi_1(\sigma_2(u'))}+2\hsq*{\Pi_{1}(\sigma_2(u'))-\Pi_1(\sigma'(u'))}\\
        &\leq2\hsq*{1-\Pi_0(\sigma_2(u))}+\frac{1}{2}\hsq*{\sigma_2(u')-\sigma'(u')}\\
        &\leq4\hsq*{1-\Pi_0(\sigma'(u))}+\hsq*{\sigma_2(u)-\sigma'(u)}+\hsq*{\sigma_2(u')-\sigma'(u')}\\
        &\leq4\sum_{w\in V\backslash\{v\},\phi\in r_\ast C}\hsq*{\Phi_{V\backslash\{v\},\phi}(1-\Pi_{\phi(w)}(\sigma'(w)))}\\
        &\qquad+4\sum_{w\in\{u,u'\},\phi\in s_\ast C'}\hsq*{\Phi_{\{u,u'\},\phi}(1-\Pi_{\phi(2)}(\sigma'(w)))}.\qedhere
    \end{align*}
\end{proof}

\begin{lemma}\label{lem:negation-constraint-quantum-sound}
    Let $C\subseteq\Z_2^V$ and $C'\subseteq\Z_2^{V'}$ be incompressible TVF constraints whose TVF graphs have a $00$ edge between $u,v\in C$ and $11$ edge between $u',v'\in C'$, respectively. Consider the constraint systems $S=(\{x,y\},\{(\{x,y\},C_{\neq})\})$ and $S'=(X,\{(V_1,C),(V_2,r_\ast C')\})$, where $X=V\cup V'\backslash\{u',v'\}$, $V_1=V$, $V_2=V'\backslash\{u',v'\}\cup\{u,v\}$, and $r:V'\rightarrow V_2$ is a bijection such that $r_2(u')=u$ and $r_2(v')=v$. Then, there exists a $4\max\{|V|,|V'|\}$-homomorphism $\alpha:\mc{A}_{c-v}(S,\mbb{u}_1)\rightarrow\mc{A}_{c-v}(S',\mbb{u}_2)$.
\end{lemma}

\begin{proof}
    Define $\alpha$ as $\alpha(\sigma'(x))=\alpha(\sigma_1(x))=-\alpha(\sigma_1(y))=\sigma'(u)$ and $\alpha(\sigma'(y))=\sigma'(v)$. Then, $\alpha(\Phi_{\{x,y\},(0,1)})=\Pi_0(\sigma'(u))$ and $\alpha(\Phi_{\{x,y\},(1,0)})=\Pi_1(\sigma'(u))$, so
    \begin{align*}
        \alpha\Big(\sum_{z\in\{x,y\},\phi\in C_{\neq}}\hsq*{\Phi_{\{x,y\},\phi}(1-\Pi_{\phi(z)}(\sigma'(z)))}\Big)&=\hsq*{\Pi_0(\sigma'(u))\Pi_0(\sigma'(v))}+\hsq*{\Pi_1(\sigma'(u))\Pi_1(\sigma'(v))}.
    \end{align*}
    We have that $\Pi_0(\sigma_1(u))\Pi_0(\sigma_1(v))=\Pi_1(\sigma_2(u))\Pi_1(\sigma_2(v)) = 0$, so
    \begin{align*}
        \hsq*{\Pi_0(\sigma'(u))\Pi_0(\sigma'(v))}&=\hsq*{\Pi_0(\sigma_1(u))\Pi_0(\sigma_1(v))-\Pi_0(\sigma'(u))\Pi_0(\sigma'(v))}\\
        &\leq2\hsq*{\Pi_0(\sigma_1(u))-\Pi_0(\sigma'(u)))}+2\hsq*{\Pi_0(\sigma_1(v))-\Pi_0(\sigma'(v))}\\
        &\leq2\sum_{w\in\{u,v\},\phi\in C}\hsq*{\Phi_{V_1,\phi}(1-\Pi_{\phi(w)}(\sigma'(w)))}\\
        &\leq2\sum_{w\in V_1,\phi\in C}\hsq*{\Phi_{V_1,\phi}(1-\Pi_{\phi(w)}(\sigma'(w)))}.
    \end{align*}
    By a similar argument, $\hsq*{\Pi_1(\sigma'(u))\Pi_1(\sigma'(v))}\leq2\sum_{w\in V_2,\phi\in C}\hsq*{\Phi_{V_2,\phi}(1-\Pi_{\phi(w)}(\sigma'(w)))}$, giving the wanted result.
\end{proof}

Now, we show that one of the constraints necessary to construct one of the realisations of basic commutativity gadget can be simulated. 

\begin{theorem}\label{thm:simulating-simple-constraints}
    Let $\Gamma$ be a set of boolean TVF constraints such that $\CSP(\Gamma)_{1,1}$ is $\NP$-complete. Then, there exists a CS $S=(X,\{(V_i,{r_i}_{\ast}C_i)\}_{i=1}^m)\in\CSP(\Gamma)$ and a $\poly(L)$-homomorphism $$\alpha:\mc{A}_{c-v}(S_0,\mbb{u}_1)\rightarrow\mc{A}_{c-v}(S,\mbb{u}_{m}),$$ where $L=\max_{(V,C)\in\Gamma}|V|$ and $S_0=(\{x,y,z\},\{(\{x,y,z\},r_\ast C)\})$ for $C$ being one of $\{(1,0,0),(0,1,0),(0,0,1)\}$, $\{(1,0,1),(0,1,1),(0,0,0)\}$, $\{(1,1,1),(0,0,1),(0,1,0)\}$, $\{(0,1,1),(1,0,1),(1,1,0)\}$, and $r:[3]\rightarrow\{x,y,z\}$ a bijection.
\end{theorem}

\begin{proof}
    Since $\CSP(\Gamma)_{1,1}$ is $\NP$-complete, we know that $\CSP(\Gamma_{\max})_{1,1}$ is $\NP$-complete by \Cref{lem:max-compression-np-hard}. In particular, by Schaefer's dichotomy theorem, we know that there exist $C_1,C_2\in\CSP(\Gamma_{\max})$ such that $C_1$ does not satisfy the majority polymorphism, and $C_2$ does not satisfy the constant $0$ polymorphism ($0\notin C_2$ and $C_2\neq\varnothing$). Further, by \Cref{cor:np-hard-edges}, there exists an incompressible $C_3\in\Gamma_{\mathrm{\max}}$ such that the TVF graph of $C_3$ has a $00$ or a $11$ edge. Suppose we are in the second case. Suppose also that no constraint in $\Gamma_{\mathrm{\max}}$ has a TVF graph with a $00$ edge. Then, we know by \Cref{prop:TVF-simulates} that $C_1$ simulates either $\{(1,0,0),(0,1,0),(0,0,1)\}$ or $\{(1,0,1),(0,1,1),(0,0,0)\}$. Let this be $C$. First, using \Cref{lem:simulation-quantum-sound}, there is a constraint system $S_1\in\CSP(C_1,\{0\},\{1\})$ and a $\poly(L)$-homomorphism $\alpha_1:\mc{A}_{c-v}(S_0)\rightarrow\mc{A}_{c-v}(S_1)$, where uniform probability distributions on the constraints are implied. Now, using \Cref{lem:zero-quantum-sound} and \Cref{lem:one-quantum-sound}, we can express the constraints $\{0\}$ and $\{1\}$, respectively, with constraint systems in terms of $C_2$ and $C_3$. This induces a $\poly(L)$-homomorphism $\alpha_2:\mc{A}_{c-v}(S_1)\rightarrow\mc{A}_{c-v}(S_2)$ where $S_2\in\CSP(C_1,C_2,C_3)\subseteq\CSP(\Gamma_{\max})$. To finish, note that using \Cref{lem:the-bends}, we can construct a CS $S\in\CSP(\Gamma)$ and a $\poly(L)$-homomorphism $\alpha_3:\mc{A}_{c-v}(S_2)\rightarrow\mc{A}_{c-v}(S)$. Taking $\alpha=\alpha_3\circ\alpha_2\circ\alpha_1$ finishes the proof in this case.

    In the case that there are no $11$ edges in the constraints of $\Gamma_{\mathrm{comp}}$, we can do an identical argument with the labels $0$ and $1$ reversed. Then, we have that $C$ may be taken to be $\{(1,1,1),(0,0,1),(0,1,0)\}$ or $\{(0,1,1),(1,0,1),(1,1,0)\}$, the negation of the possible $C$s from the previous case.

    Now, suppose that there are incompressible constraints $C_3,C_4\in\Gamma_{\mathrm{comp}}$ such that the TVF graph of $C_3$ has a $11$ edge and the TVF graph of $C_4$ has a $00$ edge. Take $C=\{(1,0,0),(0,1,0),(0,0,1)\}$. Then, by \Cref{prop:TVF-simulates-negation}, there $C_1$ simulates $C$ with negation. In particular, there exists $U\subseteq V_{C_1}$ such that $(C_1)_{\lnot U}$ simulates $C$. As before, using \Cref{lem:simulation-quantum-sound}, \Cref{lem:zero-quantum-sound}, and \Cref{lem:one-quantum-sound}, there is a constraint system $S_1\in\CSP((C_1)_{\lnot U},C_2,C_3)$ and a $\poly(L)$-homomorphism $\alpha_1:\mc{A}_{c-v}(S_0)\rightarrow\mc{A}_{c-v}(S_1)$. Now, using \Cref{lem:negation-quantum-sound}, there exists a CS $S_2\in\CSP(C_1,C_{\neq},C_2,C_3)$ and a $\poly(L)$-homomorphism $\alpha_2:\mc{A}_{c-v}(S_1)\rightarrow\mc{A}_{c-v}(S_2)$. Next, using \Cref{lem:negation-constraint-quantum-sound}, there exists a CS $S_3\in\CSP(C_1,C_2,C_3,C_4)\subseteq\CSP(\Gamma_{\max})$ and a $\poly(L)$-homomorphism $\alpha_3:\mc{A}_{c-v}(S_2)\rightarrow\mc{A}_{c-v}(S_3)$. Finally, as before, we use \Cref{lem:the-bends}, to construct a CS $S\in\CSP(\Gamma)$ and a $\poly(L)$-homomorphism $\alpha_4:\mc{A}_{c-v}(S_3)\rightarrow\mc{A}_{c-v}(S)$, and take $\alpha=\alpha_4\circ\alpha_3\circ\alpha_2\circ\alpha_1$, finishing the proof.
\end{proof}

To finish this section, we construct the commutativity gadget.

\begin{corollary}\label{cor:general-commutativity-gadget}
    Let $\Gamma$ be a set of Boolean TVF constraints such that $\CSP(\Gamma)_{1,1}$ is $\NP$-complete. Then, there exists a BCS $S=(X,\{(V_i,{r_i}_\ast C_i)\}_{i=1}^m)\in\CSP(\Gamma)$ and variables $x,y\in X$ such that $S$ is satisfiable for any assignments to $x,y$ in $\Z_2$ and
    $$\hsq*{[\sigma'(x),\sigma'(y)]}\leq\poly(L)\frac{1}{m}\sum_{i=1}^m\frac{1}{|V_i|}\sum_{\phi\in C_i,z\in V_i}\hsq*{\Phi_{V_i,\phi}(1-\Pi_{\phi(z)}(\sigma'(z)))}$$
    in $\mc{A}_{c-v}(S)$, where $L=\max_{(V,C)\in\Gamma}|V|$.
\end{corollary}

\begin{proof}
    We begin with the the BCS from \Cref{lem:basic-commutativity-gadget}, $B=\{X,\{(V_i,\{r_i\}_\ast C)\}_{i=1}^3\}$, where $C$ is one of $\{(1,0,0),(0,1,0),(0,0,1)\}$, $\{(1,0,1),(0,1,1),(0,0,0)\}$, $\{(1,1,1),(0,0,1),(0,1,0)\}$, $\{(0,1,1),(1,0,1),(1,1,0)\}$. Now, we can apply the $\poly(L)$-homomorphism $\alpha$ from \Cref{thm:simulating-simple-constraints} to each of the constraints in $B$ to get $S\in\CSP(\Gamma)$. Letting $x_0,y_0\in X$ be the approximately-commuting variables in $B$, we take $\sigma'(x)=\alpha(\sigma'(x_0))$ and $\sigma'(y)=\alpha(\sigma'(y_0))$ to get the wanted result.
\end{proof}

\subsection{Proof of the main theorem for boolean TVF CSPs}

\begin{theorem}[Part 2 of \Cref{thm:main-theorem}]\label{thm:main-theorem-part-2}
    Let $\Gamma$ be a set of boolean TVF constraints such that $\CSP(\Gamma)_{1,1}$ is $\NP$-complete. Then, there exists $s\in[0,1)$ such that $\SuccinctCSP_{c-v}(\Gamma)^\ast_{1,s}$ is $\RE$-complete.
\end{theorem}

\begin{proof}
    Noting that the set of constraints $\Gamma\cup\{\Z_2^2\}$ is non-TVF and $\NP$-complete, we have by \Cref{thm:main-theorem-part-1} that $\SuccinctCSP_{c-v}(\Gamma\cup\{\Z_2^2\})^\ast_{1,s}$ is $\RE$-complete for some $s<1$. To complete the proof, note that we can use \Cref{cor:general-commutativity-gadget} to replace any instance of an empty constraint $\Z_2^2$ by a gadget composed of the constraints in $\Gamma$ while preserving completeness and constant soundness.
\end{proof}

\subsection{Oracularisability of boolean TVF CSPs}

\begin{lemma}\label{lem:acomm-to-a-tvf}
    Suppose $S=(X,\{(V_i,C_i)\}_{i=1}^m)$ is a TVF BCS. The identity map is a $(16L^2+1)$-homomorphism $\mc{A}_{a+comm}(S,\pi)\rightarrow\mc{A}_a(S,\pi)$, where $L=\max_i|V_i|$.
\end{lemma}

\begin{proof}
    By the TVF property, for every constraint $i$ and pair of variables $x,y\in V_i$, there exist $a,b\in\Z_2$ such that $\phi\notin C_i$ if $\phi(x)=a$ and $\phi(y)=b$. Then,
    \begin{align*}
        \norm{[\Pi_a(x),\Pi_b(y)]}_{\tau}^2&\leq4\norm{\Pi_a(x)\Pi_b(y)}_{\tau}^2=4\tau(\Pi_a(x)\Pi_b(y))\\
        &=4\sum_{\substack{\phi\in\Z_2^{V_i}\text{ s.t.}\\\phi(x)=a,\,\phi(y)=b}}\norm{\Phi_{V_i,\phi}}_\tau^2\leq4\sum_{\phi\notin C_i}\norm{\Phi_{V_i,\phi}}_\tau^2.
    \end{align*}
    Since $\Pi_{\lnot a}(x)=1-\Pi_a(x)$, and similarly for $y$, this upper bound holds for all $a,b\in\Z_2$.
    Then, we get that
    \begin{align*}
        \sum_{r\in\mc{A}_a(S)}(\mu_{a,\pi}(r)+\mu_{comm,\pi}(r))r^\ast r&=\sum_{i=1}^m\pi(i)\Big(\sum_{\phi\notin C_i}\norm{\Phi_{V_i,\phi}}_\tau^2+\sum_{\substack{x,y\in V_i\\a,b\in\Z_2}}\norm{[\Pi_a(x),\Pi_b(y)]}_\tau^2\Big)\\
        &\leq\sum_{i=1}^m\pi(i)(1+16|V_i|^2)\sum_{\phi\notin C_i}\norm{\Phi_{V_i,\phi}}_\tau^2\\
        &\leq(1+16L^2)\sum_{i=1}^m\pi(i)\sum_{\phi\notin C_i}\norm{\Phi_{V_i,\phi}}_\tau^2.\qedhere
    \end{align*}
\end{proof}

\begin{theorem}[Part 1 of \Cref{cor:main-theorem-2}]
    Let $\Gamma$ be a set of boolean TVF constraints such that $\CSP(\Gamma)_{1,1}$ is $\NP$-complete. Then, there exists $s\in[0,1)$ such that $\SuccinctCSP_{a}(\Gamma)^\ast_{1,s}$ is $\RE$-complete.
\end{theorem}

\begin{proof}
    Due to \Cref{lem:cv-to-acomm,lem:acomm-to-a-tvf}, there is a mapping $\mc{A}_{c-v}(S,\pi)\rightarrow\mc{A}_a(S,\pi)$ for all $S\in\CSP(\Gamma)$ that preserves the constant soundness; and due to \Cref{lem:a-to-acomm,lem:acomm-to-cv} there is a $C$-homomorphism in the other direction, preserving completeness. Hence, the result follows from \Cref{thm:main-theorem-part-2}.
\end{proof}

\section{Hardness of 2-CSPs}\label{sec:2-csps}

\subsection{The case of 3-colouring}

In this section, we show the $\RE$-completeness of $3$-colouring in the assignment and constraint-variable settings. Our arguments are exactly those of \cite{Ji13}, but adapted to the context of imperfect completeness, where we can phrase them in the language of weighted algebras.

First, we show the oracularisability of $3$-colouring, and use this to construct a mapping from the assignment algebra to the constraint-variable algebra for $3$-colouring instances.

\begin{lemma}\label{lem:3-colouring-is-oracularisable}
    Let $\mc{A}$ be a $\ast$-algebra. Let $x,y\in\mc{A}$ be order-$3$ unitaries. Then, for any $a,b\in\Z_3$,
    \begin{align*}
        \hsq*{[\Pi_a(x),\Pi_b(y)]}\lesssim16\sum_{c\in\Z_3}\hsq*{\Pi_c(x)\Pi_c(y)}.
    \end{align*}
\end{lemma}

This can be seen as a robust version of Lemma 2 from \cite{Ji13}.

\begin{proof}
    If $a=b$, then $\hsq*{[\Pi_a(x),\Pi_b(y)]}\lesssim4\hsq*{\Pi_a(x)\Pi_a(y)}$, giving the result. Else, without loss of generality, suppose $a=0$ and $b=1$, and write $x_i=\Pi_i(x)$ and $y_i=\Pi_i(y)$. We have that
    \begin{align*}
        [x_0,y_1]&=x_0y_1-y_1x_0=(y_0+y_1+y_2)x_0y_1-y_1x_0(y_0+y_1+y_2)\\
        &=y_0x_0y_1-y_1x_0y_0+y_2x_0y_1-y_1x_0y_2\\
        &=y_0x_0y_1-y_1x_0y_0+y_2(1-x_1-x_2)y_1-y_1(1-x_1-x_2)y_2\\
        &=y_0x_0y_1-y_1x_0y_0-y_2x_1y_1+y_1x_1y_2-y_2x_2y_1+y_1x_2y_2.
    \end{align*}
    Thus, by triangle inequality,
    \begin{align*}
        \hsq*{[x_0,y_1]}&\leq 8\parens*{\hsq*{y_0x_0y_1}+\hsq*{y_1x_0y_0}+\hsq*{y_2x_1y_1}+\hsq*{y_1x_1y_2}+\hsq*{y_2x_2y_1}+\hsq*{y_1x_2y_2}}\\
        &\lesssim 16\parens*{\hsq*{x_0y_0}+\hsq*{x_1y_1}+\hsq*{x_2y_2}}.\qedhere
    \end{align*}
\end{proof}

\begin{lemma}\label{lem:acomm-to-a-3col}
    Suppose $S=(X,\{(V_i=\{x_i,y_i\},{r_i}_\ast\neq_{\Z_3})\}_{i=1}^m)$ is a $3$-colouring instance. Then, the identity map on $\mc{A}_{a}(S)$ is a $145$-homomorphism $\mc{A}_{a+comm}(S,\pi)\rightarrow\mc{A}_{a}(S,\pi)$.
\end{lemma}

\begin{proof}
    Using \cref{lem:3-colouring-is-oracularisable}, we get that
    \begin{align*}
        \sum_{r\in\mc{A}_a(S)}\mu_{comm,\pi}(r)\hsq*{r}&=\sum_{i}\frac{\pi(i)}{|V_i|}\sum_{\substack{x,y\in V_i\\a,b\in\Z_3}}\hsq*{[\Pi_a(x),\Pi_b(y)]}=\sum_{i}\pi(i)\sum_{a,b\in\Z_3}\hsq*{[\Pi_a(x_i),\Pi_b(y_i)]}\\
        &\lesssim144\sum_{i}\pi(i)\sum_{c\in\Z_3}\hsq*{\Pi_c(x_i)\Pi_c(y_i)}\\
        &=144\sum_{r\in\mc{A}_a(S)}\mu_{a,\pi}\hsq*{r}.
    \end{align*}
    Therefore, we see that $\sum_{r\in\mc{A}_a(S)}(\mu_{a,\pi}(r)+\mu_{comm,\pi}(r))\hsq*{r}\leq145\sum_{r\in\mc{A}_a(S)}\mu_{a,\pi}\hsq*{r}$
\end{proof}

Now, we show the soundness of the prism graph construction of \cite{Ji13} in the assignment algebra, which we will use as a commutativity gadget. 

\begin{lemma}\label{lem:3-clique}
    Let $\mc{A}$ be a $\ast$-algebra with tracial state $\tau$, and let $x,y,z\in\mc{A}$ be order-$3$ unitaries. Then,
    \begin{align*}
        \sum_{a\in\Z_3}\norm{\Pi_a(x)+\Pi_a(y)+\Pi_a(z)-1}_\tau^2=2\sum_{a\in\Z_3}\norm{\Pi_a(x)\Pi_a(y)}_\tau^2+\norm{\Pi_a(y)\Pi_a(z)}_\tau^2+\norm{\Pi_a(z)\Pi_a(x)}_\tau^2.
    \end{align*}
\end{lemma}

This is a robust version of Lemma 3 from \cite{Ji13}. The right-hand side corresponds to the defect of a constraint system where the three variables $x,y,z$ are connected by $3$-coloring constraints in a triangular arrangement. See \Cref{fig:a-triangle} for a graphical representation.

\begin{proof}
    As in the proof of \Cref{lem:3-colouring-is-oracularisable}, write $x_i=\Pi_i(x)$, $y_i=\Pi_i(y)$, and $z_i=\Pi_i(z)$. Then, for $i=0,1,2$, we can expand
    \begin{align*}
        \sum_i\norm{x_i+y_i+z_i-1}_\tau^2&=\sum_i\tau\parens*{(x_i+y_i+z_i-1)^2}\\
        &=\sum_i\tau\parens*{x_i+x_iy_i+x_iz_i+y_ix_i+y_i+y_iz_i+z_ix_i+z_iy_i+z_i-2(x_i+y_i+z_i)+1}\\
        &=\sum_i\parens*{2\tau(x_iy_i+y_iz_i+z_ix_i)+\tau(1-(x_i+y_i+z_i))}\\
        &=2\sum_i\parens*{\tau(x_iy_i)+\tau(y_iz_i)+\tau(z_ix_i)}+3-\sum_ix_i-\sum_iy_i-\sum_iz_i\\
        &=2\sum_i\norm{x_iy_i}^2_\tau+\norm{y_iz_i}_\tau^2+\norm{z_ix_i}_\tau^2.\qedhere
    \end{align*}
\end{proof}

\begin{figure}
    \centering
    \begin{tikzpicture}
        \draw (0,0) -- (1.5,-2.60) -- (-1.5,-2.60) -- cycle;
        \fill (0,0) circle (3pt) node[above left] {$x$};
        \fill (1.5,-2.60) circle (3pt) node[below right] {$z$};
        \fill (-1.5,-2.60) circle (3pt) node[below left] {$y$};
    \end{tikzpicture}
\caption{The triangular constraint system in \Cref{lem:3-clique}. Each vertex corresponds to a variable and each edge corresponds to a $3$-colouring constraint.}
\label{fig:a-triangle}
\end{figure}

\begin{lemma}\label{lem:prism-soundness}
    Let $\mc{A}$ be a $\ast$-algebra with tracial state $\tau$, and let $x,y,z,x',y',z'\in\mc{A}$ be order-$3$ unitaries. Writing as previously $x_i=\Pi_i(x)$ and similarly for the other variables, we have
    \begin{align*}
        \sum_{i,j}\norm{[x_i,y_j']}_\tau^2&\leq6240\sum_i\big(\norm{x_iy_i}_\tau^2+\norm{y_iz_i}_\tau^2+\norm{z_ix_i}_\tau^2+\norm{x_i'y_i'}_\tau^2+\norm{y_i'z_i'}_\tau^2+\norm{x_i'y_i'}_\tau^2\\
        &\qquad\qquad\quad+\norm{x_ix_i'}_\tau^2+\norm{y_iy_i'}_\tau^2+\norm{z_iz_i'}_\tau^2\big).
    \end{align*}
\end{lemma}

This is a version of Lemma 4 from \cite{Ji13} with imperfect completeness. As in the previous lemma, the right-hand side corresponds to the defect of a $3$-colouring constraint system. Here, the variables are arranged as the vertices of a triangular prism, as illustrated in \Cref{fig:a-prism}.

\begin{proof}
    Consider first the case $i=j=0$. Using \cite[Lemma 4]{Ji13}, we see that
    \begin{align*}
        z_0z_0'x_0&=z_0(1-x_0'-y_0')x_0+z_0(x_0'+y_0'+z_0'-1)x_0\\
        &=z_0(1-y_0')x_0-z_0x_0'x_0+z_0(x_0'+y_0'+z_0'-1)x_0\\
        &=(1-x_0-y_0)(x_0-y_0'x_0)-z_0x_0'x_0+z_0(x_0'+y_0'+z_0'-1)x_0+(x_0+y_0+z_0-1)(1-y_0')x_0\\
        &=-y_0x_0-y_0'x_0+x_0y_0'x_0+y_0y_0'x_0-z_0x_0'x_0+z_0(x_0'+y_0'+z_0'-1)x_0+(x_0+y_0+z_0-1)(1-y_0)x_0,
    \end{align*}
    and taking the adjoint $x_0z_0'z_0=-x_0y_0-x_0y_0'+x_0y_0'x_0+x_0y_0'y_0-x_0x_0'z_0+x_0(x_0'+y_0'+z_0'-1)z_0+x_0(1-y_0)(x_0+y_0+z_0-1)$. Hence, the commutator
    \begin{align*}
        [x_0,y_0']&=x_0y_0'-y_0'x_0\\
        &=-x_0y_0-x_0z_0'z_0+x_0y_0'y_0-x_0x_0'z_0+x_0(x_0'+y_0'+z_0'-1)z_0+x_0(1-y_0)(x_0+y_0+z_0-1)\\
        &\qquad+y_0x_0+z_0z_0'x_0-y_0y_0'x_0+z_0x_0'x_0-z_0(x_0'+y_0'+z_0'-1)x_0-(x_0+y_0+z_0-1)(1-y_0)x_0,
    \end{align*}
    so the norm is bounded
    \begin{align*}
        \norm{[x_0,y_0']}_\tau^2\leq 32\parens*{\norm{x_0y_0}_\tau^2+\norm{z_0z_0'}_\tau^2+\norm{y_0y_0'}_\tau^2+\norm{x_0x_0'}_\tau^2+\norm{x_0'+y_0'+z_0'-1}_\tau^2+\norm{x_0+y_0+z_0-1}_\tau^2}.
    \end{align*}
    Now, by symmetry and using \Cref{lem:3-clique},
    \begin{align*}
        \sum_i\norm{[x_i,y_i']}_\tau^2&\leq32\sum_i\parens*{\norm{x_iy_i}_\tau^2+\norm{z_iz_i'}_\tau^2+\norm{y_iy_i'}_\tau^2+\norm{x_ix_i'}_\tau^2+\norm{x_i'+y_i'+z_i'-1}_\tau^2+\norm{x_i+y_i+z_i-1}_\tau^2}\\
        &=32\sum_i\big(\norm{z_iz_i'}_\tau^2+\norm{y_iy_i'}_\tau^2+\norm{x_ix_i'}_\tau^2+2\norm{x_i'y_i'}_\tau^2+2\norm{y_i'z_i'}_\tau^2+2\norm{z_i'x_i'}_\tau^2\\
        &\qquad\qquad\qquad+3\norm{x_iy_i}_\tau^2+2\norm{y_iz_i}_\tau^2+2\norm{z_ix_i}_\tau^2\big).
    \end{align*}
    Now, consider the case $i=0,j=1$. Again using \cite[Lemma 4]{Ji13},
    \begin{align*}
        y_2'x_2'x_0&=y_2'(1-x_0'-x_1')x_0=(1-y_0'-y_1')(x_0-x_1'x_0)+y_2'x_0'x_0\\
        &=x_0-y_0'x_0-y_1'x_0-x_1'x_0+y_0'x_1'x_0+y_1'x_1'x_0+y_2'x_0'x_0,
    \end{align*}
    and taking the adjoint $x_0x_2'y_2'=x_0-x_0y_0'-x_0y_1'-x_0x_1'+x_0x_1'y_0'+x_0x_1'y_1'+x_0x_0'y_2'$. Hence,
    \begin{align*}
        [x_0,y_1']&=x_0y_1'-y_1'x_0\\
        &=-x_0x_2'y_2'-x_0y_0'-x_0x_1'+x_0x_1'y_0'+x_0x_1'y_1'+x_0x_0'y_2'\\
        &\qquad+y_2'x_2'x_0+y_0'x_0+x_1'x_0-y_0'x_1'x_0-y_1'x_1'x_0-y_2'x_0'x_0\\
        &=-[x_0,y_0']-[x_0,x_1']-x_0x_2'y_2'+y_2'x_2'x_0+x_0x_1'y_1'-y_1'x_1'x_0+x_0x_0'y_2'-y_2'x_0'x_0\\
        &\qquad+x_0[x_1',y_0']+[x_0,y_0']x_1'-y_0'[x_1',x_0].
    \end{align*}
    Taking the norm, and then using \Cref{lem:3-colouring-is-oracularisable},
    \begin{align*}
        \norm{[x_0,y_1']}_\tau^2&\leq16\big(2\norm{[x_0,y_0']}_\tau^2+\norm{[x_0,x_1']}_\tau^2+2\norm{x_2'y_2'}_\tau^2+2\norm{x_1'y_1'}_\tau^2+2\norm{x_0x_0'}_\tau^2\\
        &\qquad+\norm{[x_1',y_0']}_\tau^2+\norm{[x_1',x_0]}_\tau^2\big)\\
        &\leq16\big(2\norm{[x_0,y_0']}_\tau^2+2\norm{x_2'y_2'}_\tau^2+2\norm{x_1'y_1'}_\tau^2+2\norm{x_0x_0'}_\tau^2+16\sum_i\norm{x_i'y_i'}_\tau^2+32\sum_i\norm{x_ix_i'}_\tau^2\big).
    \end{align*}
    By symmetry,
    \begin{align*}
        \sum_{i,j}\norm{[x_i,y_j']}_\tau^2&\leq\sum_i\big(65\norm{[x_i,y_i']}_\tau^2+1664\norm{x_i'y_i'}_\tau^2+3136\norm{x_ix_i'}_\tau^2\big)\\
        &\leq\sum_i\big(5824\norm{x_i'y_i'}_\tau^2+4160\norm{y_i'z_i'}_\tau^2+4160\norm{x_i'y_i'}_\tau^2\\
        &\qquad\qquad+6240\norm{x_iy_i}_\tau^2+4160\norm{y_iz_i}_\tau^2+4160\norm{z_ix_i}_\tau^2\\
        &\qquad\qquad+5216\norm{x_ix_i'}_\tau^2+2080\norm{y_iy_i'}_\tau^2+2080\norm{z_iz_i'}_\tau^2\big).\qedhere
    \end{align*}
\end{proof}

\begin{figure}
    \centering
    \begin{tikzpicture}[3d view={30}{10}]
        \draw (-1.5,0,-2.60) -- (-1.5,6,-2.60);
        \draw[white, line width=7pt] (0,0,0) -- (1.5,0,-2.60);
        \draw (0,0,0) -- (0,6,0);
        \draw (0,0,0) -- (1.5,0,-2.60) -- (-1.5,0,-2.60) -- cycle;
        \draw (0,6,0) -- (1.5,6,-2.60) -- (-1.5,6,-2.60) -- cycle;
        \draw (1.5,0,-2.60) -- (1.5,6,-2.60);
        \fill (0,0,0) circle (3pt) node[above left] {$x$};
        \fill (0,6,0) circle (3pt) node[above right] {$x'$};
        \fill (-1.5,0,-2.60) circle (3pt) node[below left] {$y$};
        \fill (-1.5,6,-2.60) circle (3pt) node[above left] {$y'$};
        \fill (1.5,0,-2.60) circle (3pt) node[below right] {$z$};
        \fill (1.5,6,-2.60) circle (3pt) node[below right] {$z'$};
    \end{tikzpicture}
\caption{The triangular prism constraint system in \Cref{lem:prism-soundness}. Each vertex corresponds to a variable and each edge correspond to a $3$-colouring constraint.}
\label{fig:a-prism}
\end{figure}

\begin{definition}
    The \textbf{triangular prism graph} is the graph illustrated in \Cref{fig:a-prism}, \emph{i.e.} $G_{prism}=(V_{prism},E_{prism})$ where $V_{prism}=\{x,y,z,x',y',z'\}$ and $$E_{prism}=\{\{x,y\},\{y,z\},\{z,x\},\{x',y'\},\{y',z'\},\{z',x'\},\{x,x'\},\{y,y'\},\{z,z'\}\}.$$
\end{definition}

\begin{corollary}\label{cor:commutativity-gadget-3col}
    Let $S=(X,\{(V_i,{r_i}_\ast\!\!\neq_{\Z_3})\}_{i=1}^{m_0}\cup\{(V_i,\Z_3^{V_i})\}_{i=m_0+1}^m)$ be a $3$-ary $2$-CS and let $\pi$ be a probability distribution on $[m]$. Write $V_i=\{x_i,y_i'\}$. Then, define $S'=(X',\{(V_i,{r_i}_\ast\!\!\neq_{\Z_3})\}_{i=1}^{m_0}\cup\{(V_{ie},{r_{ie}}_\ast\!\!\neq_{\Z_3})\}_{i\in[m]\backslash[m_0],e\in E_{prism}})$, where $X'=X\cup\set*{y_i,y_i',z_i,z_i'}{i=m_0+1,\ldots,m}$, $V_{ie}=\{\alpha_i,\beta_i\}$ where $\{\alpha,\beta\}=e$, and $r_{ie}:[2]\rightarrow V_{ie}$ is a bijection. Let $\pi'$ be the probability distribution $\pi'(i)=\pi(i)$ if $i\in[m_0]$ and $\pi'(ie)=\frac{\pi(i)}{9}$ otherwise. Then, for any trace $\tau$ on $\mc{A}_{a}(S',\pi')$, there exists a trace $\tau'$ on $\mc{A}_{c-v}(S,\pi)$ such that $\defect(\tau')\leq C\defect(\tau)$ for some universal constant $C>0$.
\end{corollary}

The constraint system $S'$ is constructed by replacing every empty constraint in $S$ by the triangular prism gadget from \Cref{lem:prism-soundness}.

We can also use \Cref{lem:a-to-acomm,lem:acomm-to-cv} and \Cref{lem:cv-to-acomm,lem:acomm-to-a-3col} to relate traces on $\mc{A}_a(S',\pi')$ and $\mc{A}_{c-v}(S',\pi')$.

\begin{proof}
    First, via \Cref{lem:3-colouring-is-oracularisable}, for each $i=1,\ldots,m_0$, we have that
    \begin{align*}
        \sum_{a,b\in\Z_3}\norm{[\Pi_a(x_i),\Pi_b(y_i')]}_\tau^2\leq 144\sum_{c\in\Z_3}\norm{\Pi_c(x_i)\Pi_c(y_i')}_\tau^2=144\sum_{\phi\notin {r_i}_\ast\!\neq_{\Z_3}}\norm{\Phi_{V_i,\phi}}_\tau^2
    \end{align*}
    Also, using \Cref{lem:prism-soundness}, we have that for each $i=m_0+1,\ldots,m$, 
    \begin{align*}
        \sum_{a,b\in\Z_3}\norm{[\Pi_a(x_i),\Pi_b(y_i')]}_\tau^2\leq&6240\sum_{c\in\Z_3}\big(\norm{\Pi_c(x_i)\Pi_c(y_i)}_\tau^2+\norm{\Pi_c(y_i)\Pi_c(z_i)}_\tau^2+\norm{\Pi_c(z_i)\Pi_c(x_i)}_\tau^2
        \\
        &+\norm{\Pi_c(x_i')\Pi_c(y_i')}_\tau^2+\norm{\Pi_c(y_i')\Pi_c(z_i')}_\tau^2+\norm{\Pi_c(x_i')\Pi_c(y_i')}_\tau^2
        \\
        &+\norm{\Pi_c(x_i)\Pi_c(x_i')}_\tau^2+\norm{\Pi_c(y_i)\Pi_c(y_i')}_\tau^2+\norm{\Pi_c(z_i)\Pi_c(z_i')}_\tau^2\big)
        \\
        =&6240\sum_{\substack{e\in E_{prism}\\\phi\notin{r_{ie}}_\ast\!\neq_{\Z_3}}}\norm{\Phi_{V_{ie},\phi}}_\tau^2.
        \end{align*} Therefore, as $\mc{A}_a(S)$ is a subalgebra of $\mc{A}_a(S')$, the defect
    \begin{align*}
        \defect(\tau|_{\mc{A}_a(S)};\mu_{a,\pi}+\mu_{a,comm})&=\sum_{i=1}^{m_0}\pi(i)\sum_{\varphi\notin {r_i}_\ast\!\neq_{\Z_3}}\norm{\Phi_{V_i,\phi}}_\tau^2+\sum_{i=1}^{m}\pi(i)\sum_{x,y\in V_i\;a,b\in\Z_3}\norm{[\Pi_a(x),\Pi_b(y)]}_\tau^2\\
        &\leq145\sum_{i=1}^{m_0}\pi(i)\sum_{\varphi\notin {r_i}_\ast\!\neq_{\Z_3}}\norm{\Phi_{V_i,\phi}}_\tau^2+6240\sum_{i=m_0+1}^m\pi(i)\sum_{\substack{e\in E_{prism}\\\phi\notin{r_{ie}}_\ast\!\neq_{\Z_3}}}\norm{\Phi_{V_{ie},\phi}}_\tau^2\\
        &\leq56160\defect(\tau).
    \end{align*}
    Now, due to \Cref{lem:cv-to-acomm}, there exists a trace $\tau'$ on $\mc{A}_{c-v}(S,\pi)$ such that $\defect(\tau)\leq 56160\poly(k^L)\defect(\tau)$, which is a constant as $k=3$, $L=2$.
\end{proof}

\begin{theorem}[Part 3 of \Cref{thm:main-theorem} and Part 2 of \Cref{cor:main-theorem-2}]\label{thm:main-theorem-part-3}
    There exists a constant $s\in[0,1)$ such that $\SuccinctCSP_{a}(\{\neq_{\Z_3}\})^\ast_{1,s}$ and $\SuccinctCSP_{c-v}(\{\neq_{\Z_3}\})^\ast_{1,s}$ are $\RE$-complete.
\end{theorem}

\begin{proof}
    We proceed similarly to \Cref{thm:main-theorem-part-2}. By \Cref{thm:main-theorem-part-1}, we know that $\SuccinctCSP_{c-v}(\{\neq_{\Z_3},\Z_3^2\})^\ast_{1,s_0}$ is $\RE$-complete. Next, using \Cref{cor:commutativity-gadget-3col}, we can replace the commutation constraints by gadgets over the assignment algebra, and find that $\SuccinctCSP_a(\{\neq_{\Z_3}\})^\ast_{1,s_1}$ is $\RE$-complete. Finally, using \Cref{lem:cv-to-acomm,lem:acomm-to-a-3col} and \Cref{lem:a-to-acomm,lem:acomm-to-cv}, we find that $\SuccinctCSP_{c-v}(\{\neq_{\Z_3}\})^\ast_{1,s_2}$ is also $\RE$-complete.
\end{proof}

\subsection{The case of 2-CSP(\textit{k})}

\begin{definition}
    Define the CSP $2\text{-}\CSP(k)=\CSP(\Gamma)$, where $\Gamma=\set*{C\subseteq\Z_k^2}$. This corresponds to the set of all $2$-CSPs over an alphabet of size $k$. Define the corresponding promise problems as $2\text{-}\CSP(k)_{c,s}=\CSP(\Gamma)_{c,s}$, $\mathrm{Succinct}\text{-}2\text{-}\CSP(k)_{c,s}=\SuccinctCSP(\Gamma)_{c,s}$, $2\text{-}\CSP_w(k)^\ast_{c,s}=\CSP_w(\Gamma)^\ast_{c,s}$, $\mathrm{Succinct}\text{-}2\text{-}\CSP_w(k)^\ast_{c,s}=\SuccinctCSP_w(\Gamma)^\ast_{c,s}$, where $w\in\{c-c,c-v,a,a+comm\}$.
\end{definition}

Next, we relate the assignment algebra for the language of all $2$-CSPs to a language where the c-v algebra is hard.

\begin{proposition}\label{prop:cv-to-2csp}
    Let $k\geq 3$ and set $C=\set*{(x_1,...,x_k)\in\Z_2^k}{\exists! i.\;x_i=1}$. Note that $|C|=k$, so for $a\in\Z_k$, let $c_a\in C$ be the element with $1$ in the $a$-th position. Consider some BCS $S=(X,\{(V_i,{r_i}_\ast C)\}_{i=1}^m)\in\CSP(\{C\})$ and a probability distribution $\pi$ on $[m]$. Define a $k$-ary $2$-CS $S'=(Y,\{(W_{ix},D_{x})\}_{i\in[m],x\in V_i})$ by $Y=X\cup\set*{y_i}{i\in[m]}$, $W_{ix}=\{y_i,x\}$, and $D_x=\set*{(a,0)}{f(x)=0,f\circ r_i=c_a}\cup\set*{(a,b)}{f(x)=1,f\circ r_i=c_a,b\neq 0}\subseteq\Z_k^2$. Let $\pi'(i,x)=\frac{\pi(i)}{|V_i|}$. Then, there is a trace $\tau$ on $\mc{A}_{c-v}(S,\pi)$ with $\defect(\tau)=\varepsilon$ if and only if there is a trace $\tau'$ on $\mc{A}_a(S',\pi')$ with $\defect(\tau')=\varepsilon$.
\end{proposition}

\begin{proof}
    Let $\tau$ be a trace on $\mc{A}_{c-v}(S)$, and $\tau=\rho\circ\varphi$ be its GNS representation. Let $\chi$ be the representation of $\mc{A}_a(S')$ defined by $$\chi(\Pi_a(y_i))=\begin{cases}\varphi(\Phi_{V_i,f})&\exists f\in{r_i}_\ast C.\;f\circ r_i=c_a\\0&\text{ otherwise,}\end{cases}$$ $\chi(\Pi_0(x))=\varphi(\Pi_0(\sigma'(x)))$, $\chi(\Pi_{1}(x))=\varphi(\Pi_{1}(\sigma'(x)))$, and $\chi(\Pi_a(x))=0$ for $a\neq 1,\omega$. Take $\tau'=\rho\circ\chi$. Then, the defect
    \begin{align*}
        \defect(\tau')&=\sum_{i\in[m],x\in V_i}\pi'(i,x)\sum_{(a,b)\notin D_x}\norm{\Pi_a(y_i)\Pi_b(x)}_{\tau'}^2=\sum_{i\in[m],x\in V_i}\frac{\pi(i)}{|V_i|}\sum_{(a,b)\notin D_x}\tau'(\Pi_a(y_i)\Pi_b(x))\\
        &=\sum_{i\in[m],x\in V_i}\frac{\pi(i)}{|V_i|}\parens[\Bigg]{\sum_{\substack{a.\,(a,0)\notin D_x\\f\in {r_i}_\ast C.\,f\circ r_i=c_a}}\tau(\Phi_{V_i,f}\Pi_0(\sigma'(x)))+\sum_{\substack{a.\,(a,1)\notin D_x,\\f\in {r_i}_\ast C.\,f\circ r_i=c_a}}\tau(\Phi_{V_i,f}\Pi_{1}(\sigma'(x)))}\\
        &=\sum_{i\in[m],x\in V_i}\frac{\pi(i)}{|V_i|}\sum_{\phi\in {r_i}_{\ast}C}\tau(\Phi_{V_i,\phi}(1-\Pi_{\phi(x)}(\sigma'(x))))=\defect(\tau).
    \end{align*}

    Conversely, suppose that $\tau'$ is a trace on $\mc{A}_{a}(S')$ with GNS representation $\tau'=\rho'\circ\varphi'$. Then, define the representation $\chi'$ of $\mc{A}_{c-v}(S)$ by $\chi'(\Phi_{V_i,f})=\sum_{a.f\circ r_i=c_a}\varphi'(\Pi_a(y_i))$, $\chi'(\Pi_0(\sigma'(x)))=\varphi'(\Pi_0(x))$, and $\chi'(\Pi_{1}(\sigma'(x)))=\sum_{a\neq 1}\varphi'(\Pi_a(x))$. Let $\tau=\rho'\circ\chi'$. By a similar calculation as above, we have that the defect
    \begin{align*}
        \defect(\tau)&=\sum_{i\in[m],x\in V_i}\frac{\pi(i)}{|V_i|}\sum_{\phi\in {r_i}_\ast C}\tau(\Phi_{V_i,\phi}(1-\Pi_{\phi(x)}(\sigma'(x))))\\
        &=\sum_{i\in[m]}\pi'(i,x)\parens[\Bigg]{\sum_{\substack{x\in V_i,f\in {r_i}_\ast C\\f(x)=0}}\tau(\Phi_{V_i,f}\Pi_{1}(\sigma'(x)))+\sum_{\substack{x\in V_i,f\in{r_i}_\ast C\\f(x)=1}}\tau(\Phi_{V_i,f}\Pi_{0}(\sigma'(x)))}\\
        &=\sum_{i\in[m],a\in\Z_k}\pi'(i,x)\parens[\Bigg]{\sum_{\substack{x\in V_i,\\f\circ r_i=c_a,f(x)=0}}\sum_{b\neq 0}\tau'(\Pi_a(y_i)\Pi_{b}(x))+\hspace{-5mm}\sum_{\substack{x\in V_i,\\f\circ r_i=c_a,f(x)=1}}\hspace{-5mm}\tau'(\Pi_a(y_i)\Pi_{0}(x))}\\
        &=\sum_{i\in[m],x\in V_i}\pi'(i,x)\sum_{(a,b)\notin D_x}\tau'(\Pi_a(y_i)\Pi_b(x))=\defect(\tau').\qedhere
    \end{align*}
\end{proof}

\begin{theorem}[Part 3 of \Cref{cor:main-theorem-2}]
    There exists $s\in[0,1)$ such that $\mathrm{Succinct}\text{-}2\text{-}\mathrm{CSP}_{a}(k)_{1,s}^\ast$ is $\RE$-complete.
\end{theorem}

\begin{proof}
    Let $C$ be as in \Cref{prop:cv-to-2csp}. By \Cref{thm:main-theorem-part-2}, we know that $\SuccinctCSP_{c-v}(\{C\})^\ast_{1,s}$ is $\RE$-complete. But, by \Cref{prop:cv-to-2csp}, there is a value-preserving mapping from instances of $\SuccinctCSP_{c-v}(\{C\})^\ast_{1,s}$ to instances of $\mathrm{Succinct}\text{-}2\text{-}\mathrm{CSP}_{a}(k)_{1,s}$, completing the proof.
\end{proof}

\section{Constraint-variable to constraint-constraint for CSPs}\label{sec:cv-to-cc}

\begin{proposition}\label{prop:cc-to-cv-complexity}
    Let $\Gamma$ be a set of $k$-ary constraints with $(V_0,C_0)\in\Gamma$ and $v\in V_0$ such that for all $a\in\Z_k$ there exists $\phi\in C_0$ such that $\phi(v)=a$. Let $S=(X,\{(V_i,{r_i}_\ast C_i)\}_{i=1}^m)\in\CSP(\Gamma)$, and let $\pi$ be a probability distribution on $[m]$. Then, there exists a CS $S'=(X',\{(V'_i,{r'_i}_\ast C_i')_{i=1}^{m'})\in\CSP(\Gamma)$ and a probability distribution on $[m']$, such that there is a $2$-homomorphism $\alpha:\mc{A}_{c-v}(S,\pi)\rightarrow\mc{A}_{c-v}(S',\pi')$ and a $\frac{1}{2}$-homomorphism $\beta:\mc{A}_{c-v}(S',\pi')\rightarrow\mc{A}_{c-v}(S,\pi)$. Also, there exists a probability distribution $\pi''$ on $[m]\times[m]$ and a $\frac{1}{2}$-homomorphism $\gamma:\mc{A}_{c-v}(S',\pi')\rightarrow\mc{A}_{c-c}(S',\pi'')$.
\end{proposition}

As will be shown below in the proof of \Cref{cor:main-theorem-1}, note that all the sets of constraints considered in \Cref{thm:main-theorem} satisfy the necessary condition of this proposition.

\begin{proof}
    Define $S'=(X',\{(V_i,{r_i}_\ast C_i)\}_{i=1}^m\cup\{(V_{x},{r_{x}}_{\ast} C_0)\}_{x\in X})$, where $X'=X\cup\bigcup_{x\in X}V_{x}$, $V_{x}=x\cup\set*{u_{x}}{u\in V_0\backslash\{v\}}$, and
    $$r_{x}(u)=\begin{cases}x&u=v\\u_{x}&\text{otherwise.}\end{cases}$$
    Define $\pi'(i)=\frac{\pi(i)}{2}$ and $\pi'(x)=\sum_{i\in[m].\;x\in V_i}\frac{\pi(i)}{2|V_i|}$. Taking $\alpha$ to be the natural embedding $\mc{A}_{c-v}(S,\pi)\hookrightarrow\mc{A}_{c-v}(S',\pi')$ gives a $2$-homomorphism. Next, we can take $\beta$ to be the identity on $\mc{A}_{c-v}(S,\pi)\subseteq\mc{A}_{c-v}(S',\pi')$ and then take $\beta(\sigma_x(x))=\sigma'(x)$, and $\beta(\sigma'(u_x))=\beta(\sigma_x(u_x))\in\gen*{\sigma'(x)}$ such that they give a satisfying assignment to ${r_x}_\ast C_0$. Therefore, the terms corresponding to $(V_x,{r_x}_\ast C_0)$ are sent to $0$ and we have that
    \begin{align*}
        \beta\Big(\sum_{i=1}^{m'}\frac{\pi'(i)}{|V_i|}\sum_{x\in V_i,\phi\in {r_i}_\ast C_i}\hsq*{\Phi_{V_i,\phi}(1-\Pi_{\phi(x)}(\sigma'(x)))}\Big)=\frac{1}{2}\sum_{i=1}^{m}\frac{\pi(i)}{|V_i|}\sum_{x\in V_i,\phi\in {r_i}_\ast C_i}\hsq*{\Phi_{V_i,\phi}(1-\Pi_{\phi(x)}(\sigma'(x)))},
    \end{align*}
    giving the wanted $\frac{1}{2}$-homomorphism.

    Now, take $\pi''(i,i)=\pi''(x,x)=\pi''(x,i)=0$ and $$\pi''(i,x)=\begin{cases}\frac{\pi(i)}{|V_i|}&x\in V_i\\0&\text{else}\end{cases}.$$
    Let $\gamma(\Phi_{V_i,\phi})=\Phi_{V_i,\phi}$, $\gamma(\sigma'(x))=\gamma(\sigma_x(x))=\sigma_x(x)$, and $\gamma(\sigma'(U_x))=\gamma(\sigma_x(u_x))=\sigma_x(u_x)$. Then, we have that
    \begin{align*}
        \gamma\Big(\sum_{i=1}^{m'}\frac{\pi'(i)}{|V_i|}&\sum_{x\in V_i,\phi\in {r_i}_\ast C_i}\hsq*{\Phi_{V_i,\phi}(1-\Pi_{\phi(x)}(\sigma'(x)))}\Big)\\
        &=\sum_{i=1}^{m}\frac{\pi(i)}{2|V_i|}\sum_{x\in V_i,\phi\in {r_i}_\ast C_i}\hsq*{\Phi_{V_i,\phi}(1-\Pi_{\phi(x)}(\sigma_x(x)))}\\
        &\lesssim\frac{1}{2}\sum_{i\in[m],x\in V_i}\pi''(i,x)\sum_{\phi\in {r_i}_\ast C_i,a\neq\phi(x)}\Phi_{V_i,\phi}\Pi_{a}(\sigma_x(x))\\
        &=\frac{1}{2}\sum_{i\in[m],x\in V_i}\pi''(i,x)\sum_{\substack{\phi\in {r_i}_\ast C_i,\psi\in {r_x}_\ast C_0\\\phi|_{V_i\cap V_x}\neq\psi|_{V_i\cap V_x}}}\Phi_{V_i,\phi}\Phi_{V_x,\psi}\\
        &\lesssim\frac{1}{2}\sum_{i,j=1}^{m'}\pi''(i,j)\sum_{\substack{\phi\in {r_i}_\ast C_i,\psi\in {r_j}_\ast C_j\\\phi|_{V_i\cap V_j}\neq\psi|_{V_i\cap V_j}}}\hsq*{\Phi_{V_i,\phi}\Phi_{V_j,\psi}},
    \end{align*}
    giving the wanted $\tfrac{1}{2}$-homomorphism.
\end{proof}

\begin{proof}[Proof of \Cref{cor:main-theorem-1}]
    Let $\Gamma$ be a set constraints satisfying the conditions of \Cref{thm:main-theorem}. In the case that $\Gamma$ is boolean, there must be a variable that takes both assignment $0$ and $1$, as there must be at least one constraint that has two distinct satisfying assignments. Next, in the case that $\Gamma=\{\neq_{\Z_3}\}$, either of the two variables can take all three values in $\Z_3$. Finally, in the case that $\Gamma$ is non-TVF, there is a pair of variables that can take any pair of values, so in particular either one of them can take any value. In all three cases, $\Gamma$ satisfies the condition of \Cref{prop:cc-to-cv-complexity}. Further, we know by \Cref{thm:main-theorem} that $\SuccinctCSP_{c-v}(\Gamma)_{1,s}^\ast$ is $\RE$-complete. Then, by the first part of \Cref{prop:cc-to-cv-complexity}, the instances of this problem can be mapped to a subset of the instances in such a way that the constant gap is preserved. Then, using the second part of \Cref{prop:cc-to-cv-complexity} and \Cref{lem:cv-to-cc}, we find a $C$-homomorphism between these instances and the corresponding constraint-constraint algebras, thus $\SuccinctCSP_{c-c}(\Gamma)_{1,s'}^\ast$ is also $\RE$-complete.
\end{proof}

\newpage

\bibliographystyle{acm}
\bibliography{citations}

\begin{thebibliography}{10}

\bibitem{Arkh12}
{\sc Arkhipov, A.}
\newblock Extending and characterizing quantum magic games.
\newblock {\em arXiv preprint arXiv:1209.3819\/} (2012).

\bibitem{BFL91}
{\sc Babai, L., Fortnow, L., and Lund, C.}
\newblock Non-deterministic exponential time has two-prover interactive
  protocols.
\newblock {\em Computational complexity 1\/} (1991), 3--40.

\bibitem{Bell64}
{\sc Bell, J.~S.}
\newblock On the {E}instein {P}odolsky {R}osen paradox.
\newblock {\em Physics Physique Fizika 1\/} (Nov 1964), 195--200.

\bibitem{blackadar06}
{\sc Blackadar, B.}
\newblock {\em Operator algebras: theory of {C}*-algebras and von Neumann
  algebras}, vol.~122.
\newblock Springer Berlin, Heidelberg, 2006.

\bibitem{Bul17}
{\sc Bulatov, A.~A.}
\newblock A dichotomy theorem for nonuniform {CSP}s.
\newblock In {\em 2017 IEEE 58th Annual Symposium on Foundations of Computer
  Science (FOCS)\/} (2017), IEEE, pp.~319--330.

\bibitem{chapman2023efficiently}
{\sc Chapman, M., Vidick, T., and Yuen, H.}
\newblock Efficiently stable presentations from error-correcting codes.
\newblock {\em arXiv preprint arXiv:2311.04681\/} (2023).

\bibitem{cleve2010consequences}
{\sc Cleve, R., Hoyer, P., Toner, B., and Watrous, J.}
\newblock Consequences and limits of nonlocal strategies.
\newblock In {\em Proceedings. 19th IEEE Annual Conference on Computational
  Complexity, 2004.\/} (2004), pp.~236--249.

\bibitem{CM14}
{\sc Cleve, R., and Mittal, R.}
\newblock Characterization of binary constraint system games.
\newblock In {\em Automata, Languages, and Programming: 41st International
  Colloquium, ICALP 2014, Copenhagen, Denmark, July 8-11, 2014, Proceedings,
  Part I 41\/} (2014), Springer, pp.~320--331.

\bibitem{culf2024}
{\sc Culf, E., Mousavi, H., and Spirig, T.}
\newblock Approximation algorithms for noncommutative {CSP}s.
\newblock {\em arXiv preprint arXiv:2312.16765\/} (2024).

\bibitem{DFNQXY23}
{\sc Dong, Y., Fu, H., Natarajan, A., Qin, M., Xu, H., and Yao, P.}
\newblock The computational advantage of {MIP}* vanishes in the presence of
  noise.
\newblock {\em arXiv preprint arXiv:2312.04360\/} (2023).

\bibitem{Dwork1992LowC2}
{\sc Dwork, C., Feige, U., Kilian, J., Naor, M., and Safra, S.}
\newblock Low communication 2-prover zero-knowledge proofs for {NP}.
\newblock In {\em Annual International Cryptology Conference\/} (1992).

\bibitem{EPR35}
{\sc Einstein, A., Podolsky, B., and Rosen, N.}
\newblock Can quantum-mechanical description of physical reality be considered
  complete?
\newblock {\em Phys. Rev. 47\/} (May 1935), 777--780.

\bibitem{FJVY19}
{\sc Fitzsimons, J., Ji, Z., Vidick, T., and Yuen, H.}
\newblock Quantum proof systems for iterated exponential time, and beyond.
\newblock In {\em Proceedings of the 51st Annual ACM SIGACT Symposium on Theory
  of Computing\/} (New York, NY, USA, 2019), STOC 2019, Association for
  Computing Machinery, p.~473–480.

\bibitem{GW02}
{\sc Galliard, V., and Wolf, S.}
\newblock Pseudo-telepathy, entanglement, and graph colorings.
\newblock In {\em Proceedings IEEE International Symposium on Information
  Theory,\/} (2002), IEEE, p.~101.

\bibitem{Har24b}
{\sc Harris, S.~J.}
\newblock Approximate quantum 3-colorings of graphs and the quantum max 3-cut
  problem.
\newblock {\em arXiv preprint arXiv:2412.19405\/} (2024).

\bibitem{Har24}
{\sc Harris, S.~J.}
\newblock Universality of graph homomorphism games and the quantum coloring
  problem.
\newblock In {\em Annales Henri Poincar{\'e}\/} (2024), Springer, pp.~1--36.

\bibitem{Has01}
{\sc H{\aa}stad, J.}
\newblock Some optimal inapproximability results.
\newblock {\em Journal of the ACM (JACM) 48}, 4 (2001), 798--859.

\bibitem{helton2017algebras}
{\sc Helton, J.~W., Meyer, K.~P., Paulsen, V.~I., and Satriano, M.}
\newblock Algebras, synchronous games, and chromatic numbers of graphs.
\newblock {\em New York J. Math 25\/} (2019), 328--361.

\bibitem{IKM09}
{\sc Ito, T., Kobayashi, H., and Matsumoto, K.}
\newblock Oracularization and two-prover one-round interactive proofs against
  nonlocal strategies.
\newblock In {\em 2009 24th Annual IEEE Conference on Computational
  Complexity\/} (2009), pp.~217--228.

\bibitem{IV12}
{\sc Ito, T., and Vidick, T.}
\newblock A multi-prover interactive proof for {NEXP} sound against entangled
  provers.
\newblock In {\em 2012 IEEE 53rd Annual Symposium on Foundations of Computer
  Science\/} (2012), pp.~243--252.

\bibitem{Ji13}
{\sc Ji, Z.}
\newblock Binary constraint system games and locally commutative reductions.
\newblock {\em arXiv preprint arXiv:1310.3794\/} (2013).

\bibitem{Ji16}
{\sc Ji, Z.}
\newblock Classical verification of quantum proofs.
\newblock In {\em Proceedings of the Forty-Eighth Annual ACM Symposium on
  Theory of Computing\/} (New York, NY, USA, 2016), STOC '16, Association for
  Computing Machinery, p.~885–898.

\bibitem{Ji17}
{\sc Ji, Z.}
\newblock Compression of quantum multi-prover interactive proofs.
\newblock In {\em Proceedings of the 49th Annual ACM SIGACT Symposium on Theory
  of Computing\/} (New York, NY, USA, 2017), STOC 2017, Association for
  Computing Machinery, p.~289–302.

\bibitem{ji2022mipre}
{\sc Ji, Z., Natarajan, A., Vidick, T., Wright, J., and Yuen, H.}
\newblock {MIP}*={RE}.
\newblock {\em arXiv preprint arXiv:2001.04383\/} (2022).

\bibitem{KKM+11}
{\sc Kempe, J., Kobayashi, H., Matsumoto, K., Toner, B., and Vidick, T.}
\newblock Entangled games are hard to approximate.
\newblock {\em SIAM Journal on Computing 40}, 3 (2011), 848--877.

\bibitem{KRT10}
{\sc Kempe, J., Regev, O., and Toner, B.}
\newblock Unique games with entangled provers are easy.
\newblock {\em SIAM Journal on Computing 39}, 7 (2010), 3207--3229.

\bibitem{Kho02}
{\sc Khot, S.}
\newblock On the power of unique 2-prover 1-round games.
\newblock In {\em Proceedings 17th IEEE Annual Conference on Computational
  Complexity\/} (2002), p.~25.

\bibitem{Kim_2018}
{\sc Kim, S.-J., Paulsen, V., and Schafhauser, C.}
\newblock A synchronous game for binary constraint systems.
\newblock {\em Journal of Mathematical Physics 59}, 3 (Mar 2018).

\bibitem{lin2024tracialembeddable}
{\sc Lin, J.}
\newblock Tracial embeddable strategies: Lifting {MIP}* tricks to {MIP}co.
\newblock {\em arXiv preprint arXiv:2304.01940\/} (2024).

\bibitem{marrakchi2023synchronous}
{\sc Marrakchi, A., and de~la Salle, M.}
\newblock Almost synchronous correlations and tomita-takesaki theory.
\newblock {\em arXiv preprint arXiv:2307.08129\/} (2023).

\bibitem{MS24}
{\sc Mastel, K., and Slofstra, W.}
\newblock Two prover perfect zero knowledge for {MIP}*.
\newblock In {\em Proceedings of the 56th Annual ACM Symposium on Theory of
  Computing\/} (New York, NY, USA, 2024), STOC 2024, Association for Computing
  Machinery, p.~991–1002.

\bibitem{mermin90simple}
{\sc Mermin, N.~D.}
\newblock Simple unified form for the major no-hidden-variables theorems.
\newblock {\em Phys. Rev. Lett. 65\/} (Dec 1990), 3373--3376.

\bibitem{MS24b}
{\sc Mousavi, H., and Spirig, T.}
\newblock A quantum unique games conjecture.
\newblock {\em arXiv preprint arXiv:2409.20028\/} (2024).

\bibitem{NN24}
{\sc Natarajan, A., and Nirkhe, C.}
\newblock The status of the quantum pcp conjecture (games version).
\newblock {\em arXiv preprint arXiv:2403.13084\/} (2024).

\bibitem{NV18a}
{\sc Natarajan, A., and Vidick, T.}
\newblock Low-degree testing for quantum states, and a quantum entangled games
  {PCP} for {QMA}.
\newblock In {\em 2018 {IEEE} 59th Annual Symposium on Foundations of Computer
  Science ({FOCS})\/} (Oct 2018), {IEEE}.

\bibitem{NV18b}
{\sc Natarajan, A., and Vidick, T.}
\newblock Two-player entangled games are {NP}-hard.
\newblock In {\em Proceedings of the 33rd Computational Complexity
  Conference\/} (Dagstuhl, DEU, 2018), CCC '18, Schloss
  Dagstuhl--Leibniz-Zentrum fuer Informatik.

\bibitem{natarajan2019neexp}
{\sc Natarajan, A., and Wright, J.}
\newblock {NEEXP} is contained in {MIP}*.
\newblock In {\em 2019 IEEE 60th Annual Symposium on Foundations of Computer
  Science (FOCS)\/} (Los Alamitos, CA, USA, nov 2019), IEEE Computer Society,
  pp.~510--518.

\bibitem{ozawa2013connes}
{\sc Ozawa, N.}
\newblock About the {C}onnes embedding conjecture: algebraic approaches.
\newblock {\em Japanese Journal of Mathematics 8}, 1 (2013), 147--183.

\bibitem{Pad22}
{\sc Paddock, C.}
\newblock Rounding near-optimal quantum strategies for nonlocal games to
  strategies using maximally entangled states.
\newblock {\em arXiv preprint arXiv:2203.02525\/} (2022).

\bibitem{PS23}
{\sc Paddock, C., and Slofstra, W.}
\newblock Satisfiability problems and algebras of boolean constraint system
  games.
\newblock {\em arXiv preprint arXiv:2310.07901\/} (2023).

\bibitem{PERES1990107}
{\sc Peres, A.}
\newblock Incompatible results of quantum measurements.
\newblock {\em Physics Letters A 151}, 3 (1990), 107--108.

\bibitem{Rag08}
{\sc Raghavendra, P.}
\newblock Optimal algorithms and inapproximability results for every {CSP}?
\newblock In {\em Proceedings of the Fortieth Annual ACM Symposium on Theory of
  Computing\/} (New York, NY, USA, 2008), STOC '08, Association for Computing
  Machinery, p.~245–254.

\bibitem{Sch78}
{\sc Schaefer, T.~J.}
\newblock The complexity of satisfiability problems.
\newblock In {\em Proceedings of the Tenth Annual ACM Symposium on Theory of
  Computing\/} (New York, NY, USA, 1978), STOC '78, Association for Computing
  Machinery, p.~216–226.

\bibitem{Schmdgen2020}
{\sc Schm\"{u}dgen, K.}
\newblock {\em An Invitation to Unbounded Representations of $\ast$-Algebras on
  Hilbert Space}.
\newblock Springer International Publishing, 2020.

\bibitem{Vid14}
{\sc Vidick, T.}
\newblock {CS286 Seminar in Computer Science: Around the quantum PCP
  conjecture}, 2014.
\newblock URL: \url{users.cms.caltech.edu/~vidick/teaching/286_qPCP/}. Last
  visited on 2024/10/04.

\bibitem{Vid16}
{\sc Vidick, T.}
\newblock Three-player entangled {XOR} games are {NP}-hard to approximate.
\newblock {\em SIAM Journal on Computing 45}, 3 (2016), 1007--1063.

\bibitem{Vid20eratum}
{\sc Vidick, T.}
\newblock Erratum: Three-player entangled {XOR} games are {NP}-hard to
  approximate.
\newblock {\em SIAM Journal on Computing 49}, 6 (2020), 1423--1427.

\bibitem{Vidick_2022}
{\sc Vidick, T.}
\newblock Almost synchronous quantum correlations.
\newblock {\em Journal of Mathematical Physics 63}, 2 (Feb 2022).

\bibitem{Zhu17}
{\sc Zhuk, D.}
\newblock A proof of {CSP} dichotomy conjecture.
\newblock In {\em 2017 IEEE 58th Annual Symposium on Foundations of Computer
  Science (FOCS)\/} (2017), IEEE, pp.~331--342.

\end{thebibliography}

\end{document}